\documentclass[ejs]{imsart}
\RequirePackage{amsthm,amsmath,amsfonts,amssymb}
\RequirePackage[authoryear]{natbib}
\RequirePackage[colorlinks,citecolor=blue,urlcolor=blue]{hyperref}
\RequirePackage{graphicx, subcaption, enumitem, caption, xcolor}
\RequirePackage{booktabs, comment} 

\startlocaldefs
\theoremstyle{plain}
\newtheorem{theorem}{Theorem}[section]
\newtheorem{proposition}[theorem]{Proposition}
\newtheorem{lemma}[theorem]{Lemma}
\newtheorem{corollary}[theorem]{Corollary}
\theoremstyle{definition}

\newtheorem{assumption}[theorem]{Assumption}
\newtheorem{example}{Example}

\theoremstyle{remark}

\newtheorem{remark}{Remark} 
\endlocaldefs

\begin{document}

\begin{frontmatter}
\title{An association measure for mixed-type variables} 
\runtitle{An association measure for mixed-type variables} 

\begin{aug}
\author[A]{\fnms{Yongjae}~\snm{Kim}\ead[label=e1]{cherrycooky@snu.ac.kr}},
\author[A,B,C]{\fnms{Haeun}~\snm{Moon}\ead[label=e2]{haeunmoon@snu.ac.kr}}
\and
\author[A,C]{\fnms{Sungkyu}~\snm{Jung}\ead[label=e3]{sungkyu@snu.ac.kr}}
\address[A]{Department of Statistics,
Seoul National University\printead[presep={,\ }]{e1,e3}}

\address[B]{School of Transdisciplinary Innovations, Seoul National University\printead[presep={,\ }]{e2}}

\address[C]{Institute for Data Innovation in Science,
Seoul National University}

\runauthor{Y. Kim et al.}
\end{aug}

\begin{abstract}
Quantifying the association between a real-valued variable and a categorical variable is a fundamental task in data analysis. Existing methods often rely on parametric assumptions or arbitrary integer encoding, which may lead to unstable results. 
We propose a label-invariant population measure of association, $\xi'$, specifically designed for the mixed real-valued--categorical setting. The proposed measure is normalized between 0 and 1; it equals 0 if and only if the variables are independent and 1 if and only if the categorical variable is a measurable function of the real-valued one. We also introduce a corresponding sample estimator, $\xi_n'$, computable in $O(n \log n)$ time. These measures are invariant to permutations of category labels and strictly monotone transformations of the real-valued variable. 
We establish the strong consistency and asymptotic normality of the estimator $\xi_n'$, enabling a computationally efficient, permutation-free Wald test for independence, and an asymptotic confidence interval for the population measure $\xi'$. Extensive simulations and an application to The Cancer Genome Atlas (TCGA) data demonstrate that the proposed method provides coding stability, competitive power, and substantial computational advantages in nominal mixed-type settings.
\end{abstract}



\end{frontmatter}

%
%

\section{Introduction}
\label{sec:intro-xi}

Measuring and testing association between two random variables is a central task across many studies. Variable pairs are commonly classified as numerical–numerical, numerical–categorical (mixed), or categorical–categorical. \linebreak 
Across these settings, researchers have developed coefficients to quantify relationships and hypothesis tests to assess independence. 
For numerical–numerical pairs, there is a long literature ranging from classical coefficients to modern statistics; see, for example, Pearson’s correlation, Spearman’s rho \citep{Spearman1904}, Kendall’s tau and some extensions \citep{Kendall1938, BergsmaDassios2014, MoonChen2022}, distance correlation \citep{SzekelyRizzoBakirov2007}, Hilbert Schmidt Independence Criterion (HSIC)\citep{Gretton2008}, the recently developed rank-based coefficient of \citet{Chatterjee2021}, and a dependence measure based on monotone rearrangements \citep{strothmann_rearranged_2024}.
For categorical–categorical pairs, many coefficients and tests of association are also well established; to name a few, Cramér’s \(V\) \citep{Cramer1946}, Goodman–Kruskal’s \(\tau\) \citep{GoodmanKruskal1954}, and the chi-squared test \citep{Pearson1900}. 

However, real-world data analysis often involves heterogeneous variable types. 
Indeed, mixed-type data have become increasingly prevalent in modern research.
For instance, in genomics, researchers routinely analyze dependencies between numerical gene expression measurements and categorical clinical subtypes. 
Similarly, in social sciences and survey data, numerical socioeconomic variables (e.g., income, age) are frequently studied alongside categorical demographic factors.
In contrast to the extensive literature for homogeneous pairs, general-purpose tools for the mixed-type case are fewer.

Several established tools address the numerical–categorical case. When $Y$ is dichotomous, the point–biserial correlation \citep{GlassHopkins1995}, which equals Pearson's correlation with the categorical $Y$ encoded as a binary variable, has been used.  
The polyserial correlation \citep{OlssonDrasgowDorans1982} targets linear association between a continuous \(X\) and an ordinal \(Y\) under a latent-normal model. 
The correlation ratio $\eta^2$ of \citet{Pearson1915PartialEta} measures the fraction of variance in the numerical variable explained by between-group differences. While the ANOVA $F$-test \citep{Fisher1921CropI} is the standard procedure for testing these mean differences, the resulting $F$-statistic is not, in itself, a dependence measure on a fixed scale. 
The polyserial correlation relies on linear or ordinal assumptions \citep{Drasgow1986}; the correlation ratio summarizes only mean effects and can be zero when other forms of dependence exist; the point-biserial correlation applies only to binary $Y$ \citep{Tate1954}; and the ANOVA $F$-statistic yields test statistics without a calibrated $0$–$1$ scale.

Due to these limitations, common workarounds in practice include encoding categories as integers so that methods for numerical variables can be applied, or discretizing numerical variables so that categorical methods can be used. These workarounds, however, may discard information and create sensitivity to binning, or impose an arbitrary order and spacing on $Y$. As a result, rank-based statistics, such as Chatterjee’s $\xi_n$, can fluctuate across different choices of codings. Figure~\ref{fig:codeAB} illustrates the issue with the first workaround (integer coding): applying two different encodings to $Y$ (Code 1 and Code 2) causes Chatterjee’s statistic to vary significantly between $\xi_n=0.971$ and $\xi_n=0.788$. By contrast, our proposed coding-invariant measure yields the stable value $\xi_n' = 0.903$.

\begin{figure}[tp]
  \centering
  \includegraphics[width=\linewidth]{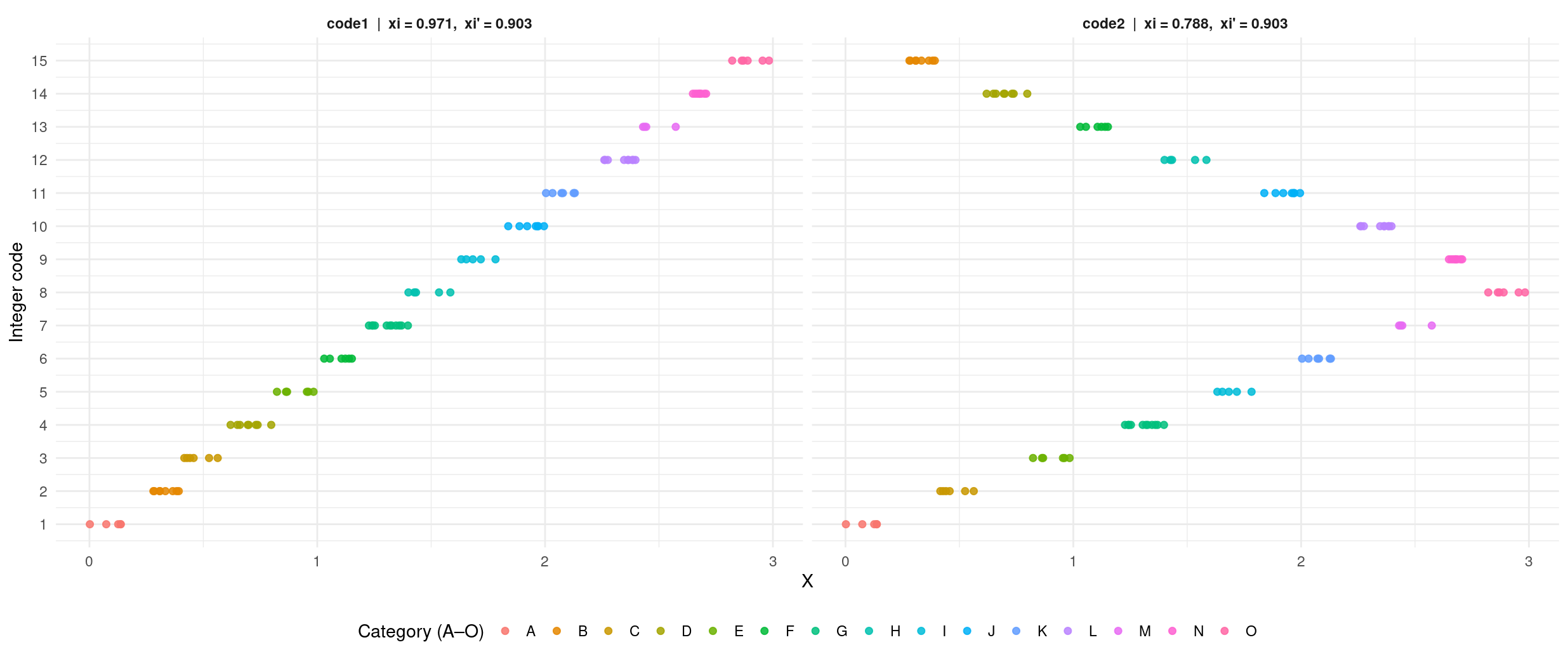}
  \caption{Illustration of the integer coding problem. Applying two different arbitrary integer encodings (Code 1 and Code 2) to the categorical variable $Y$ causes Chatterjee’s statistic $\xi_n$ to fluctuate significantly ($0.971$ vs. $0.788$), whereas the proposed statistic $\xi_n'$ remains invariant ($0.903$).}
  \label{fig:codeAB}
\end{figure}

\citet{Chatterjee2021} introduced the rank-based coefficient $\xi$ and its empirical version $\xi_n$ for real-valued random variables $X$ and $Y$. 
This coefficient takes values in $[0,1]$, equals $0$ if and only if $X$ and $Y$ are independent, and equals $1$ if and only if $Y$ is almost surely a measurable function of $X$. It is invariant to strictly monotone transformations, requires no smoothing or tuning, and is computable in $O(n\log n)$ time.
Because of these properties, Chatterjee's rank correlation has attracted substantial interest. \citet{LinHan2022} established asymptotic normality and consistent variance estimation under dependence for continuous random variables, whereas \citet{Kroll2026} established asymptotic normality in a general setting. Other developments include conditional dependence measures \citep{AzadkiaChatterjee2021}, power enhancements \citep{LinHan2023_boosting}, bootstrap theory and resampling inference \citep{lin_failure_2024,DetteKroll2025}, and related distribution-free measures for detecting dependence \citep{YangAzadkiaWang2025}.

Chatterjee's coefficient is defined for any numeric $X$ and $Y$, including continuous, discrete, and mixed distributions, thus it is well-suited for ordinal variables. For nominal categorical responses, however, applying Chatterjee's coefficient requires assigning integers to the category labels. As demonstrated in Figure~\ref{fig:codeAB}, different encodings of the same nominal categories may lead to different values. Our goal is therefore to develop a label-invariant coefficient tailored to nominal categorical responses, by aggregating class indicators directly rather than imposing an artificial order on the labels.

Motivated by Chatterjee's idea, we introduce a label-invariant association measure for the numerical–categorical setting that satisfies Rényi’s criteria \citep{Renyi1959}: The proposed measure $\xi'$ is normalized in $[0,1]$, equals 0 only under independence, and equals 1 only under almost sure functional dependence.
The corresponding sample measure $\xi_n'$ of association quantifies the frequency of adjacent observations, in terms of the $X$ variable, having the same category. The measure $\xi_n'$ is then scaled to be an estimator of the population measure $\xi'$.

Our construction of the sample measure $\xi_n'$ is nonparametric and model-free, invariant to permutations of the labels of \(Y\) and to strictly monotone transformations of \(X\), and computable in \(O(n\log n)\) time.
Furthermore, we establish the asymptotic normality of \(\xi_n'\) under both independence and general dependence. This asymptotic theory not only yields a simple permutation-free Wald test for independence but also enables the construction of confidence intervals for the population coefficient $\xi'$ using a consistent variance estimator.

The remainder of this paper is organized as follows. 
Section 2 formalizes the population measure $\xi'(X,Y)$ and its sample estimator $\xi_n'(X,Y)$. We discuss its theoretical motivation, establish fundamental properties, including the characterization of independence and functional dependence, and explore its connections to the coefficient of determination ($R^2$), classical runs statistics \citep{BartonDavid1957}, and Chatterjee's rank correlation.
Section 3 investigates the asymptotic theory. We prove the consistency of $\xi_n'$ and derive its limiting distributions under both independence and general dependence. These results are then used in the construction of a permutation-free Wald test and asymptotically valid confidence intervals.
Section 4 presents comprehensive simulation studies, evaluating the finite-sample performance of $\xi_n'$ in terms of coding invariance, calibration, and power of the proposed independence test relative to established alternatives.
Section 5 demonstrates the practical utility of the proposed measure using genomic data from the TCGA breast cancer study, highlighting its ability to detect general dependencies such as variance heterogeneity.
Section 6 provides a concluding discussion.

\section{A coefficient for real-valued--categorical association}
\label{sec:xiprime}

\subsection{Definition}

Let \(X\) be a real-valued random variable and \(Y\) be a categorical random variable with \(k \geq 2\) levels (so that \(Y\) is not almost surely constant). Throughout the paper, write
\[
    p_j:=P(Y=j),
    \qquad
    B:=\sum_{j=1}^k p_j^2,
    \qquad
    \rho:=\sum_{j=1}^k p_j^3.
\]
For each \(j=1,\ldots,k\), choose a Borel-measurable function
\(g_j:\mathbb R\to[0,1]\) such that
\[
    g_j(X)
    =
    P\!\left(Y=j\mid\sigma(X)\right)
    \qquad\text{almost surely}.
\]
These functions may be chosen jointly so that
$\sum_{j=1}^k g_j(x)=1$ for every $x\in\mathbb R$. Fix such versions $(g_1,\ldots,g_k)$ throughout the paper. 
Set
\[
\begin{aligned}
    &h(x):=\sum_{j=1}^k g_j(x)^2,
    \qquad
    \tau(x):=\sum_{j=1}^k g_j(x)^3,
    \\
    &m(x):=\sum_{j=1}^k p_jg_j(x),
    \qquad
    A:=E[h(X)].
\end{aligned}
\]
With these notations, the population association measure is
\begin{align}
\xi^{\prime}(X,Y)
  &=\frac{\displaystyle
     E\!\Bigl[\sum_{j=1}^k P(Y=j\mid X)^2\Bigr]
     -\sum_{j=1}^k P(Y=j)^2}
     {\displaystyle 1-\sum_{j=1}^k P(Y=j)^2} \nonumber\\
  &=\frac{A-B}{1-B}.
\label{eq:xitrue}
\end{align}
Given i.i.d.\ samples \((X_1,Y_1),\ldots,(X_n,Y_n)\) from \((X,Y)\), reorder the samples based on the order of $X$ as $(X_{(1)},Y_{(1)}),\ldots,(X_{(n)},Y_{(n)})$ with $X_{(1)}\le\cdots\le X_{(n)}$, where ties in $X$ are broken independently at random. Define
\[
    A_n:=\frac{1}{n-1}\sum_{i=1}^{n-1}\mathbf 1\{Y_{(i+1)}=Y_{(i)}\},
    \qquad
    \hat p_j:=\frac1n\sum_{i=1}^n \mathbf 1\{Y_i=j\},
    \qquad
    B_n:=\sum_{j=1}^k\hat p_j^2.
\]
The sample estimator is
\begin{align}
\xi_n^{\prime}(X,Y)
  &=\frac{A_n-B_n}{1-B_n}.
\label{eq:xisample}
\end{align}
We define $\xi_n' = 0$ if the denominator is zero.

The two terms $A$ and $B$ in the numerator of Eq.~\eqref{eq:xitrue} quantify the concentration of the conditional distribution of $Y$ given $X$ and that of the marginal distribution of $Y$, respectively. This measure of concentration, the sum of squared probability masses, is also called the Herfindahl--Hirschman Index \citep{HallTideman1967}.
The first term, $A$, has the maximum value of $1$, which is achieved if and only if $Y$ is completely determined by $X$ (i.e., $Y=f(X)$ almost surely). 
Conversely, its minimum value is $B$, which is attained when $X$ and $Y$ are independent. In the latter case, the concentration of the conditional distribution is identical to that of the marginal distribution.
Consequently, the difference $A-B$ reflects the relative increase in concentration of the conditional distribution compared to the marginal distribution. 
Equivalently, one may view the numerator as the reduction in Gini impurity $1-B$ \citep{Gini1912}, which is a measure widely used in statistical learning \citep{Breiman1984}. In other words, it quantifies the information gain about $Y$ obtained by observing $X$. Therefore, this difference serves as a valid basis for an association measure. 
To normalize the measure to the range $[0,1]$, we divide the difference by its maximum achievable value, $1-B$.

As a toy example, we illustrate the behavior of the proposed population measure with a three-class response variable $Y \in \{1, 2, 3\}$ and a continuous predictor $X \in [0, 4]$. Data were generated using a multinomial logit model where a parameter $\gamma$ controls the strength of dependence (see Appendix~\ref{app:sim_details} for the detailed data generation process).
We examined two scenarios: a weak dependency setting resulting in significant class overlap, and a strong dependency setting yielding nearly deterministic dependence.
In the weak dependency scenario, the conditional distribution of $Y$ given $X=x$ remains relatively dispersed (top row, Figure~\ref{fig:concentration_gain2}). Consequently, the conditional concentration curve $x \mapsto \sum_{j=1}^3 P(Y=j \mid X=x)^2$ rises only slightly above the marginal concentration baseline $B$. The difference of the two terms, corresponding to the numerator of $\xi'$, is minimal (0.145) in this case, and yields a low coefficient value of $\xi' \approx 0.22$.
In contrast, under the strong dependency, $Y$ is nearly determined by the value $X=x$ (bottom row, Figure~\ref{fig:concentration_gain2}). Here, the conditional concentration at each $x$ approaches its maximum value of $1$, creating a substantial gap from the marginal baseline. 
This large blue shaded area ($0.631$) visually demonstrates how the proposed measure quantifies the increase in concentration achieved by conditioning on $X=x$, resulting in a high coefficient value of $\xi' \approx 0.96$.

\begin{figure}[tp!]
\centering
\includegraphics[width=0.95\linewidth]{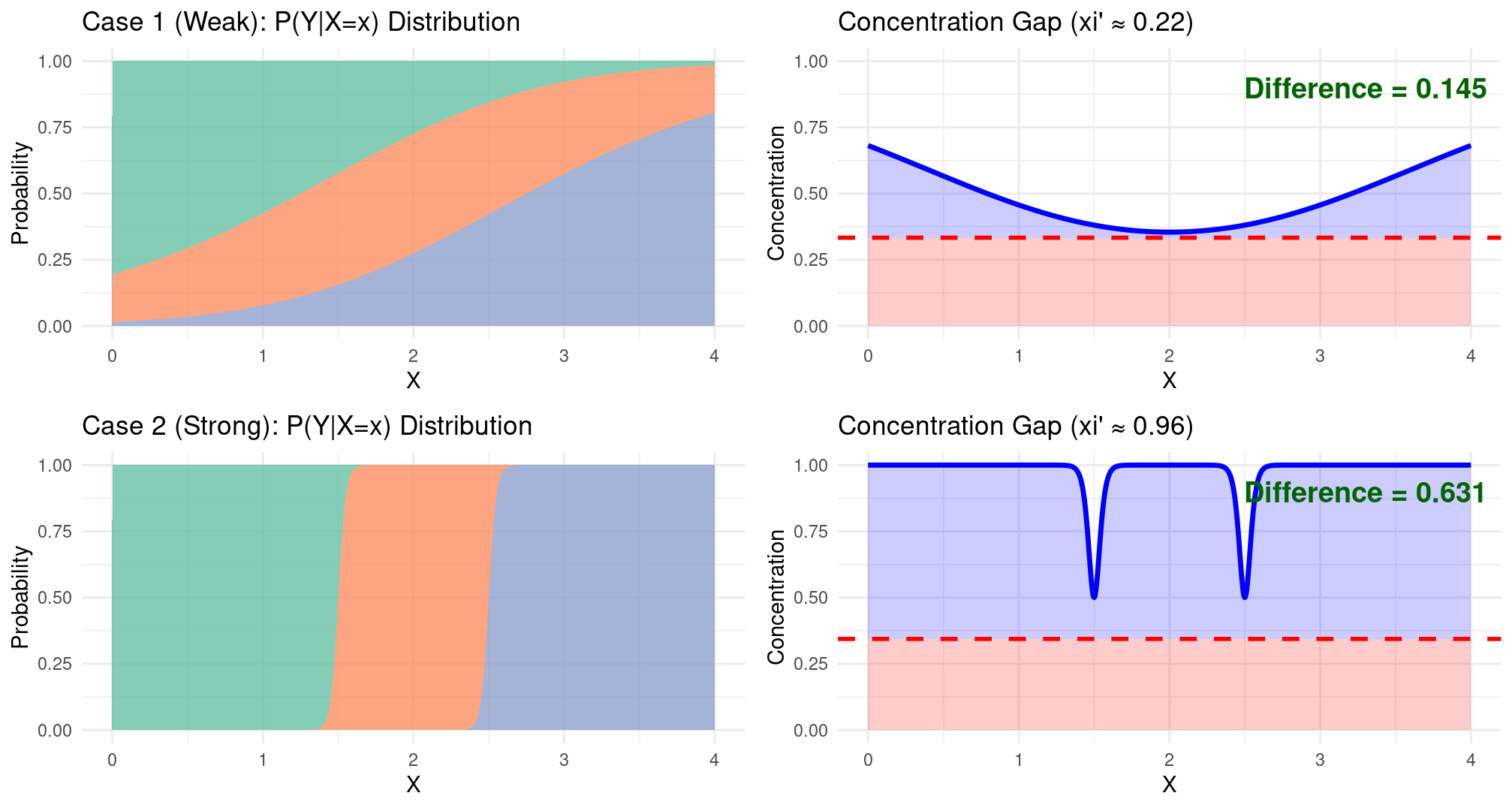}
\caption{The left panels display the conditional class probabilities $P(Y=j \mid X=x)$ for $j = 1,2,3$. The right panels compare the conditional concentration, $\sum_{j} P(Y=j \mid X=x)^2$ (solid blue line), with the marginal concentration baseline, $B$ (dashed red line). The blue shaded area represents the difference between these two terms, quantifying the information about $Y$ gained by observing $X=x$. The normalized coefficient values are $\xi' \approx 0.22$ (weak dependence) and $\xi' \approx 0.96$ (strong dependence).}
\label{fig:concentration_gain2}
\end{figure}

The rationale behind the estimator $\xi_n'$, defined in Eq.~\eqref{eq:xisample}, is rooted in the concept of conditional probability of coincidence.
Consider a fixed value $X=x$. Let $Y$ and $Y'$ be two independent random variables drawn from the conditional distribution of $Y$ given $X=x$. The probability that these two independent replicates coincide is given by
\[
\begin{aligned}
P(Y = Y' \mid X=x)
&= \sum_{j=1}^k P(Y=j \mid X=x)\,P(Y'=j \mid X=x) \\
&= \sum_{j=1}^k P(Y=j \mid X=x)^2.
\end{aligned}
\]
The population measure $\xi'(X,Y)$ essentially aggregates these local coincidences over the distribution of $X$, via the expectation $E[\sum_{j=1}^k P(Y=j \mid X)^2]$.
A practical challenge arises because the data provide only one response \(Y_i\) for each \(X_i\). To approximate replicates, we reorder the sample such that \(X_{(1)}\le\cdots\le X_{(n)}\), using independent random tie-breaking when ties occur. For each adjacent pair, we use \(\mathbf{1}\{Y_{(i)}=Y_{(i+1)}\}\) as a surrogate for the coincidence indicator of two conditional replicates. Averaging these indicators gives the adjacent-match term $A_n$.

The validity of this construction rests on the conditional class probabilities at adjacent observations being close. This would follow directly if each \(g_j\) were continuous. However, global continuity need not be assumed.
In the proof of strong consistency, we represent the randomized ordering using i.i.d.\ uniform variables on \((0,1)\). Lusin's theorem \citep{Folland1999} then gives a compact set \(C\subset(0,1)\), whose complement has arbitrarily small Lebesgue measure, on which the conditional class probability functions expressed in terms of these uniform variables are uniformly continuous. 
Since adjacent uniform order statistics become close as the sample size increases, their conditional class probabilities also become close on this compact set, while pairs involving its complement make an arbitrarily small contribution to the average. This argument is formalized in Theorem~\ref{thm:consistency-xin}, with the full proof provided in Appendix~\ref{app:proof_consistency}.

The sample estimator \(\xi_n'(X,Y)\) then compares the conditional and marginal coincidence probabilities. 
The adjacent-match term \(A_n\) estimates
\newline
\(
  A =  E[\sum_{j=1}^k P(Y=j\mid X)^2],
\)
whereas \(B_n\) estimates the marginal coincidence probability \(B=\sum_{j=1}^k P(Y=j)^2\), which is the corresponding value under independence.
Subtracting \(B_n\) removes this marginal baseline, and dividing by \(1-B_n\) normalizes the statistic so that its maximum possible value is \(1\).

\begin{remark}
The structure of our proposed measure $\xi'$ is motivated by the recently proposed regression-based approach to measuring dependence between real-valued random variables, commonly denoted by $\xi(X,Y)$. The measure of association $\xi(X,Y)$  was originally proposed by \citet{DetteSiburgStoimenov2013} and later popularized by \citet{Chatterjee2021}. It quantifies dependence through the integrated variance of conditional expectations:
\begin{align}
\xi(X,Y) = \frac{\displaystyle\int \mathrm{Var}\bigl(E[\mathbf{1}\{Y\ge t\}\mid X]\bigr) \, dF_Y(t)}
{\displaystyle\int \mathrm{Var}(\mathbf{1}\{Y\ge t\}) \, dF_Y(t)}.
\label{eq:Chatterjee}
\end{align}
This population quantity, or equivalent formulations thereof, appears in the literature under various names, including the Dette-Siburg-Stoimenov dependence measure, Chatterjee's coefficient, Chatterjee’s rank correlation, or more broadly as a measure of regression dependence \citep{AzadkiaChatterjee2021, LinHan2022, lin_failure_2024}.
Analogously, our proposed measure $\xi'$ can be expressed as a ratio of aggregated variances:
\begin{align}
\xi'(X,Y) = \frac{\displaystyle \sum_{j=1}^k \mathrm{Var}\bigl(P(Y=j \mid X)\bigr)}
{\displaystyle \sum_{j=1}^k \mathrm{Var}(\mathbf{1}\{Y=j\})}.
\label{eq:repofxip}
\end{align}
This representation highlights the structural similarity between the two measures in capturing regression dependence. 
A comprehensive discussion on the connection to Chatterjee's coefficient $\xi$ is provided in Section~\ref{subsubsec:conn_chat}.
\end{remark}

\subsection{Basic properties}

We establish fundamental theoretical properties of the proposed association measure $\xi'$, demonstrating its validity as a measure of dependence.

\begin{proposition}[Basic properties of $\xi'$]
\par\vspace{-0.3ex}\noindent
\begin{enumerate}
\item The range of $\xi'(X,Y)$ is $[0,1]$.
\item $\xi'(X,Y)=0$ if and only if $X$ and $Y$ are independent.
\item $\xi'(X,Y)=1$ if and only if there exists a measurable function $f:\mathbb R\to\{1,\dots,k\}$ such that $Y=f(X)$ almost surely.
\item $\xi'(X,Y)$ is invariant under strictly monotone transformations of $X$. That is, for any strictly monotone function $\varphi:\mathbb R\to\mathbb R$, we have $\xi'(\varphi(X),Y) = \xi'(X,Y)$.
\item $\xi'(X,Y)$ is invariant under permutations of the category labels. That is, for any bijection $\pi:\{1,\dots,k\}\to \{1,\dots,k\}$,
\[
\xi'(X,\pi(Y)) = \xi'(X,Y).
\]
The same invariance holds for the sample coefficient $\xi_n'$.
\item If $\;Y$ is a binary random variable taking values in $\{0,1\}$, then $\xi'$ coincides with Chatterjee’s coefficient $\xi$.
\end{enumerate}
\label{prop:basic-xi-prime}
\end{proposition}
The detailed proof is provided in Appendix~\ref{subsec:proof-basic-properties}.

Proposition~\ref{prop:basic-xi-prime} confirms that $\xi'$ satisfies the standard requirements for a directed measure of dependence in the sense of \citet{Renyi1959}. Specifically, Properties 1--3 show that the measure is properly normalized and correctly identifies the two extreme cases of association: independence and deterministic functional dependence of $Y$ on $X$. As with Chatterjee's coefficient, $\xi'$ is asymmetric; it measures the predictive information contained in $X$ about $Y$, rather than a symmetric notion of association between the two variables.

Property 4 establishes invariance under strictly monotone transformations of $X$, reflecting the rank-based nature of the construction, since only the ordering of the observations according to $X$ is used. The same invariance also holds for the sample estimator $\xi_n'$.

Property 5 is the key distinction between the proposed coefficient and a direct application of Chatterjee's coefficient to a numerically encoded categorical response. For nominal categorical variables, the labels have no intrinsic order, and therefore a valid measure should not depend on how the categories are encoded or permuted. The proposed coefficient satisfies this label-permutation invariance both at the population level and at the sample level.

Finally, Property 6 shows that, in the binary case, the proposed population coefficient agrees with Chatterjee's population coefficient. For $k\ge 3$, however, $\xi'$ should not be viewed as a direct extension of Chatterjee's coefficient. We also note that, even in the binary case, the sample statistics $\xi_n$ and $\xi_n'$ do not coincide exactly.

\subsection{Connections to classical measures}
\label{subsec:connections_classical}

The proposed measure $\xi'$ admits intuitive interpretations that link it to classical statistical concepts. In this section, we first establish its connection to regression dependence, demonstrating that $\xi'$ generalizes the coefficient of determination ($R^2$) via Gini impurity reduction. Subsequently, we relate it to Chatterjee's coefficient $\xi$ and compare weighted and unweighted ways of aggregating class-specific contributions.

\subsubsection{Connection to \texorpdfstring{$R^2$}{R-squared} and Gini impurity}
\label{subsubsec:r2_gini}
First, when both $X$ and $Y$ are binary with values in $\{0,1\}$, $\xi'(X,Y)$ reduces exactly to the coefficient of determination $R^2$, from a linear regression of $Y$ on $X$:
\[
\xi'(X,Y) = \frac{\mathrm{Var}(E[Y\mid X])}{\mathrm{Var}(Y)} = \frac{\mathrm{Var}(Y) - E[\mathrm{Var}(Y|X)]}{\mathrm{Var}(Y)}.
\]
This identity demonstrates that $\xi'$ quantifies the proportion of variance in $Y$ explained by $X$, matching the principle of regression dependence in \citet{DetteSiburgStoimenov2013}.

For the general case where $X$ is a real-valued random variable and $Y$ is a categorical random variable with $k$ levels, $\xi'$ generalizes this concept by replacing the variance with Gini impurity, a standard measure of dispersion for categorical variables. Specifically,
\[
\xi'(X,Y) = \frac{G(Y) - E[G(Y\mid X)]}{G(Y)},
\]
where $G(Y):=1-B$ is the marginal Gini impurity, and $G(Y\mid X=x):=1-h(x)$ is the conditional Gini impurity \citep{Gini1912}.
Thus, $\xi'$ represents the relative reduction in Gini impurity, or equivalently, the information gain achieved by observing $X$. A verification of these equivalences is provided in Appendix~\ref{subsec:r2-gini-proof}.

\subsubsection{Connection to Chatterjee's coefficient \texorpdfstring{$\xi$}{xi}}
\label{subsubsec:conn_chat}

The proposed association measure $\xi'$, defined in Eq.~\eqref{eq:xitrue}, is closely related to Chatterjee's $\xi$, defined in Eq.~\eqref{eq:Chatterjee}, but it is not a straightforward extension to a numerically encoded categorical response.
In Chatterjee's measure dependence is quantified through the variance of the conditional tail probability
\[
\mathrm{Var}\bigl(P(Y\ge t\mid X)\bigr)
=
\mathrm{Var}\!\left(E[\mathbf{1}\{Y\ge t\}\mid X]\right).
\]
For a categorical response, replacing the cumulative indicators $\mathbf{1}\{Y\ge t\}$ with the class indicators $\mathbf{1}\{Y=j\}$ yields the class-specific terms
\[
\mathrm{Var}\bigl(P(Y=j\mid X)\bigr),
\qquad j=1,\ldots,k.
\]
The crucial step lies in how these class specific terms are aggregated. In Chatterjee's coefficient, the variance of the conditional tail probability is integrated over $t$ with respect to $dF_Y(t)$, so response values with higher marginal probabilities receive greater weight. For a nominal response, either retaining or omitting the analogous prevalence weight may be appropriate, depending on the role assigned to class prevalence. This leads to the following two strategies.

\begin{itemize}
    \item Candidate 1: \textbf{Weighted Aggregation} ($\xi''$). A direct analogue of the integral with respect to the distribution function $F_Y(t)$ is the weighting of each conditional variance term by the marginal probability $P(Y=j)$. Noting that $E[\mathbf{1}\{Y=j\}\mid X] = P(Y=j \mid X)$, this yields a weighted coefficient $\xi''$, defined as:
\[
\xi''(X,Y) = \frac{\displaystyle \sum_{j=1}^k P(Y=j) \cdot \mathrm{Var}\bigl(P(Y=j \mid X)\bigr)}
{\displaystyle \sum_{j=1}^k P(Y=j) \cdot \mathrm{Var}(\mathbf{1}\{Y=j\})}.
\]
\item Candidate 2: \textbf{Unweighted Aggregation} ($\xi'$).
Alternatively, $\xi'$ aggregates the class-specific variance contributions without an additional prevalence factor and therefore treats category labels symmetrically:
\[
\xi'(X,Y) = \frac{\displaystyle \sum_{j=1}^k \mathrm{Var}\bigl(P(Y=j \mid X)\bigr)}
{\displaystyle \sum_{j=1}^k \mathrm{Var}(\mathbf{1}\{Y=j\})}.
\]

\end{itemize}

We observe that $\xi' = \xi''$ for $k=2$. The two population coefficients also coincide whenever the class probabilities are uniform and yield similar values when the distribution is approximately uniform. Under class imbalance, however, the effect of prevalence weighting depends on which classes carry the association. The following two examples focus on minority-class signals, for which the two aggregation schemes can differ substantially.

\begin{example}[Rare signal setting]
\label{ex:rare_perfect}
Fix $k\ge3$, and let $\epsilon\in(0,1/k)$. Suppose that the marginal distribution of $Y$ is given by
\[
p_1 = P(Y=1) = \epsilon,
\qquad
p_j = P(Y=j) = \frac{1-\epsilon}{k-1},\quad j=2,\dots,k.
\]
Let $S \subset \mathbb{R}$ be a measurable set satisfying $P(X \in S) = \epsilon$. We define the conditional probability functions as
\begin{align*} 
g_1(x) &:= P(Y=1 \mid X=x) = \mathbf{1}_S(x), \\ 
g_j(x) &:= P(Y=j \mid X=x) = \frac{1}{k-1}\,\mathbf{1}_{S^c}(x),\ \quad j\ge 2. 
\end{align*}
Under this model, $X$ deterministically identifies the rare class $Y=1$ (i.e., $x\in S \Rightarrow Y=1$), while providing no discriminatory power among the remaining $k-1$ classes. See Figure~\ref{fig:example_k4}(a) for an illustration of the conditional probability $g_j(x)$'s (for the case $k = 4$). 

\begin{figure}[tp]
    \centering
    \begin{subfigure}[b]{\linewidth}
        \centering
        \includegraphics[width=\linewidth]{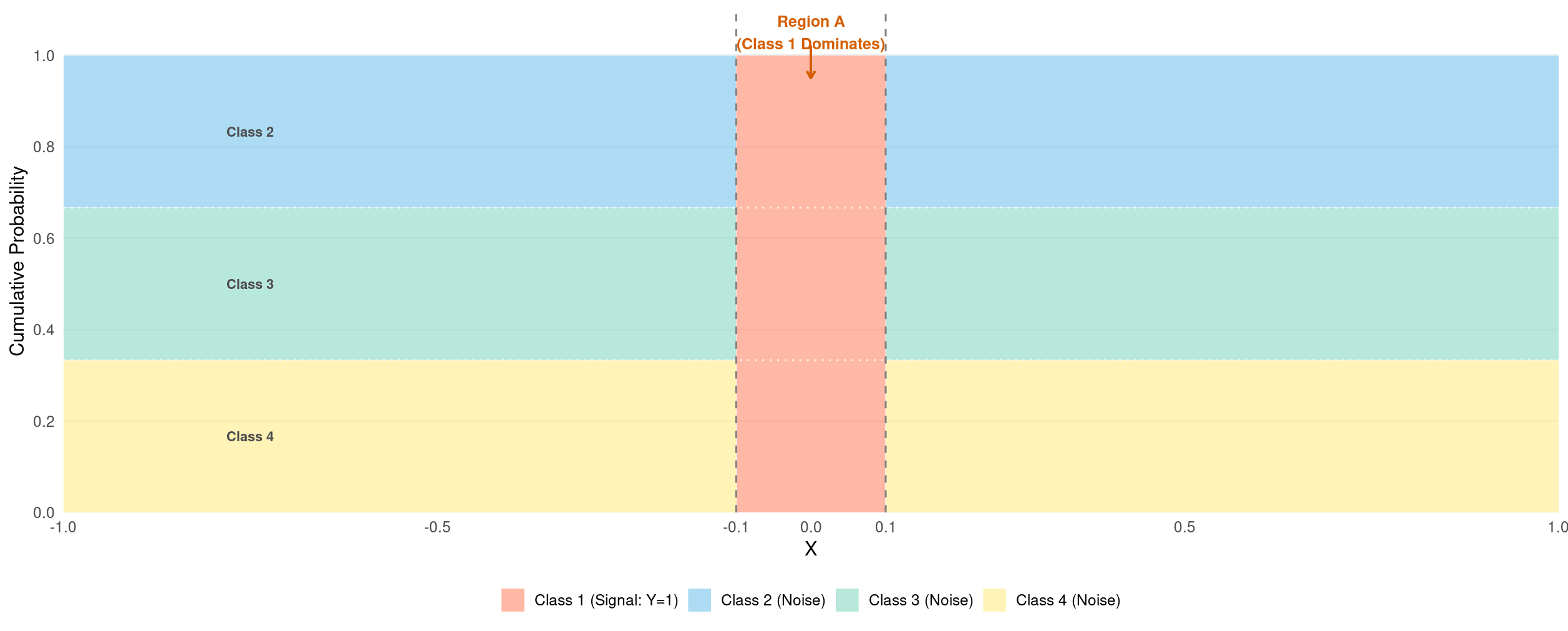}
        \caption{Locally perfect prediction}
        \label{fig:scenario_A}
    \end{subfigure}
    
    \vspace{1em} 
    
    \begin{subfigure}[b]{\linewidth}
        \centering
        \includegraphics[width=\linewidth]{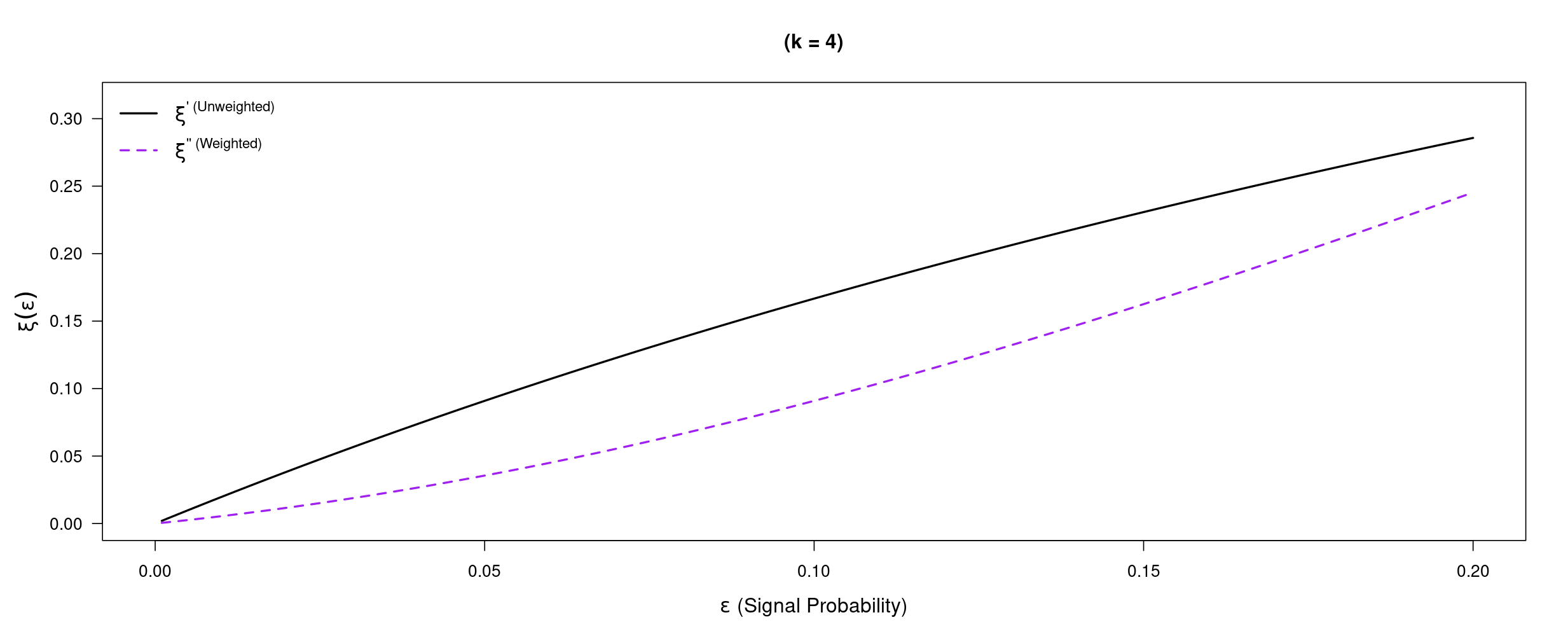}
        \caption{Comparison between $\xi'$ and $\xi''$ ($k=4$)}
        \label{fig:comparison_k4}
    \end{subfigure}
    
    \caption{Illustration of the rare signal setting described in Example~\ref{ex:rare_perfect}.}
    \label{fig:example_k4}
\end{figure}

A direct computation yields
\begin{align*}
\xi'(\epsilon,k) &= \frac{\epsilon k}{k(1+\epsilon)-2},
\\
\xi''(\epsilon,k) &=
\frac{\displaystyle \epsilon^2 + \frac{\epsilon(1-\epsilon)}{(k-1)^2}}
     {\displaystyle \epsilon^2 + \frac{(1-\epsilon)(k-2+\epsilon)}{(k-1)^2}}.
\end{align*}
Figure \ref{fig:example_k4}(b) compares $\xi'(\epsilon,k)$ with $\xi''(\epsilon,k)$ for $k = 4$, for a range of $\epsilon$ values. As $\epsilon \to 0$, both coefficients vanish, since the signal is restricted to a subset with vanishing probability. However, their limiting ratio is
\[
\lim_{\epsilon \to 0}\frac{\xi'(\epsilon,k)}{\xi''(\epsilon,k)} = k.
\]
Thus, although both coefficients vanish as the signal region becomes rare, $\xi'$ is asymptotically $k$ times as large as $\xi''$. Additional comparisons for $k=3$ and $k=10$ are provided in Appendix~\ref{app:rare_case_plots}.
\end{example}

\begin{example}[Dominant noise setting]
\label{ex:dominant_noise}
To compare the population sensitivity of the two coefficients to rare signals, we consider a dominant noise setting for general $k \ge 3$. Specifically, let $\epsilon \in (0, \frac{1}{k-1})$ denote the marginal probability of each signal class (i.e., rare signals), and let $\delta$ be the probability measure of each signal interval $I_j$ (i.e., $P(X \in I_j) = \delta$). We assume the intervals $I_1, \dots, I_{k-1}$ are disjoint and satisfy $0 < \delta < \frac{1}{k-1}$. 
The marginal probabilities of $Y$ are given by:
\begin{align*}
p_j  = \epsilon, \quad j=1,\dots,k-1, \qquad
p_k = 1 - (k-1)\epsilon.
\end{align*}
Here, the $k$-th class represents the dominant noise background, which absorbs the majority of the probability mass as $\epsilon \to 0$.
The conditional probability functions $g_l(x) = P(Y=l \mid X=x)$ are defined explicitly:
\begin{align}
g_k(x) &= 1 - (k-1)\epsilon, \quad \text{for all } x, \label{eq:gk_def} \\
g_j(x) &= 
\begin{cases} 
(k-1)\epsilon & \text{if } x \in I_j, \\
\epsilon & \text{if } x \notin \bigcup_{m=1}^{k-1} I_m, \\
0 & \text{otherwise},
\end{cases} \quad j=1,\dots,k-1. \label{eq:gj_def}
\end{align}
Eq.~\eqref{eq:gk_def} ensures that the dominant class provides no predictive information, while Eq.~\eqref{eq:gj_def} shows that each signal class concentrates its probability mass within a specific interval $I_j$ to form a local signal. 
Under this setting, the explained variances for each class are derived as:
\begin{align*}
v_k &:= \mathrm{Var}(g_k(X)) = 0, \\
v_j &:= \mathrm{Var}(g_j(X)) = \delta(k-1)(k-2)\epsilon^2, \quad j=1,\dots,k-1.
\end{align*}
By substituting the derived variance terms into the definitions of $\xi'$ and $\xi''$, we obtain their exact closed-form expressions:
\begin{align}
\xi'(\epsilon, \delta, k) &= \frac{\delta(k-1)(k-2)\epsilon}{2 - k\epsilon}, \label{eq:xi_prime_exact} \\
\xi''(\epsilon, \delta, k) &= \frac{\delta(k-1)(k-2)\epsilon^2}{\epsilon(1-\epsilon) + \left[1-(k-1)\epsilon\right]^2}. \label{eq:xi_double_exact}
\end{align}
See Figure~\ref{fig:dominant_noise} for an illustration of the setting for the case $k=4$.

\begin{figure}[tp]
    \centering
    \begin{subfigure}[b]{\linewidth}
        \centering
        \includegraphics[width=\linewidth]{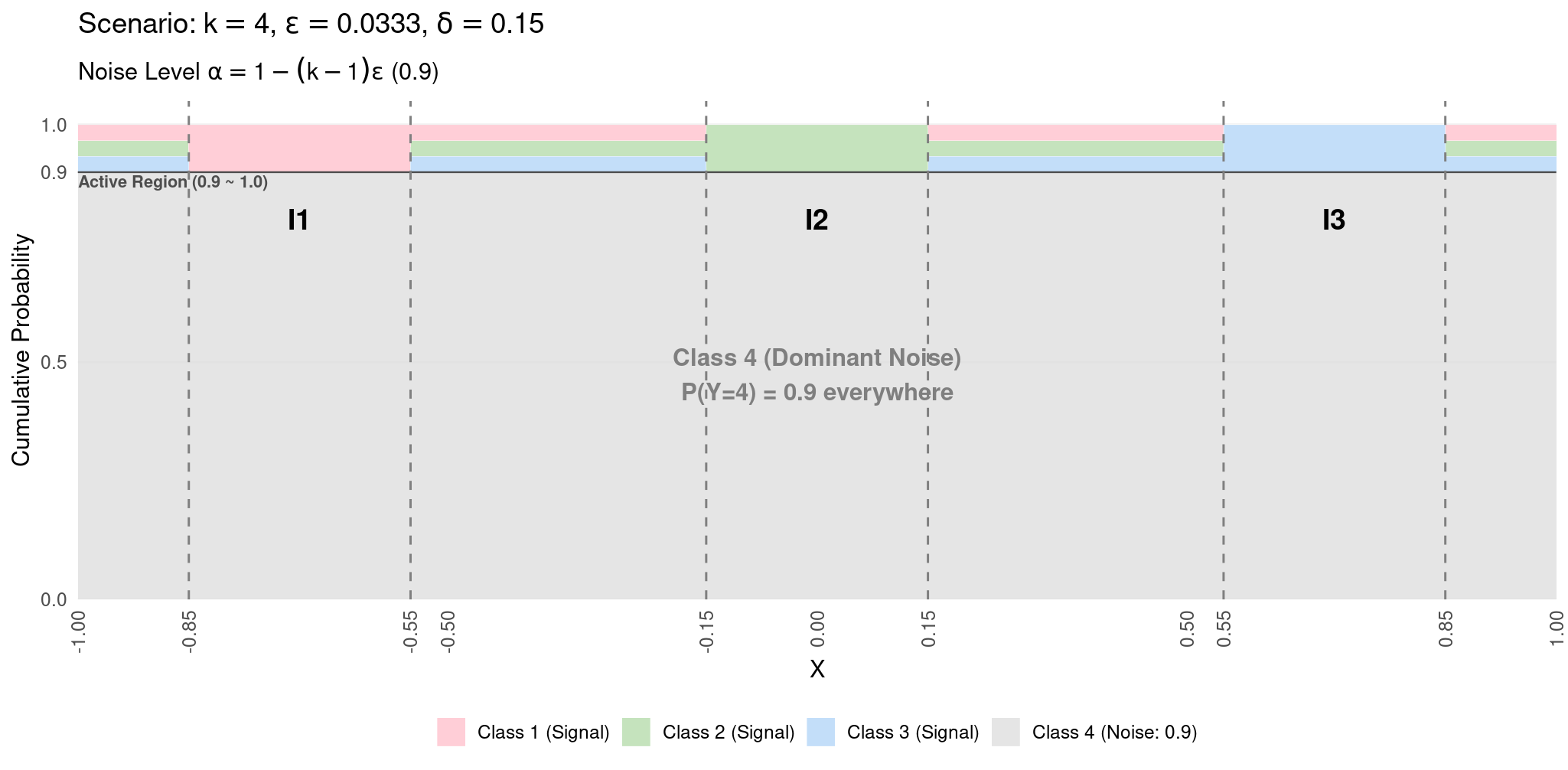}
        \caption{Pattern discovery in noise.}
        \label{fig:scenario_B}
    \end{subfigure}
    
    \vspace{1em} 
    
    \begin{subfigure}[b]{\linewidth}
        \centering
        \includegraphics[width=\linewidth]{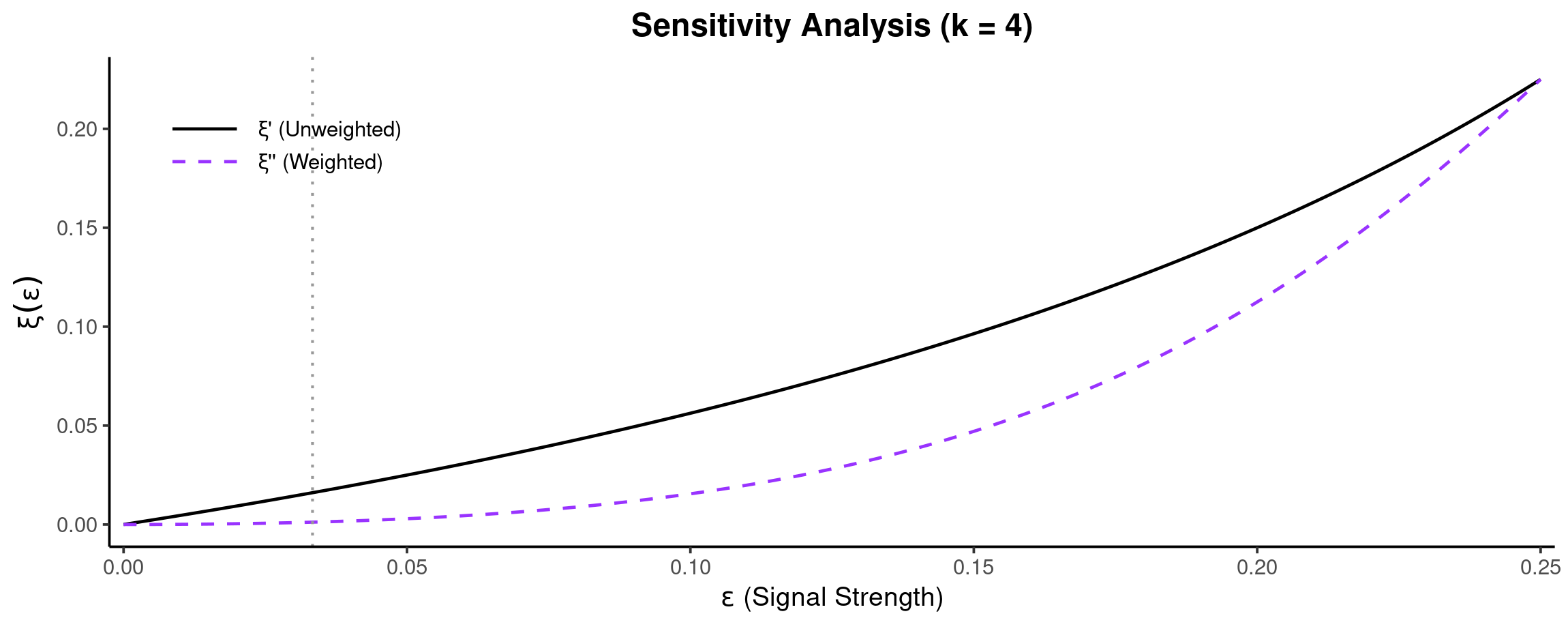}
        \caption{Population-sensitivity comparison. Comparison of $\xi'$ (solid black) and $\xi''$ (dashed purple) as a function of signal strength $\epsilon$. In this rare-signal design, $\xi'$ is of linear order in $\epsilon$, whereas $\xi''$ is of quadratic order.}
        \label{fig:detectability}
    \end{subfigure}
    \caption{Illustration of the dominant noise setting described in Example~\ref{ex:dominant_noise}.}
    \label{fig:dominant_noise}
\end{figure}

As $\epsilon \to 0$, the denominators in Eqs.~\eqref{eq:xi_prime_exact} and~\eqref{eq:xi_double_exact} converge to $2$ and $1$, respectively. Hence, for fixed $\delta$ and $k$, $\xi'$ vanishes at a linear rate as $\epsilon\to0$, whereas $\xi''$ vanishes at a quadratic rate.
\end{example}

At the population level, the examples above show that $\xi'$ can be much larger than $\xi''$ for minority class signals. Additional population comparisons are provided in Appendix~\ref{app:rare_case_plots}, while finite-sample permutation experiments in Appendices~\ref{app:dominant_noise_sim} and~\ref{app:signal_location_sim} show that the test based on $\xi_n'$ can have much higher power in this regime while remaining competitive for majority class signals. Together with the equality of the population targets under uniform class probabilities, these results from the settings examined motivate our choice of $\xi'$ as the primary proposal.

\section{Asymptotic properties of \texorpdfstring{$\xi_n'$}{xi'}}
In this section, we investigate the fundamental asymptotic properties of the estimator $\xi_n'$. We begin by establishing a theoretical connection to the classical runs statistic, which provides intuitive insight into the estimator's behavior. Subsequently, we prove the strong consistency of $\xi_n'$, and derive its asymptotic limiting distributions under both the independence and general dependence structures, thereby enabling formal statistical inference.

\subsection{Connection to the runs statistic}
\label{subsubsec:runs-xiprime}
While the population measure $\xi'$ is interpreted through variance or Gini reduction, the sample estimator $\xi_n'$ offers a complementary perspective through the lens of the sorted sequence $(Y_{(1)},\dots,Y_{(n)})$. 

In this empirical view, we establish that $\xi_n'$ is functionally equivalent to the classical runs statistic studied by \citet{BartonDavid1957}. Specifically, let $R_n$ denote the total number of runs, defined as maximal consecutive blocks of identical labels, in the sorted sequence of $n$ labels. Then, $\xi_n'$ can be expressed explicitly as an affine transformation of $R_n$:
\begin{align}
\label{eq:xi-runs-relation}
\xi_n' = 1 - \frac{R_n-1}{(n-1)(1-B_n)},
\end{align}
where $B_n$ is the sum of squared sample proportions, that is, the empirical marginal coincidence probability.

This relationship is derived from the observation that a new run begins exactly when the label changes between consecutive positions. The total number of label changes is $R_n-1$. Recalling that $A_n$ is the adjacent-match average, we have the identity:
\[
1 - A_n = \frac{1}{n-1}\sum_{i=1}^{n-1} \mathbf{1}\{Y_{(i+1)}\neq Y_{(i)}\} = \frac{R_n-1}{n-1}.
\]
Substituting $A_n = 1 - \frac{R_n-1}{n-1}$ into the definition $\xi_n'=(A_n-B_n)/(1-B_n)$ yields Eq.~\eqref{eq:xi-runs-relation}.

Eq.~\eqref{eq:xi-runs-relation} provides a clear intuition: $\xi_n'$ is maximized when the number of runs $R_n$ is minimized. This connection also elucidates the limiting behavior of $\xi_n'$.
Under independence between $X$ and $Y$, the sequence of labels behaves essentially as a random permutation. Classical theory on runs indicates that, in this regime, the number of runs grows linearly with the sample size (i.e., $R_n \approx n(1-B_n)$). Consequently, the fraction in Eq.~\eqref{eq:xi-runs-relation} converges to $1$, driving $\xi_n'$ towards $0$.

Conversely, the presence of dependence implies that adjacent samples in the sorted sequence are likely to share the same label. Whenever the dependence structure ensures that the number of runs grows strictly slower than the sample size (i.e., $R_n = o(n)$), the subtraction term vanishes asymptotically, whereby $\xi_n'$ tends to 1.

This equivalence places $\xi_n'$ within the classical theory of runs and provides intuition for its behavior under independence. While classical results on \(R_n\) provide an alternative route to establishing asymptotic normality under the null, our proof works directly with \(A_n\), the average of the adjacent-match indicators.

\subsection{Convergence to the population coefficient}
\label{subsec:convtopop}
We establish the strong consistency of \(\xi_n'\) for an arbitrary real-valued \(X\), including continuous, discrete, and mixed distributions.
The result imposes no support, density, or moment conditions on \(X\), and no continuity assumptions on the conditional class probabilities.

\begin{theorem}[Consistency of $\xi_n'$]
\label{thm:consistency-xin}
Let $(X_i,Y_i)_{i=1}^n$ be i.i.d. samples from a pair $(X,Y)$, where $X$ is an arbitrary real-valued random variable and $Y$ is a categorical random variable taking values in $\{1,\ldots,k\}$ with $k\ge2$. Assume that $Y$ is not almost surely constant. Then
\[
   \xi_n^\prime (X,Y) \;\xrightarrow{a.s.}\; \xi^\prime(X,Y),
   \qquad n\to\infty.
\]
\end{theorem}
As discussed following Eq.~\eqref{eq:xisample}, the proof is based on Lusin's theorem \citep{Folland1999}. The full details are provided in Appendix~\ref{app:proof_consistency}.

\subsection{Asymptotic normality under independence}
We next derive the asymptotic distribution of $\xi_n'$ under the independence hypothesis. The key point is that, under $X\perp Y$, the ordering induced by $X$ contains no information about the labels. This remains true even when $X$ is not continuous, provided that ties in $X$ are broken independently at random. 

Consequently, the ordered label sequence $(Y_{(1)},\ldots,Y_{(n)})$ has the same distribution as an i.i.d.\ categorical sequence with probabilities $p_1,\ldots,p_k$.
As noted in Section~\ref{subsubsec:runs-xiprime}, this representation links $\xi_n'$ to the classical runs statistic. One could therefore appeal directly to classical asymptotic results for runs statistics, such as those of \citet{BartonDavid1957}. 
In our technical argument in Appendix~\ref{app:proof_asymptotic_normality}, however, we work directly with the adjacent-match indicators
\[
W_i=\mathbf{1}\{Y_{(i+1)}=Y_{(i)}\},
\qquad i=1,\ldots,n-1.
\]
Under independence, the sequence \((W_i)_{i=1}^{n-1}\) is bounded and \(1\)-dependent (i.e., \(W_i\) and \(W_j\) are dependent only when \(|i-j|\leq 1\)). This formulation provides a direct route to the null limit and aligns with the decomposition used later under general dependence.

Using the notation introduced in Section~\ref{sec:xiprime}, set
\[
\sigma^2:=B(1+B)-2\rho,
\qquad
\kappa^2:=\frac{\sigma^2}{(1-B)^2}.
\]

\begin{theorem}[Asymptotic normality of $\xi_n'$ under independence]
\label{thm:asymptotic-xin}
Let $(X_i,Y_i)_{i=1}^n$ be i.i.d.\ samples from $(X,Y)$, where $X$ is an arbitrary real-valued random variable and $Y\in\{1,\ldots,k\}$ with $k\ge2$. Assume that $X\perp Y$ and that $Y$ is not almost surely constant. Then
\[
    \sqrt n\,\xi_n'
    \xrightarrow{d}
    N(0,\kappa^2),
\qquad n\to\infty,
\]
where $\kappa^2>0$.
Moreover, with the plug-in estimators
\[
\begin{aligned}
    \hat\rho
    :=
    \sum_{j=1}^k\hat p_j^3, \qquad
    \hat\sigma^2
    :=
    B_n(1+B_n)-2\hat\rho,
    \qquad
    \hat\kappa^2
    :=
    \frac{\hat\sigma^2}{(1-B_n)^2},
\end{aligned}
\]
we have $\hat\kappa^2\xrightarrow{p}\kappa^2$ and
\[
    Z_n
    :=
    \frac{\sqrt n\,\xi_n'}{\hat\kappa}
    \xrightarrow{d}
    N(0,1),
    \qquad n\to\infty.
\]
\end{theorem}
The detailed proof is provided in Appendix~\ref{app:proof_asymptotic_normality}.

Unlike the general asymptotic normality result discussed later, Theorem~\ref{thm:asymptotic-xin} requires neither smoothness of the conditional class probabilities nor continuity of the marginal distribution of $X$. The only role of the ordering by $X$ is to generate a random permutation of the labels, and under independence this permutation is independent of the labels themselves. Thus arbitrary real-valued $X$, including discrete and mixed distributions, is allowed once ties are broken independently at random.

Under independence, the statistic has an exact conditional finite sample negative bias on the event $\{B_n<1\}$. Let $N_j=\sum_{i=1}^n\mathbf 1\{Y_i=j\}$. Conditional on $(N_1,\ldots,N_k)$, the sorted labels are uniformly distributed over all permutations of the observed labels, and hence
\[
\begin{aligned}
E(A_n\mid N_1,\ldots,N_k)
&=\frac{\sum_{j=1}^kN_j(N_j-1)}{n(n-1)}\\
&=\frac{nB_n-1}{n-1}.
\end{aligned}
\]
Consequently,
\begin{align}
    E(\xi_n'\mid N_1,\ldots,N_k)=-\frac1{n-1}.
\end{align}
Thus, the negative mean occurs in finite sample.

Motivated by this identity, define the finite sample centered statistic
\begin{align}
    Z_n^{\mathrm{c}}
    :=
    \frac{\sqrt n\left\{\xi_n'+1/(n-1)\right\}}{\hat\kappa}.
\end{align}
The correction does not affect the first-order limit because
\[
    Z_n^{\mathrm{c}}-Z_n
    =\frac{\sqrt n}{(n-1)\hat\kappa}
    =o_p(1).
\]
Hence $Z_n^{\mathrm{c}}\xrightarrow{d}N(0,1)$ under independence.

The asymptotic distribution derived in Theorem~\ref{thm:asymptotic-xin} can therefore be employed to construct a finite sample centered test of independence.
Since $\xi_n'$ tends to be larger under dependence, it is natural to consider a one-sided test. We consider the hypothesis testing problem:
\[
H_0: X \perp Y \quad \text{vs.} \quad H_1: X \not\perp Y.
\]
At significance level $\alpha \in (0,1)$, we propose the test
\begin{align}
\phi_n \;=\;
\begin{cases}
1, & \text{if } Z_n^{\mathrm{c}} > z_{1-\alpha}, \\[0.5ex]
0, & \text{otherwise},
\end{cases}
\end{align}
where $\phi_n=1$ indicates rejection of $H_0$.

The validity of the proposed test is guaranteed by the following corollaries, which establish the asymptotic size control and the consistency of the test.

\begin{corollary}[Asymptotic size]
\label{cor:type1_error}
Under the null hypothesis $H_0: X \perp Y$, the test $\phi_n$ controls the type I error asymptotically at level $\alpha$. That is,
\[
\lim_{n\to\infty} P_{H_0}(\phi_n = 1) = \alpha.
\]
\end{corollary}

\begin{corollary}[Consistency of the test]
\label{cor:power_consistency}
Suppose that the alternative holds such that $\xi' > 0$. Then the test $\phi_n$ is consistent, meaning that the power approaches 1 as the sample size increases:
\[
\lim_{n\to\infty} P_{H_1}(\phi_n = 1) = 1.
\]
\end{corollary}
The detailed proofs of Corollaries~\ref{cor:type1_error} and~\ref{cor:power_consistency} are provided in Appendix~\ref{subsec:proof-corollaries}.

Although a permutation test can be used to assess the significance of $\xi_n'$ for small samples, it becomes computationally expensive as the sample size $n$ increases. The asymptotic test based on Theorem~\ref{thm:asymptotic-xin} offers a significant advantage by avoiding this computational cost.

\subsection{Asymptotic normality under dependence}
\label{subsec:general_clt}

While Theorem~\ref{thm:asymptotic-xin} establishes asymptotic normality under independence, we can also derive the distribution of $\xi_n'$ under general dependence. This result is crucial for constructing confidence intervals for the population coefficient $\xi'$.

Using the notation introduced in Section~\ref{sec:xiprime}, define
\begin{align}
\sigma^2_A &:= E\bigl[h(X) - 3h(X)^2 + 2\tau(X)\bigr] + \mathrm{Var}(h(X)), \label{eq:sigma_A}\\
\sigma^2_B &:= 4\bigl(\rho-B^2\bigr), \label{eq:sigma_B}\\
\sigma_{AB} &:= 4 E \biggl[\sum_{j=1}^k p_j g_j(X)^2 - h(X)m(X)\biggr] \nonumber \\
&\quad + 2E\Bigl[\bigl(h(X)-A\bigr)\bigl(m(X)-B\bigr)\Bigr]. \label{eq:sigma_AB}
\end{align}
Let
\[
\Sigma :=
\begin{pmatrix}
\sigma^2_A & \sigma_{AB} \\
\sigma_{AB} & \sigma^2_B
\end{pmatrix}.
\]

\begin{assumption}\label{asm:general_clt}
There exists a choice of the Borel versions $g_1,\ldots,g_k$ introduced in Section~\ref{sec:xiprime} and a deterministic sequence of bounded intervals
\[
    I_n=[\ell_n,u_n]\subset\mathbb R
\]
such that
\[
    nP(X\notin I_n)\to0
\]
and
\[
    \max_{1\le j\le k}V_{I_n}(g_j)=o(\sqrt n),
\]
where, for an interval $I\subset\mathbb R$,
\[
    V_I(g_j)
    :=
    \sup_{m\in\mathbb N}
    \sup_{x_0<\cdots<x_m,\,x_0,\ldots,x_m\in I}
    \sum_{\ell=1}^m
    |g_j(x_\ell)-g_j(x_{\ell-1})|.
\]
\end{assumption}

The first condition ensures, by a union bound, that all sample points lie in $I_n$ with probability tending to one. On this event, for every fixed integer $r\ge1$,
\[
    \sum_{i=1}^{n-r}
    |g_j(X_{(i+r)})-g_j(X_{(i)})|
    \le rV_{I_n}(g_j)
    =o(\sqrt n),
\]
so the second condition controls the cumulative oscillation of the conditional class probabilities at the $\sqrt n$ scale.

If $X$ is supported on a bounded interval and each $g_j$ has bounded variation on that interval, Assumption~\ref{asm:general_clt} is immediate. It also permits infinitely many discontinuities. For example, if $X\sim N(0,1)$ and
\[
    g_1(x)=\frac12+\frac14\operatorname{sgn}(\sin x),
    \qquad g_2(x)=1-g_1(x),
\]
then, with $I_n=[-2\sqrt{\log n},2\sqrt{\log n}]$, the tail condition holds and $V_{I_n}(g_j)=O(\sqrt{\log n})=o(\sqrt n)$.

\begin{theorem}[General asymptotic normality of $\xi_n'$]
\label{thm:general_clt}
Let $(X_i,Y_i)_{i=1}^n$ be i.i.d. samples from $(X,Y)$, where $X$ is a real-valued random variable and $Y$ is a categorical random variable taking a fixed number $k\geq2$ of categories. Suppose Assumption~\ref{asm:general_clt} holds and $Y$ is not almost surely constant. Then, as $n \to \infty$,
\[
\sqrt{n} (\xi_n' - \xi') \xrightarrow{d} N(0, \kappa^2_D),
\]
where
\[
\kappa^2_D = \nabla f(A,B)^T \Sigma \nabla f(A,B),
\]
with $f(a,b)=(a-b)/(1-b)$ and
\[
\nabla f(A,B)=
\begin{pmatrix}
\dfrac{1}{1-B}\\[0.8em]
\dfrac{A-1}{(1-B)^2}
\end{pmatrix}.
\]
\end{theorem}
The detailed proof is deferred to Appendix~\ref{app:proof_general_clt}.

Since the asymptotic variance $\kappa^2_D$ depends on the unknown conditional probability functions $g_j(x)$, practical implementation requires a plug-in estimator. Throughout the asymptotic theory, $k$ denotes the fixed number of response categories. It is distinct from $k_{\mathrm{nn}}$, the neighborhood size used only in the k-nearest-neighbor plug-in variance estimator. We use the following nonparametric procedure.
\begin{enumerate}
    \item Estimate the conditional probabilities $\hat{g}_j(X_i)$ using a k-nearest-neighbor (k-NN) estimator. Specifically, for each sample $X_i$, let $\mathcal{N}_{k_{\mathrm{nn}}}(X_i)$ denote the set of indices of the $k_{\mathrm{nn}}$ nearest neighbors of $X_i$. The estimator is given by the local proportion of class $j$:
    \[
    \hat{g}_j(X_i) = \frac{1}{k_{\mathrm{nn}}} \sum_{\ell \in \mathcal{N}_{k_{\mathrm{nn}}}(X_i)} \mathbf{1}\{Y_\ell = j\}. \tag{*}
    \]
    For consistency of the plug-in estimator, $k_{\mathrm{nn}}$ is chosen so that $k_{\mathrm{nn}}\to\infty$ and $k_{\mathrm{nn}}/n\to0$. In the simulations, we use the default choice $k_{\mathrm{nn}}=\lfloor\sqrt n\rfloor$.

    \item Compute the empirical counterparts of the moment functions:
    \[
    \hat{h}(X_i) = \sum_{j=1}^k \hat{g}_j(X_i)^2, \qquad \hat{\tau}(X_i) = \sum_{j=1}^k \hat{g}_j(X_i)^3.
    \]

    \item Replace the population expectations $E[\cdot]$ in Eq.~\eqref{eq:sigma_A}--\eqref{eq:sigma_AB} with their moment estimators and substitute the sample proportions $\hat{p}_j$ for $p_j$.
\end{enumerate}
The above procedure yields a plug-in estimator $\hat{\kappa}_{D}^2$, which enables the construction of Wald-type confidence intervals for $\xi'$.

\begin{proposition}[Consistency of $\hat{\kappa}_{D}$]
\label{prop:variance_consistency}
Assume the conditions of Theorem~\ref{thm:general_clt}. In addition, suppose that the nonparametric estimators satisfy
\[
    \frac1n\sum_{i=1}^n |\hat g_j(X_i)-g_j(X_i)|\xrightarrow{p}0,
    \qquad j=1,\ldots,k. \tag{**}
\]
Then the plug-in variance estimator is consistent:
\[
\hat{\kappa}_{D}^2 \xrightarrow{p} \kappa_{D}^2.
\]
\end{proposition}
The proof is provided in Appendix~\ref{app:proof_variance_consistency}.
Moreover, if $X$ has an absolutely continuous distribution and each $g_j$ is bounded and continuous, then it follows from \cite{Gyorfietal2002} that the k-NN estimator defined in~(*) satisfies~(**).

Based on the asymptotic normality and the consistent variance estimator, we construct an asymptotic $(1-\alpha)$ confidence interval for $\xi'$. Let $z_{\alpha/2}$ denote the $(1-\alpha/2)$-quantile of the standard normal distribution. We define the confidence interval $\mathcal{C}_n(1-\alpha)$ as
\begin{equation} \label{eq:ci_definition}
    \mathcal{C}_n(1-\alpha) = \left[ \xi_n' - z_{\alpha/2} \frac{\hat{\kappa}_D}{\sqrt{n}}, \quad \xi_n' + z_{\alpha/2} \frac{\hat{\kappa}_D}{\sqrt{n}} \right].
\end{equation}

\begin{corollary}[Asymptotic confidence interval]
\label{cor:confidence_interval}
Under the conditions of Theorem~\ref{thm:general_clt} and Proposition~\ref{prop:variance_consistency}, assume in addition that $\kappa_D^2>0$. Then the confidence interval $\mathcal{C}_n(1-\alpha)$ defined in Eq.~\eqref{eq:ci_definition} is asymptotically valid. That is, for any $\alpha \in (0,1)$,
\[
\lim_{n\to\infty} P\bigl(\xi' \in \mathcal{C}_n(1-\alpha)\bigr) = 1-\alpha.
\]
\end{corollary}

The proof is provided in Appendix~\ref{app:proof_confidence_interval}.

\begin{remark}[Practical implementation]
Since the population coefficient $\xi'$ is bounded within $[0,1]$, the asymptotic confidence interval $\mathcal{C}_n(1-\alpha)$ constructed above may occasionally exceed these bounds in finite samples. In practice, the interval should be truncated to lie within $[0,1]$:
\[
\tilde{\mathcal{C}}_n(1-\alpha) = \mathcal{C}_n(1-\alpha) \cap [0, 1].
\]
If a bootstrap-based confidence interval is desired, the standard bootstrap with replacement is not suitable for $\xi_n'$, because duplicate observations create artificial adjacent matches after sorting. In our numerical experiments, the $m$-out-of-$n$ procedure, implemented by sampling $m$ observations without replacement, provides a more reasonable resampling alternative. A similar procedure is used for $\xi_n$ in \cite{DetteKroll2025}. See Section~\ref{subsec:bootstrap-comparison} for a numerical comparison.
\end{remark}

\section{Simulation studies}
\label{sec:sim}

We report comprehensive simulation results to evaluate the finite-sample performance of $\xi_n'$. Our simulation study is organized into four parts. First, we examine the asymptotic normality and calibration, confirming that the standardized statistic follows the standard normal distribution under both independence ($H_0$) and dependence ($H_1$), for finite but large sample sizes. We also validate the empirical coverage probabilities of our confidence intervals, ensuring that they align with the nominal levels in finite samples. Second, we compare the proposed plug-in confidence intervals with the standard bootstrap and the $m$-out-of-$n$ bootstrap. Third, we demonstrate coding invariance by contrasting the stability of $\xi_n'$ under random label permutations, compared with Chatterjee's $\xi_n$. Finally, we numerically compare the performance of the $\xi_n'$-based independence test against established competitors such as the independence tests based on Distance Covariance (dCov) and the Hilbert-Schmidt Independence Criterion (HSIC).

\subsection{Simulation setup: Block design model}
To control the dependence strength, we employ a piecewise-constant block design model. We consider $k=6$ categories for $Y$, and draw $X \sim \mathrm{Unif}(0,1)$ and partition the unit interval into 6 equal-sized blocks. For each block $b \in \{1, \dots, 6\}$, we assign a preferred class label $j_b = b$. Conditional on $X$ falling in block $b$, the label $Y$ is sampled from:
\begin{align}
\Pr(Y=j_b \mid X \in \text{block } b) &= \theta + \frac{1-\theta}{6}, \nonumber \\
\Pr(Y = j \mid X \in \text{block } b) &= \frac{1-\theta}{6}, \qquad j \neq j_b,
\label{setting:block}
\end{align}
where $\theta \in [0,1]$ controls the signal strength. 
Here, $\theta=0$ corresponds to independence between $X$ and $Y$ (uniform noise), while $\theta=1$ yields a deterministic functional relationship.
Figure~\ref{fig:xy-panels} visualizes the data distribution. As $\theta$ increases, the concentration of labels within specific $X$-intervals becomes pronounced.

\begin{figure}[tp]
\centering
\includegraphics[width=\linewidth]{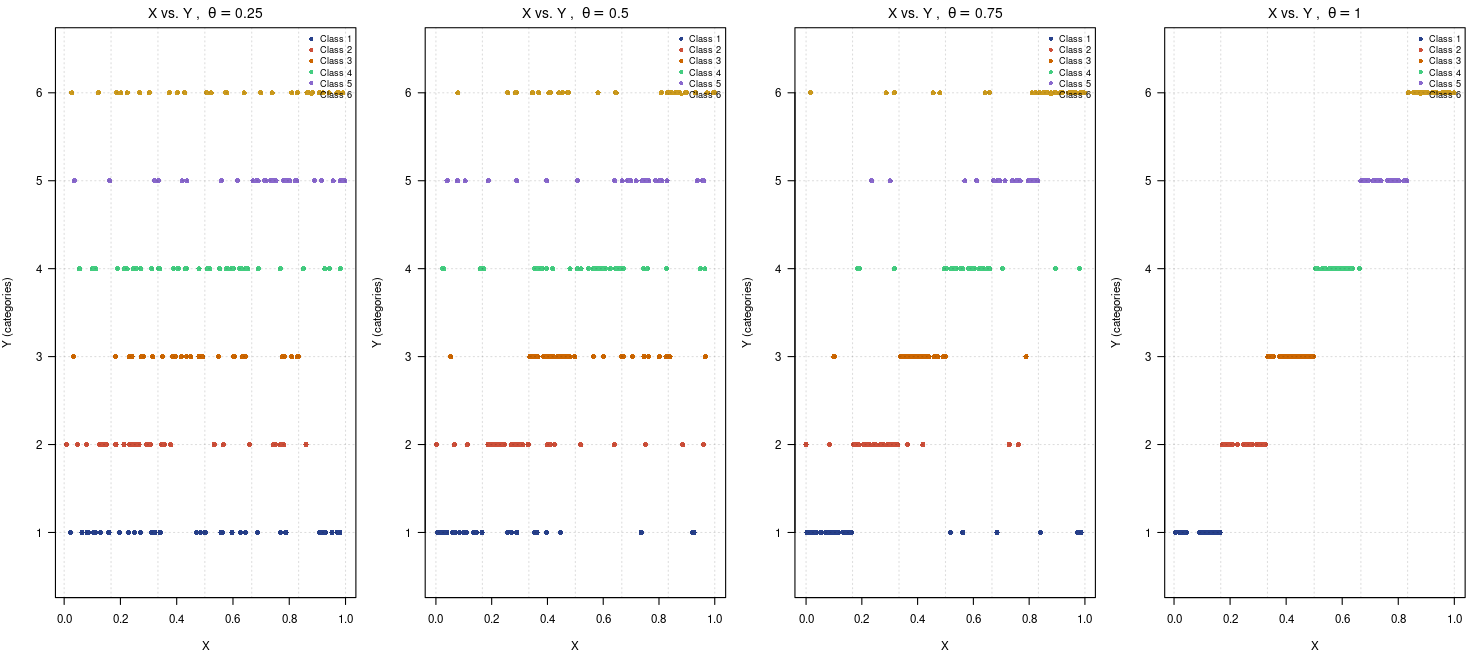}
\caption{Scatterplots of $(X,Y)$ under the block design model. As the signal $\theta$ increases from 0 to 1, the label homogeneity within each block becomes more noticeable, transitioning from random noise to a deterministic step function.}
\label{fig:xy-panels}
\end{figure}

\subsection{Asymptotic normality and validity of confidence intervals}
\subsubsection{Asymptotic normality under independence}
\label{subsec:nullcalib}

We first validate Theorem~\ref{thm:asymptotic-xin} by examining the null distribution under independence between $X$ and $Y$.
We generated data with $\theta=0$ and $k=6$ for varying sample sizes $n \in \{50, 100, 200, 400\}$. For each replicate, we computed both the original standardized statistic
\[
Z_n=\frac{\sqrt n\,\xi_n'}{\widehat\kappa}
\]
and its finite-sample centered version
\[
Z_n^{\mathrm{c}}
=
\frac{\sqrt n\left\{\xi_n'+1/(n-1)\right\}}{\widehat\kappa},
\]
where $\widehat\kappa^2$ is the plug-in variance estimator defined in Theorem~\ref{thm:asymptotic-xin}.

Figure~\ref{fig:null-panels} displays the histograms of $Z_n^{\mathrm{c}}$ based on $R_{\mathrm{MC}}=2000$ Monte Carlo replicates. The solid blue curve is the standard normal density, and the red dashed curve is the empirical Gaussian fit. Table~\ref{tab:null-variance} compares the original and centered statistics. The original $Z_n$ has a visible negative mean, whereas centering moves the empirical mean close to zero and brings the one-sided rejection rate closer to the nominal level. The empirical variance of $Z_n^{\mathrm{c}}$ approaches one, and the Shapiro--Wilk results indicate improving normal approximation as $n$ increases. These results support the use of the centered one-sided Wald test for independence.

\begin{figure}[tp]
\centering
\includegraphics[width=\linewidth]{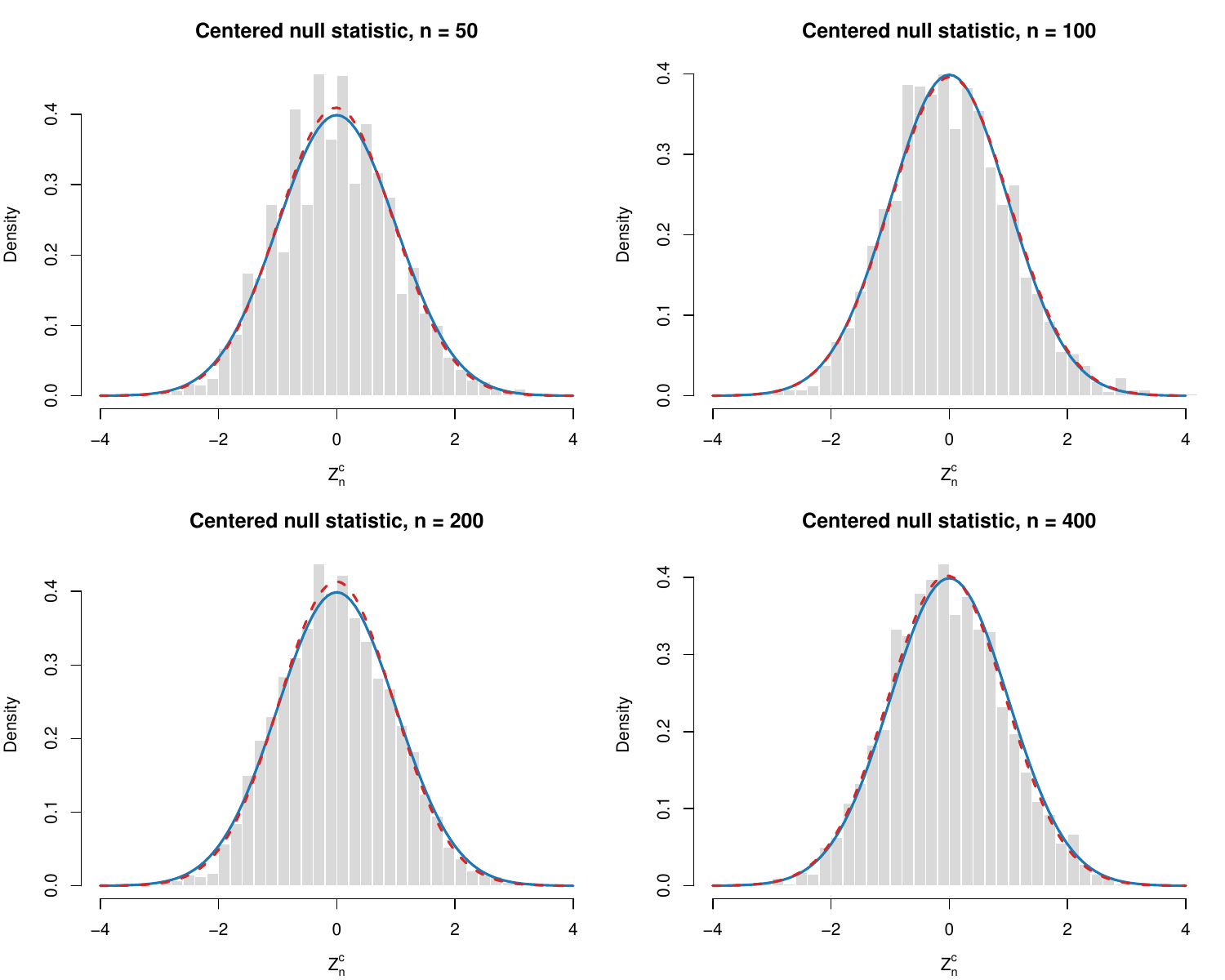}
\caption{Null calibration of $Z_n^{\mathrm{c}}=\sqrt n\{\xi_n'+1/(n-1)\}/\widehat\kappa$ for $k=6$. Histograms are overlaid with the standard normal density (solid blue) and the empirical Gaussian fit (red dashed).}
\label{fig:null-panels}
\end{figure}

\begin{table}[t]
\centering
\small
\caption{Null calibration under $H_0$ based on $R_{\mathrm{MC}}=2000$ Monte Carlo replicates. Reported are the empirical means of the original and centered standardized statistics, the empirical variance of $Z_n^{\mathrm{c}}$, their one-sided rejection rates at level $0.05$, and the Shapiro--Wilk test $p$-value for $Z_n^{\mathrm{c}}$.}
\label{tab:null-variance}
\begin{tabular}{rcccccc}
\toprule
$n$ & Mean $Z_n$ & Mean $Z_n^{\mathrm{c}}$ & Var.\ $Z_n^{\mathrm{c}}$ & Rej.\ $Z_n$ & Rej.\ $Z_n^{\mathrm{c}}$ & Shapiro $p$ \\
\midrule
 50 & $-0.3234$ & $-0.0129$ & 0.9494 & 0.024 & 0.046 & $6.02\times 10^{-3}$ \\
100 & $-0.2067$ & $ 0.0147$ & 1.0135 & 0.038 & 0.055 & $2.48\times 10^{-4}$ \\
200 & $-0.1639$ & $-0.0066$ & 0.9300 & 0.031 & 0.045 & $1.44\times 10^{-2}$ \\
400 & $-0.1513$ & $-0.0398$ & 0.9835 & 0.041 & 0.050 & $3.98\times 10^{-1}$ \\
\bottomrule
\end{tabular}
\end{table}

\subsubsection{Asymptotic normality under dependence}
\label{subsec:general_normality}

Next, we verify the asymptotic normality under dependence (Theorem~\ref{thm:general_clt}). 
We computed the standardized statistic $Z = \sqrt{n}(\xi_n' - \xi') / \hat{\kappa}_{D}$ using the consistent variance estimator $\hat{\kappa}_{D}^2$ described in Section~\ref{subsec:general_clt}.

We simulated data under the block design model with a moderate signal ($\theta=0.5$) and sample size $n=2000$. Figure~\ref{fig:h1-normality} shows the histogram of the standardized statistic across 1000 replicates. The distribution aligns well with the standard normal density (red dashed curve).

\begin{figure}[tp]
\centering
\includegraphics[width=\linewidth]{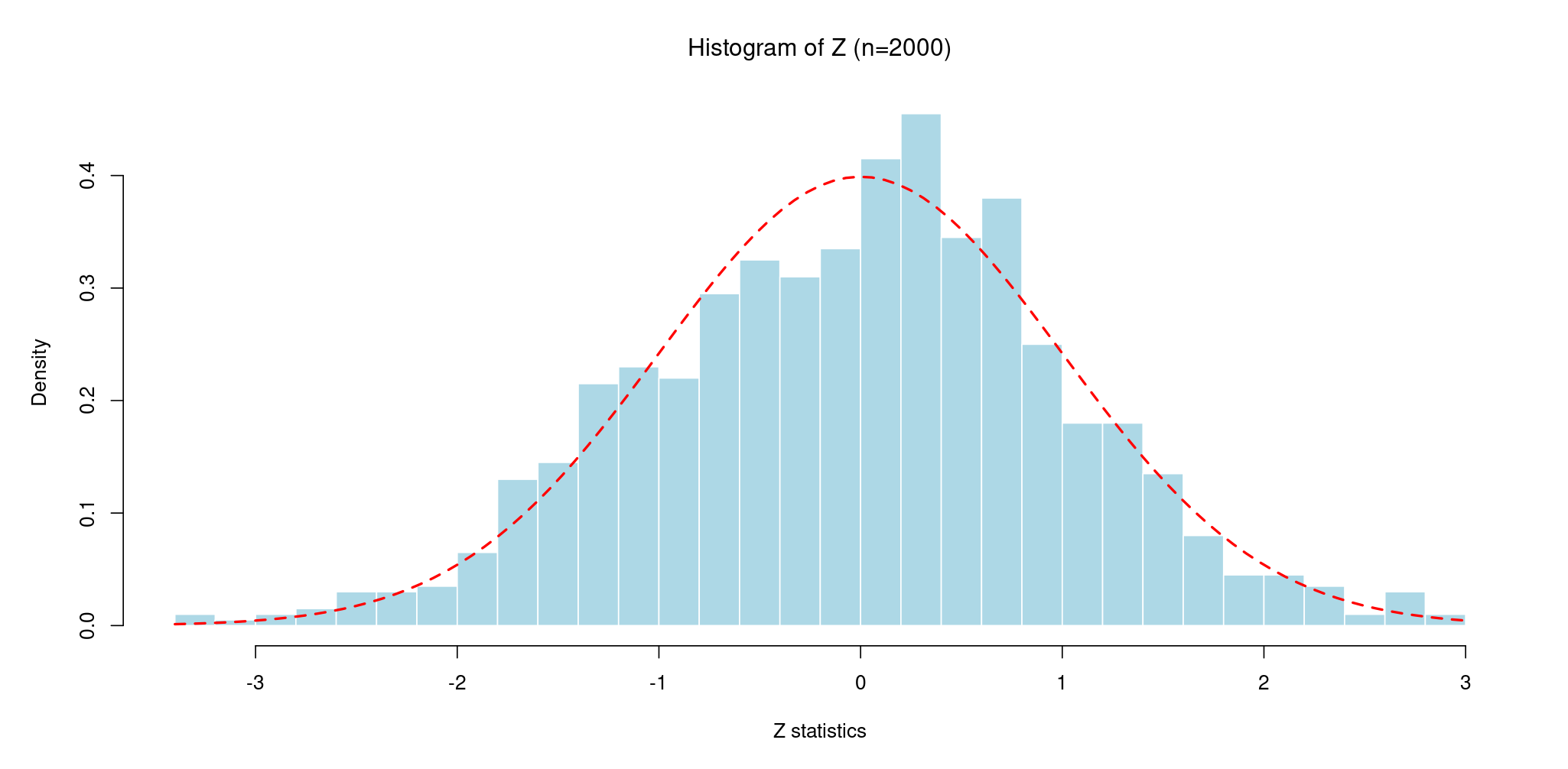} 
\caption{Distribution of the standardized statistic under dependence ($\theta=0.5, n=2000$). The histogram represents the Monte Carlo samples of $\sqrt{n}(\xi_n' - \xi')/\hat{\kappa}_{D}$, and the red dashed curve is the standard normal density $N(0,1)$.}
\label{fig:h1-normality}
\end{figure}

Table~\ref{tab:h1-stats} summarizes the empirical characteristics of the standardized statistic $Z$. The empirical mean and standard deviation are remarkably close to $0$ and $1$, respectively. This confirms that the proposed variance estimator $\hat{\kappa}_{D}^2$ is consistent even under dependence, allowing for the construction of valid confidence intervals for $\xi'$.

\begin{table}[t]
\centering
\small
\caption{Summary statistics of the standardized statistic $Z = \sqrt{n}(\xi_n' - \xi')/\hat{\kappa}_{D}$ under dependence ($\theta=0.5, n=2000$) based on 1000 replicates.}
\label{tab:h1-stats}
\begin{tabular}{ccccc}
\toprule
Mean & Median & SD & Min & Max \\
\midrule
-0.021 & 0.059 & 1.031 & -3.376 & 2.961  \\
\bottomrule
\end{tabular}
\end{table}

The finite-sample slight downward bias shown in Table~\ref{tab:h1-stats} is not universal but an artifact of this block-design experiment. To see this, consider the adjacent-match term $A_n$. Write $g_i=g(X_{(i)})$, where $g(x)=(g_1(x),\ldots,g_k(x))$, and let $q_i=\lVert g_i\rVert^2$. Conditional on the ordered $X$-values,
\[
\begin{aligned}
E(A_n\mid X_1,\ldots,X_n)-\frac1n\sum_{i=1}^nq_i
&=\frac{n^{-1}\sum_{i=1}^nq_i-(q_1+q_n)/2}{n-1}\\
&\quad-\frac1{2(n-1)}\sum_{i=1}^{n-1}\lVert g_i-g_{i+1}\rVert^2.
\end{aligned}
\]
In general, the first term need not have a fixed sign, so this finite-sample bias in the adjacent-match term does not have a fixed sign. In the block design model, however, the conditional probability vectors in different blocks are permutations of the same vector, and hence $q_1=\cdots=q_n$. The first term is then zero, while the second is non-positive. Thus, $A_n$ has a negative bias in this design.

\subsubsection{Empirical coverage of confidence intervals}
\label{subsec:ci_coverage}

A distinct advantage of the block design model defined in Eq.~\eqref{setting:block} is that the population quantity $\xi'$ admits a closed-form expression for any given $\theta$. This analytical tractability allows us to evaluate the finite-sample validity of our proposed confidence intervals. 
We validated whether the 95\% confidence intervals constructed using the asymptotic variance estimator $\widehat{\kappa}_D^2$ correctly contain the true $\xi'$.
Table~\ref{tab:coverage_sim} summarizes the results based on 1000 replicates. We considered a range of signal strengths $\theta \in \{0.1, 0.3, 0.5, 0.7, 0.9\}$ and sample sizes $n \in \{200, 400, 800, 1600\}$.
The results demonstrate that the proposed method achieves a coverage close to the nominal 95\% level as $n$ increases.

\begin{table}[t]
\centering
\small
\caption{Empirical coverage probabilities of 95\% confidence intervals for $\xi'$ under the block design model (1000 replicates).}
\label{tab:coverage_sim}
\begin{tabular}{cccccc}
\toprule
 & \multicolumn{5}{c}{Signal Strength ($\theta$)} \\
\cmidrule(lr){2-6}
Sample Size ($n$) & 0.1 & 0.3 & 0.5 & 0.7 & 0.9 \\
\midrule
200  & 0.970 & 0.955 & 0.946 & 0.918 & 0.929 \\
400  & 0.964 & 0.960 & 0.931 & 0.927 & 0.941 \\
800  & 0.964 & 0.962 & 0.940 & 0.944 & 0.952 \\
1600 & 0.963 & 0.946 & 0.949 & 0.951 & 0.952 \\
\bottomrule
\end{tabular}
\end{table}

Figure~\ref{fig:ci_zoom} provides a focused view of the finite-sample performance for $n=500$ within the interval $\theta \in [0.20, 0.25]$.
The x-axis represents the signal strength $\theta$, and the y-axis denotes the value of the statistic. The red line indicates the true population value $\xi'$, while the solid blue line tracks the sample estimates $\xi_n'$, obtained from one sample of size $n=500$ for each value of $\theta$. The shaded blue region represents the 95\% pointwise confidence intervals.

As illustrated, the lower limit of the confidence interval (dotted blue line) crosses zero at $\theta \approx 0.233$. At this threshold, the population association is merely $\xi' \approx 0.054$. This indicates that our confidence interval is precise enough to identify meaningful dependence even when the underlying signal is very weak.

\begin{figure}[tp]
\centering
\includegraphics[width=\linewidth]{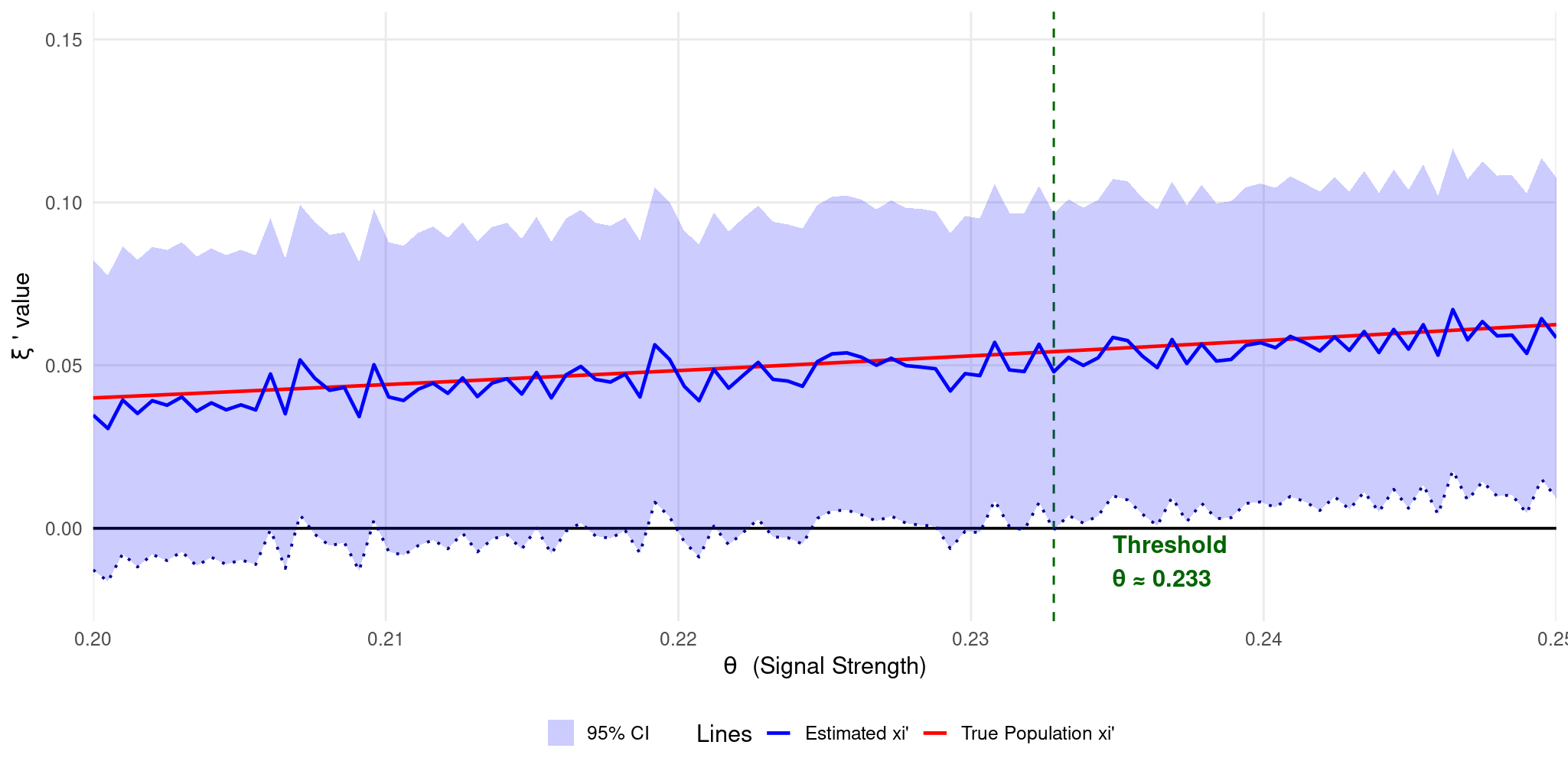}
\caption{Finite-sample performance of $\xi_n'$ around the detection threshold ($n=500$). The plot zooms in on the interval $\theta \in [0.20, 0.25]$. The lower confidence bound (dotted blue line) crosses the zero at $\theta \approx 0.233$, corresponding to a true association of $\xi' \approx 0.054$.}
\label{fig:ci_zoom}
\end{figure}

\subsubsection{Comparison with bootstrap confidence intervals}
\label{subsec:bootstrap-comparison}

We further compare three methods for constructing confidence intervals for $\xi'$ under dependence:
\begin{enumerate}[label=(\roman*)]
\item the proposed plug-in variance estimator $\hat\kappa_D$;
\item an $m$-out-of-$n$ bootstrap implemented by drawing $m_n=\lfloor n^\gamma\rfloor$ observations without replacement, equivalently a subsampling procedure;
\item the standard bootstrap with replacement.
\end{enumerate}
The $m$-out-of-$n$ bootstrap has been studied for inference on Chatterjee's rank correlation, whose rank-adjacency structure is closely related to that of $\xi_n'$ \citep{DetteKroll2025}. In that setting, the standard bootstrap is shown to fail, which motivates the use of an $m$-out-of-$n$ resampling scheme.

All methods in this subsection are evaluated under the six-category block design model in Eq.~\eqref{setting:block}, so the bootstrap comparison uses the same simulation setting as the rest of this section. In this model, the marginal class probabilities are $p_j=1/6$, and the population coefficient has the closed form $\xi'=\theta^2$.
The simulation was run with $R_{\mathrm{MC}}=1000$ Monte Carlo repetitions and $R_{\mathrm{boot}}=999$ bootstrap repetitions. Thus, for a nominal coverage probability of $0.95$, the Monte Carlo standard error is approximately
\[
    \sqrt{\frac{0.95(1-0.95)}{1000}}\approx 0.0069.
\]

For the standard bootstrap, we estimate the asymptotic variance by
\[
\widehat v_{\mathrm{boot}}
=
n\,\operatorname{Var}^*(\xi_n^{\prime *}),
\]
whereas the $m$-out-of-$n$ method uses
\[
\widehat v_{m}
=
m_n\,\operatorname{Var}^*(\xi_{m_n}^{\prime *}).
\]
Both resampling intervals are the Wald intervals
\[
\xi_n'\pm z_{\alpha/2}\sqrt{\widehat v/n},
\]
with the corresponding variance estimate $\widehat v$.

The standard bootstrap is structurally problematic because $\xi'$ is based on adjacent label matches after sorting by $X$. In an ordinary $n$-out-of-$n$ bootstrap sample, the proportion of distinct original observations converges to $1-e^{-1}$, leaving a nonvanishing fraction of repeated copies. After sorting by $X$, these copies become adjacent and share the same label, creating artificial adjacent matches whose effect does not vanish as $n$ increases. This distorts the bootstrap variance estimate. A related failure of the standard bootstrap has been established for Chatterjee's rank correlation, which has a closely related adjacency-based structure \citep{lin_failure_2024}.

Table~\ref{tab:bootstrap-method-comparison} summarizes the results for the representative dependence setting $\theta=0.75$. Here $k_{\mathrm{nn}}$ denotes the number of neighbors used in the plug-in estimator, and $m_n=\lfloor n^\gamma\rfloor$ for the $m$-out-of-$n$ procedure. For this comparison, we use $k_{\mathrm{nn}}=\lfloor\sqrt n\rfloor$ for the plug-in estimator and $m_n=\lfloor n^{0.65}\rfloor$ for the $m$-out-of-$n$ procedure.
Because \citet{DetteKroll2025} consider several fixed and data-adaptive choices of $\gamma$, we examined
\[
    \gamma\in\{0.50,0.65,0.75\}.
\]
The choices $0.50$ and $0.75$ are among those considered by \citet{DetteKroll2025}, and $0.65$ is an intermediate value. We use $\gamma=0.65$, as it yielded the best performance in our setting. This choice is empirical rather than theoretically justified, and the results should therefore be interpreted as a finite-sample comparison rather than a theoretical validation.
Here rRMSE denotes the relative root mean squared error of the variance estimate, defined as
\[
    \mathrm{rRMSE}
    =
    \frac{
    \sqrt{
    R_{\mathrm{MC}}^{-1}
    \sum_{r=1}^{R_{\mathrm{MC}}}
    \left(\widehat v_r-\kappa_D^2\right)^2
    }
    }{
    \kappa_D^2
    },
\]
where $\widehat v_r$ is the variance estimate from the $r$-th Monte Carlo repetition. Under the block design model, the population asymptotic variance also has the closed form
\[
\kappa_D^2
=
\frac{1+6\theta^2+8\theta^3-15\theta^4}{5},
\]
which we use as the target variance in the rRMSE calculation.

\begin{table}[htbp]
\centering
\small
\caption{Comparison of variance and confidence interval methods under the block design model for $\theta=0.75$ ($\xi'=0.5625$). Table reports coverage, average confidence interval length, and relative RMSE of the variance estimate.}
\label{tab:bootstrap-method-comparison}
\begin{tabular}{clccc}
\toprule
$n$ & Method & Coverage & Avg. length & rRMSE \\
\midrule
200 & Plug-in, $k_{\mathrm{nn}}=\lfloor\sqrt n\rfloor$ & 0.929 & 0.199 & 0.143 \\
200 & $m$-out-of-$n$, $m_n=\lfloor n^{0.65}\rfloor$ & 0.893 & 0.181 & 0.291 \\
200 & Standard bootstrap & 0.685 & 0.114 & 0.718 \\
\midrule
400 & Plug-in, $k_{\mathrm{nn}}=\lfloor\sqrt n\rfloor$ & 0.932 & 0.144 & 0.109 \\
400 & $m$-out-of-$n$, $m_n=\lfloor n^{0.65}\rfloor$ & 0.918 & 0.136 & 0.207 \\
400 & Standard bootstrap & 0.686 & 0.081 & 0.716 \\
\midrule
800 & Plug-in, $k_{\mathrm{nn}}=\lfloor\sqrt n\rfloor$ & 0.932 & 0.103 & 0.081 \\
800 & $m$-out-of-$n$, $m_n=\lfloor n^{0.65}\rfloor$ & 0.925 & 0.099 & 0.153 \\
800 & Standard bootstrap & 0.711 & 0.057 & 0.714 \\
\bottomrule
\end{tabular}
\end{table}

In this block design experiment, the plug-in estimator performs best overall, providing coverage closest to the nominal level and the lowest rRMSE across sample sizes. The $m$-out-of-$n$ procedure improves as the sample size increases and clearly outperforms the standard bootstrap, although it remains less accurate than the plug-in estimator in terms of rRMSE. 
By contrast, the standard bootstrap exhibits undercoverage and much larger rRMSE. Its shorter intervals reflect an underestimation of variance rather than improved performance. Small coverage differences should be interpreted in view of Monte Carlo variability.

\begin{table}[htbp]
\centering
\small
\caption{Coverage by dependence strength under the block design model for $n=800$.}
\label{tab:bootstrap-coverage-strength}
\begin{tabular}{ccccc}
\toprule
$\theta$ & $\xi'$ & Plug-in & $m$-out-of-$n$ & Standard bootstrap \\
\midrule
0.25 & 0.0625 & 0.957 & 0.944 & 0.884 \\
0.50 & 0.2500 & 0.949 & 0.942 & 0.793 \\
0.75 & 0.5625 & 0.932 & 0.925 & 0.711 \\
\bottomrule
\end{tabular}
\end{table}
Table~\ref{tab:bootstrap-coverage-strength} compares the three methods across dependence strengths at $n=800$. The plug-in method performs best overall, with coverage remaining closest to the nominal level and the most accurate variance estimation. The $m$-out-of-$n$ procedure follows closely, supporting its use as a viable empirical alternative. By contrast, the standard bootstrap fails substantially as the signal becomes stronger.

\subsection{Stability under label permutations}
\label{subsec:perm-stability}

A critical limitation of applying integer-based rank correlations (such as Chatterjee's $\xi_n$) to categorical data is their sensitivity to the arbitrary mapping of nominal categories to integers. To demonstrate this, we generated a single dataset using the block design ($\theta=0.5$, $n=100$, $k=6$) and computed both $\xi_n'$ and $\xi_n$ for all possible $6! = 720$ permutations of the class labels (See Eq.~\eqref{eq:chatterjee-xin} for $\xi_n$).
Figure~\ref{fig:xi-perm-hist} illustrates the results. 

The proposed $\xi_n'$ is coding-invariant, yielding a constant value of $0.16$ (pink solid line) regardless of the labeling. 

In contrast, $\xi_n$ exhibits significant variability across permutations, with values ranging from approximately $0.05$ to $0.31$ (histogram). This structural failure stems from defining rank distances on unordered data. Whether groups are separated by means or overlap due to variance, the estimator penalizes transitions based on the numerical distance between category codes. The score is arbitrarily maximized when adjacent groups in the $X$-space are assigned adjacent integer codes, and minimized otherwise.

This sensitivity extends beyond mere numerical fluctuations, and affects the validity of hypothesis tests. 
Consider testing the independence between continuous $X$ and categorical $Y$, utilizing either $\xi_n$ or $\xi_n'$. When $\xi_n$ is calculated based on the standard integer coding $\{1, \dots, 6\}$ (blue solid line), the observed $\xi_n$ fails to exceed its corresponding critical value (blue dashed line). 
Consequently, the test fails to reject the null hypothesis despite the presence of dependence (a Type II error).
In contrast, the observed $\xi_n'$ clearly exceeds its critical value (pink dashed line), demonstrating that the proposed method reliably rejects $H_0$ and avoids the pitfalls of arbitrary encoding.
\begin{figure}[tp]
\centering
\includegraphics[width=\linewidth]{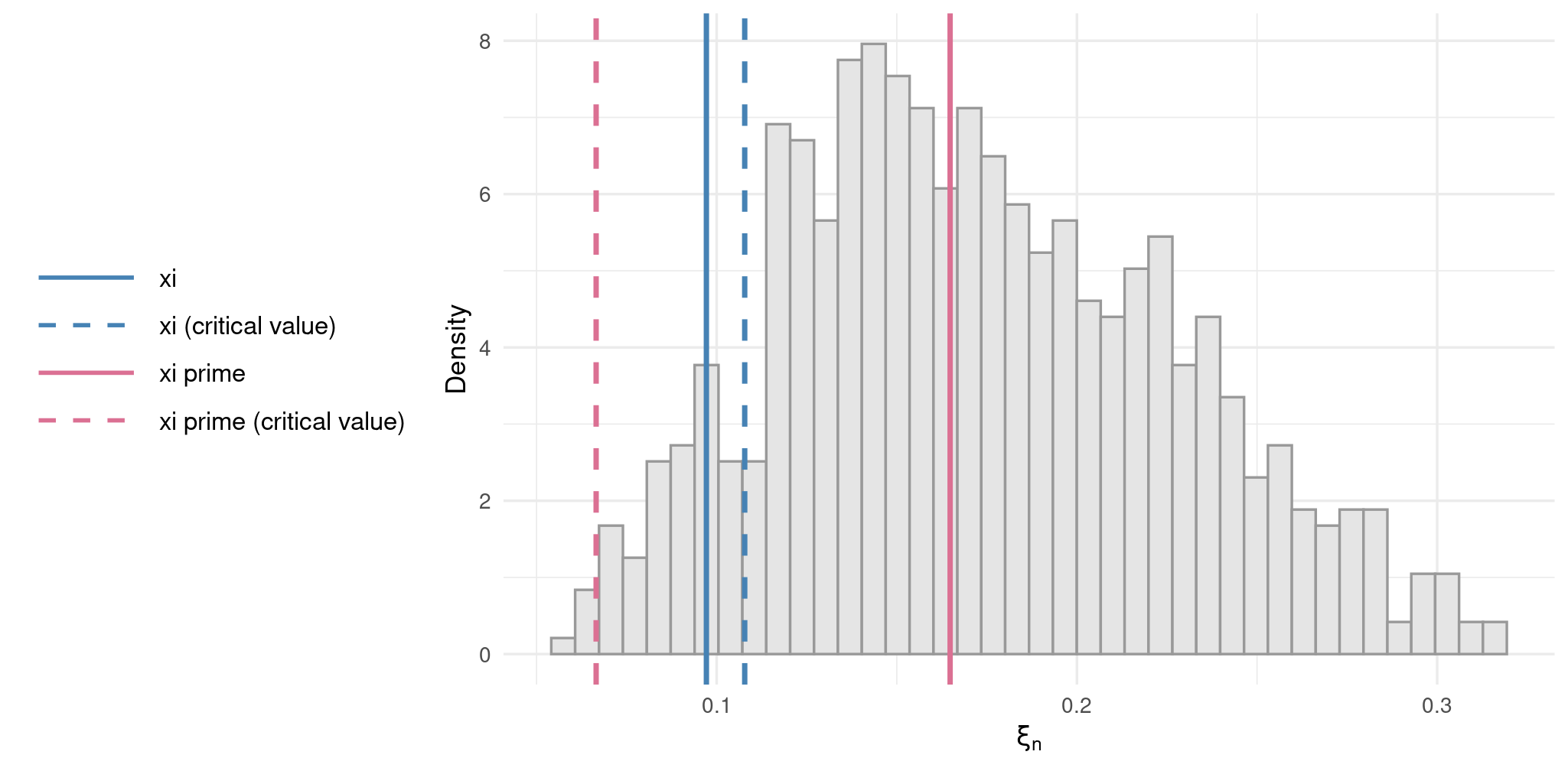}
\caption{Stability comparison under $6!$ label permutations. While the proposed measure $\xi_n'$ (pink solid) remains invariant, $\xi_n$ (histogram) exhibits high variability, crossing the rejection threshold (dashed line) purely due to relabeling.}
\label{fig:xi-perm-hist}
\end{figure}

\subsection{Power and run time comparisons}
\label{subsec:power-comparison}

We benchmark $\xi_n'$ against several widely used measures of dependence between a continuous covariate $X$ and a categorical response $Y$. For clarity, we provide the explicit forms of the statistics:

\begin{itemize}
    \item[1.] \textbf{Distance covariance (dCov)} \citep{SzekelyRizzoBakirov2007}:   
    Let $a_{ij} = |X_i - X_j|$ and $b_{ij} = \|e_{Y_i}-e_{Y_j}\|$, where $e_y=(\mathbf 1\{y=1\},\ldots,\mathbf 1\{y=k\})$ is the standard dummy-variable, or one-hot, representation of category $y$. Define the double-centered distance matrices $D^X=(D^X_{ij})$ and $D^Y=(D^Y_{ij})$ by
    \[
        D^X_{ij}=a_{ij}-a_{i\cdot}-a_{\cdot j}+a_{\cdot\cdot},
        \qquad
        D^Y_{ij}=b_{ij}-b_{i\cdot}-b_{\cdot j}+b_{\cdot\cdot}.
    \]
    The test statistic is defined as the sample squared distance covariance:
    \begin{align}
    \mathrm{dCov}_n^2(X,Y) = \frac{1}{n^2}\sum_{i,j} D^X_{ij}D^Y_{ij}.
    \end{align}

    \item[2.] \textbf{Hilbert--Schmidt independence criterion (HSIC)} \citep{Gretton2005}:   
    Given a Gaussian kernel $k(\cdot,\cdot)$ for $X$ and a delta kernel $\delta(\cdot,\cdot)$ for $Y$ (where $\delta(y, y') = 1$ if $y=y'$ and 0 otherwise), let $K=(K_{ij})$ and $L=(L_{ij})$ denote the $n \times n$ kernel matrices with entries $K_{ij} = k(X_i, X_j)$ and $L_{ij} = \delta(Y_i, Y_j)$. The test statistic is defined as:
    \begin{align}
    \mathrm{HSIC}_n(X,Y) = \frac{1}{n^2}\sum_{i,j} K_{ij}L_{ij} 
      + \frac{1}{n^4}\sum_{i,j,q,r} K_{ij}L_{qr}
      - \frac{2}{n^3}\sum_{i,j,q} K_{ij}L_{iq}.
    \end{align}

    \item[3.] \textbf{Chatterjee’s $\xi_n$} \citep{Chatterjee2021}:   
    To apply this method to the categorical response $Y$, we utilize integer coding by assigning arbitrary values $\{1, \dots, k\}$ to the categories. With the data sorted by $X$ such that $X_{(1)} \le \dots \le X_{(n)}$, define
    \[
    r_i=\sum_{m=1}^n\mathbf 1\{Y_m\le Y_{(i)}\},
    \qquad
    l_i=\sum_{m=1}^n\mathbf 1\{Y_m\ge Y_{(i)}\},
    \]
    where the comparisons use the assigned integer codes. The tie-adjusted test statistic is given by:
    \begin{align}
    \label{eq:chatterjee-xin}
    \xi_n(X,Y)
    =1-\frac{n\sum_{i=1}^{n-1}|r_{i+1}-r_i|}
    {2\sum_{i=1}^n l_i(n-l_i)}.
    \end{align}

    \item[4.] \textbf{ANOVA ($\eta^2$)}:
    Let $n_j$ denote the number of observations in the $j$-th category of $Y$, and let $\bar{X}_j$ denote the sample mean of $X$ within that category. Letting $\bar{X}$ denote the grand mean of $X$, we define the between-group sum of squares $\mathrm{SS}_{\text{bet}} = \sum_{j=1}^k n_j (\bar{X}_j - \bar{X})^2$ and the total sum of squares $\mathrm{SS}_{\text{tot}} = \sum_{i=1}^n (X_i - \bar{X})^2$. The test statistic is defined as:
    \begin{align}
    \eta^2 = \frac{\mathrm{SS}_{\text{bet}}}{\mathrm{SS}_{\text{tot}}}.
    \end{align}
\end{itemize}

While all these statistics are consistent with their respective population counterparts under mild conditions, they differ in computational complexity. 
Specifically, the HSIC requires $O(n^2)$ operations. 
In contrast, the classical ANOVA statistic $\eta^2$ is computationally lightweight, requiring only $O(n)$ operations.
Regarding dCov, although its naive computation is quadratic, \citet{HuoSzekely2016} proposed a fast algorithm that reduces its complexity to $O(n \log n)$. 
Consequently, dCov, $\xi_n$, $\xi_n'$, and $\eta^2$ all scale efficiently as $O(n \log n)$ or better, making them suitable for large-scale applications.

Regarding the testing procedure, distinct approaches are required depending on the statistic. 
Chatterjee's $\xi_n$ and our proposed $\xi_n'$ possess asymptotically normal null distributions, enabling the construction of computationally efficient Wald tests. 
In contrast, the limiting null distributions of dCov and HSIC involve complex infinite sums of weighted chi-square variables, for which analytical critical values are difficult to compute directly. 
Similarly, while the standard test for $\eta^2$ (i.e., ANOVA F-test) relies on the normality assumption, it may not be valid for general non-Gaussian data.
Therefore, to ensure a fair and consistent comparison across all methods without relying on distributional assumptions, we employed the permutation test with $199$ random permutations. 
For our proposed $\xi_n'$, we additionally report the asymptotic Wald test. The proposed version uses the finite-sample correction in $Z_n^{\mathrm{c}}$.

First, we examined empirical size under the null hypothesis ($\theta=0$). At $n=200$, the rejection rates of all methods were close to the nominal level. The centered $\xi_n'$ Wald test also remained well calibrated over a broad range of sample sizes in the separate null-centering experiment reported in Table~\ref{tab:null-negative-centering-response} of Appendix~\ref{sec:appendix_sim}.

\begin{table}[t]
\centering
\caption{Empirical power comparison under the weak signal regime ($\theta=0.25$) across varying sample sizes based on 1,000 replicates.}
\label{tab:power}
\begin{tabular}{c c c c c c c}
\toprule
$n$ & $\xi_n'$ (centered Wald) & $\xi_n'$ (Perm) & $\xi_n$ & HSIC & dCov & $\eta^2$ \\
\midrule
50  & 0.225 & 0.179 & 0.176 & 0.285 & 0.277 & 0.210 \\
100 & 0.355 & 0.306 & 0.243 & 0.579 & 0.569 & 0.430 \\
150 & 0.483 & 0.448 & 0.330 & 0.799 & 0.801 & 0.632 \\
200 & 0.554 & 0.512 & 0.366 & 0.914 & 0.920 & 0.763 \\
\bottomrule
\end{tabular}
\end{table}

Table~\ref{tab:power} summarizes empirical power under the weak signal block design. Power increases steadily with the sample size for all methods. In this setting, HSIC and dCov are the most powerful, followed by the ANOVA-based $\eta^2$. The tests based on $\xi_n'$ have moderate power but consistently outperform Chatterjee's $\xi_n$. Under moderate-to-strong dependence, these differences diminish as all methods approach unit power for sufficiently large samples, as shown in Appendix~\ref{sec:appendix_sim}.
Additional simulations examining sensitivity to the number of response categories are reported in Appendix~\ref{subsec:category_k_sensitivity}.

To complement the block design, we also consider a wiggly block design in which the preferred label changes repeatedly across narrow intervals of the $X$ domain and the same category recurs over several disjoint regions. This setting is designed to examine dependence carried by rapidly alternating local label patterns. While the kernel-based and distance-based methods are less sensitive to the rapidly changing label signal over the range considered, $\xi_n'$ maintains sensitivity by focusing on local label coincidence. A visualization of the design and results for additional sample sizes are provided in Appendix~\ref{sec:appendix_wiggly}.

As shown in Table~\ref{tab:power_wiggly_part}, all methods maintain appropriate rejection rates under the null. As the signal strengthens, the tests based on $\xi_n'$ gain power rapidly and are particularly effective under weak-to-moderate dependence. Chatterjee's $\xi_n$ catches up under stronger signals, whereas HSIC, dCov, and $\eta^2$ remain comparatively insensitive throughout the examined range. This contrast reflects the local adjacency structure of $\xi_n'$, which can capture rapid label changes that are less visible to methods based on global distances, kernel similarities, or mean differences.

\begin{table}[t]
\centering
\caption{Empirical power comparison under the wiggly block design, across varying signal strength based on 1,000 replicates.}
\label{tab:power_wiggly_part}
\begin{tabular}{c c c c c c c c}
\toprule
$n$ & $\theta$ & $\xi_n'$ (centered Wald) & $\xi_n'$ (Perm) & $\xi_n$ & HSIC & dCov & $\eta^2$ \\
\midrule
100 & 0.00 & 0.053 & 0.041 & 0.049 & 0.055 & 0.057 & 0.050 \\
    & 0.25 & 0.210 & 0.172 & 0.160 & 0.041 & 0.041 & 0.045 \\
    & 0.50 & 0.897 & 0.875 & 0.769 & 0.057 & 0.065 & 0.059 \\
    & 0.75 & 1.000 & 1.000 & 0.999 & 0.083 & 0.104 & 0.088 \\
    & 1.00 & 1.000 & 1.000 & 1.000 & 0.080 & 0.149 & 0.088 \\
\bottomrule
\end{tabular}
\end{table}

Table~\ref{tab:cost} reports the average runtime under the implementations and calibration procedures used in this study. The centered $\xi_n'$ Wald test is the fastest procedure across the examined sample sizes because it avoids the repeated computation required by permutation calibration. The permutation-calibrated methods are substantially more expensive, particularly for HSIC and dCov. These differences demonstrate the practical computational benefit of the centered Wald test in this setting.

\begin{table}[t]
\centering
\caption{Average test runtimes (milliseconds) per replicate under the implementations used in this study, averaged across all tested signal strengths ($\theta$). All permutation-calibrated procedures use $199$ random permutations.}
\label{tab:cost}
\begin{tabular}{c c c c c c c}
\toprule
$n$ & $\xi_n'$ (centered Wald) & $\xi_n'$ (Perm) & $\xi_n$ & HSIC & dCov & $\eta^2$ \\
\midrule
50  & 1.43 & 61.67 & 117.81 & 32.60 & 203.97 & 5.51 \\
100 & 1.26 & 67.31 & 128.99 & 70.31 & 329.18 & 5.08 \\
150 & 1.34 & 68.52 & 135.45 & 139.02 & 438.53 & 5.94 \\
200 & 1.31 & 72.09 & 143.24 & 273.90 & 536.97 & 6.22 \\
\bottomrule
\end{tabular}
\end{table}

\section{Application to the Cancer Genome Atlas (TCGA)}
\label{sec:realdata}

\subsection{Dataset description and biological validation}

To validate our proposed measure, we analyzed genomic data from \textit{the Cancer Genome Atlas} (TCGA) Breast Invasive Carcinoma (BRCA) project \citep{CancerGenomeAtlasNetwork2012}. The dataset was accessed via the \texttt{RTCGA} package \citep{RTCGA} in R.

Breast cancer is a heterogeneous disease comprising distinct molecular subtypes characterized by specific gene expression profiles \citep{Perou2000}. Critical biomarkers in this classification include the estrogen receptor (ER) and the progesterone receptor (PR). The \textit{ESR1} gene encodes the estrogen receptor alpha (ER$\alpha$) protein, a primary driver of hormone-receptor-positive breast cancer. Since functional ER signaling typically induces PR expression, these two receptors are closely linked and routinely evaluated together to determine clinical subtypes.
\citet{Perou2000} demonstrated a high concordance between \textit{ESR1} mRNA abundance and clinical ER protein status determined by immunohistochemistry.
This well-established biological relationship serves as a positive control to evaluate whether $\xi_n'$ can correctly identify strong functional dependencies in real-world data.

The objective is to quantify the dependence between the following variables:
\begin{itemize}
    \item Continuous variable ($X$): The mRNA expression level of the \textit{ESR1} gene.
    \item Categorical variable ($Y$): A composite clinical status derived from estrogen receptor (ER) and progesterone receptor (PR) outcomes. Patients were classified into four groups: ER-/PR-, ER-/PR+, ER+/PR-, and ER+/PR+.
\end{itemize}
We matched RNA-Seq samples with clinical records via patient identifiers. After removing samples with missing or indeterminate receptor status, the final dataset comprised $n=1,143$ patients. 
The proportions of the four subtypes were: \texttt{ER+/PR+} ($n=754$, 66.0\%), \texttt{ER-/PR-} ($n=238$, 20.8\%), \texttt{ER+/PR-} ($n=133$, 11.6\%), and \texttt{ER-/PR+} ($n=18$, 1.6\%).

\subsection{Detection of general dependence}
\label{subsec:real_data_analysis}

To evaluate the capability of the proposed measure in detecting general dependence structures, we analyzed the relationship between \textit{ESR1} expression ($X$) and breast cancer subtypes ($Y$) under three distinct scenarios. We considered the original data, mean-centered residuals where group means are aligned to zero, and standardized data where each group is normalized to have zero mean and unit variance. 
Note that removing mean and variance differences does not imply independence. 

Figure~\ref{fig:comparison_measures} and Table~\ref{tab:dependence_comparison} summarize the performance of various dependence measures across these scenarios. 

\begin{figure}[tp!]
    \centering
    \includegraphics[width=0.85\linewidth]{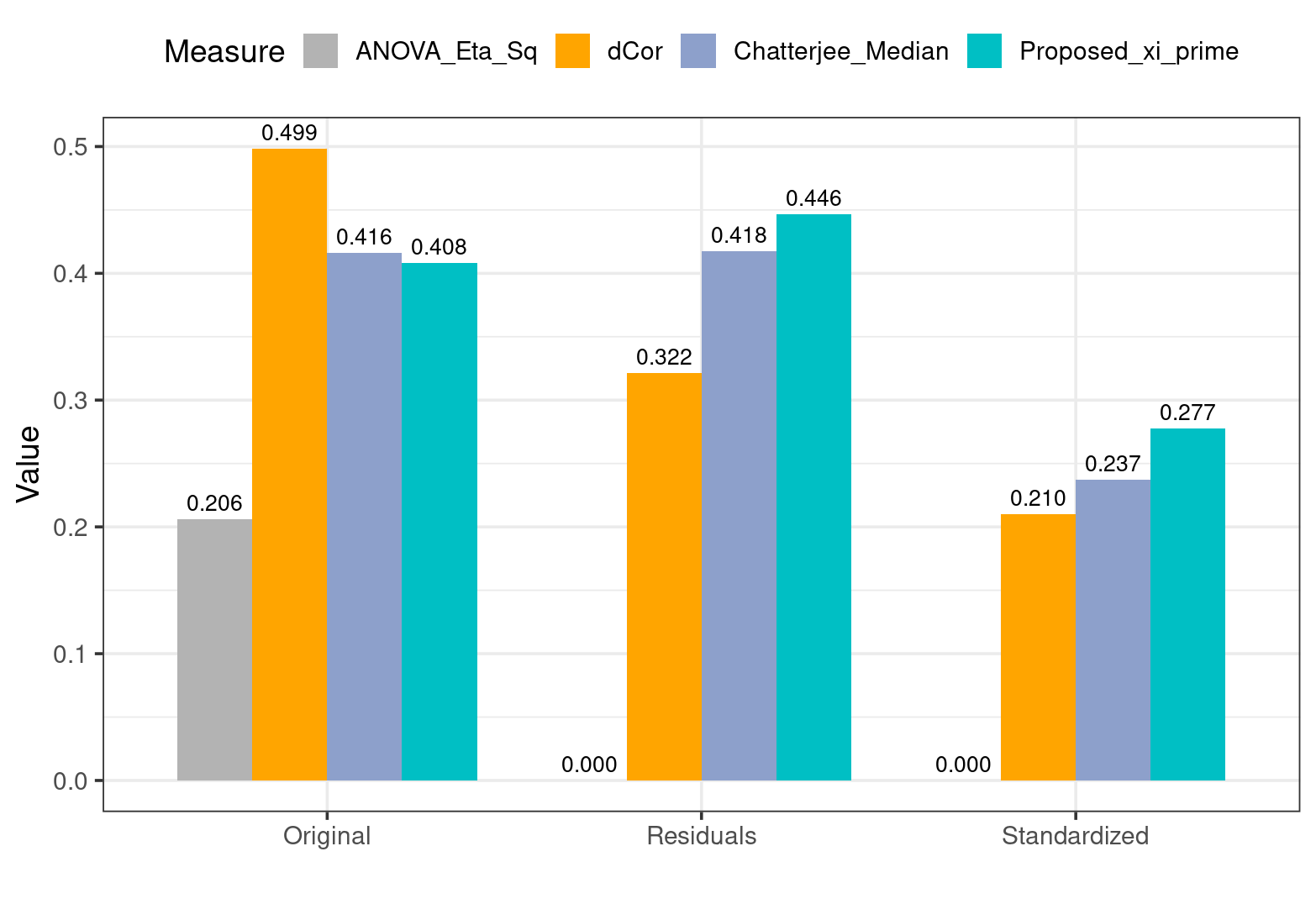}
    \caption{Comparison of dependence measures across three scenarios. While ANOVA ($\eta^2$) collapses to zero when mean differences are removed (Residuals), $\xi_n'$ increases to $0.446$, effectively detecting the variance heterogeneity. Distance correlation (dCor) also detects the signal but shows lower sensitivity to local variance structures than $\xi_n'$.}
    \label{fig:comparison_measures}
\end{figure}

\begin{table}[tp]
\centering
\caption{Comparison of dependence measures across three scenarios. The population target $\xi'$ and the measures $\eta^2$ and dCor lie in $[0,1]$, although the finite-sample estimator $\xi_n'$ can be negative. HSIC is not on a fixed $[0,1]$ scale.}
\label{tab:dependence_comparison}
\begin{tabular}{lccccc}
\toprule
\textbf{Scenario} & \textbf{Proposed} ${\xi_n'}$ & \textbf{ANOVA} ${\eta^2}$ & \textbf{Chatterjee} ${\xi_n}$ & \textbf{dCor} & \textbf{HSIC} \\
\midrule
Original      & 0.408 & 0.206 & 0.416 & 0.499 & 0.019 \\
Residuals     & 0.446 & 0.000 & 0.418 & 0.322 & 0.013 \\
Standardized  & 0.277 & 0.000 & 0.237 & 0.210 & 0.007 \\
\bottomrule
\end{tabular}
\end{table}

In the original data scenario, all methods captured the association. The proposed $\xi_n'$ yielded a score of $0.408$, while Chatterjee's $\xi_n$ (median of 24 values), ANOVA ($\eta^2$), and distance correlation (dCor) produced $0.416$, $0.206$, and $0.499$, respectively. The HSIC statistic was $0.019$; however, since it is not normalized to the $[0, 1]$ range, it is excluded from Figure~\ref{fig:comparison_measures} to preserve the visual scale.

In the residuals scenario, ANOVA ($\eta^2$) resulted in $0.000$, reflecting the absence of mean differences. 
The rank-based measures maintained their strong detection power: $\xi_n'$ increased to $0.446$, and the median of Chatterjee's $\xi_n$ remained comparable at $0.418$. 
In contrast, dCor and HSIC values decreased to $0.322$ and $0.013$, respectively, compared to the original setting. This comparison demonstrates that rank-based correlations are particularly sensitive to variance heterogeneity that persists after group mean centering, whereas traditional mean- or distance-based measures show reduced magnitudes in the absence of mean shifts.

Figure~\ref{fig:combined_analysis} visually clarifies the dependence structure in the residuals scenario. While the residuals for all subtypes share a zero mean, Figure~\ref{fig:sub_density} reveals differences in variance: the \texttt{ER+/PR+} group is widely dispersed, whereas the \texttt{ER-/PR-} group is highly concentrated. Figure~\ref{fig:sub_quantile} confirms this as a structural shift in conditional probability $P(Y=j|X)$. Across the 20 equal-frequency quantiles of ESR1 expression, we observe that the central quantiles are dominated by the \texttt{ER-/PR-} group, while the tails are dominated by the \texttt{ER+/PR+} group. The proposed $\xi_n'$ effectively captures this distributional variation as a strong signal of dependence.

\begin{figure}[tp]
    \centering
    \begin{subfigure}[b]{\linewidth}
        \centering
        \includegraphics[width=0.9\linewidth]{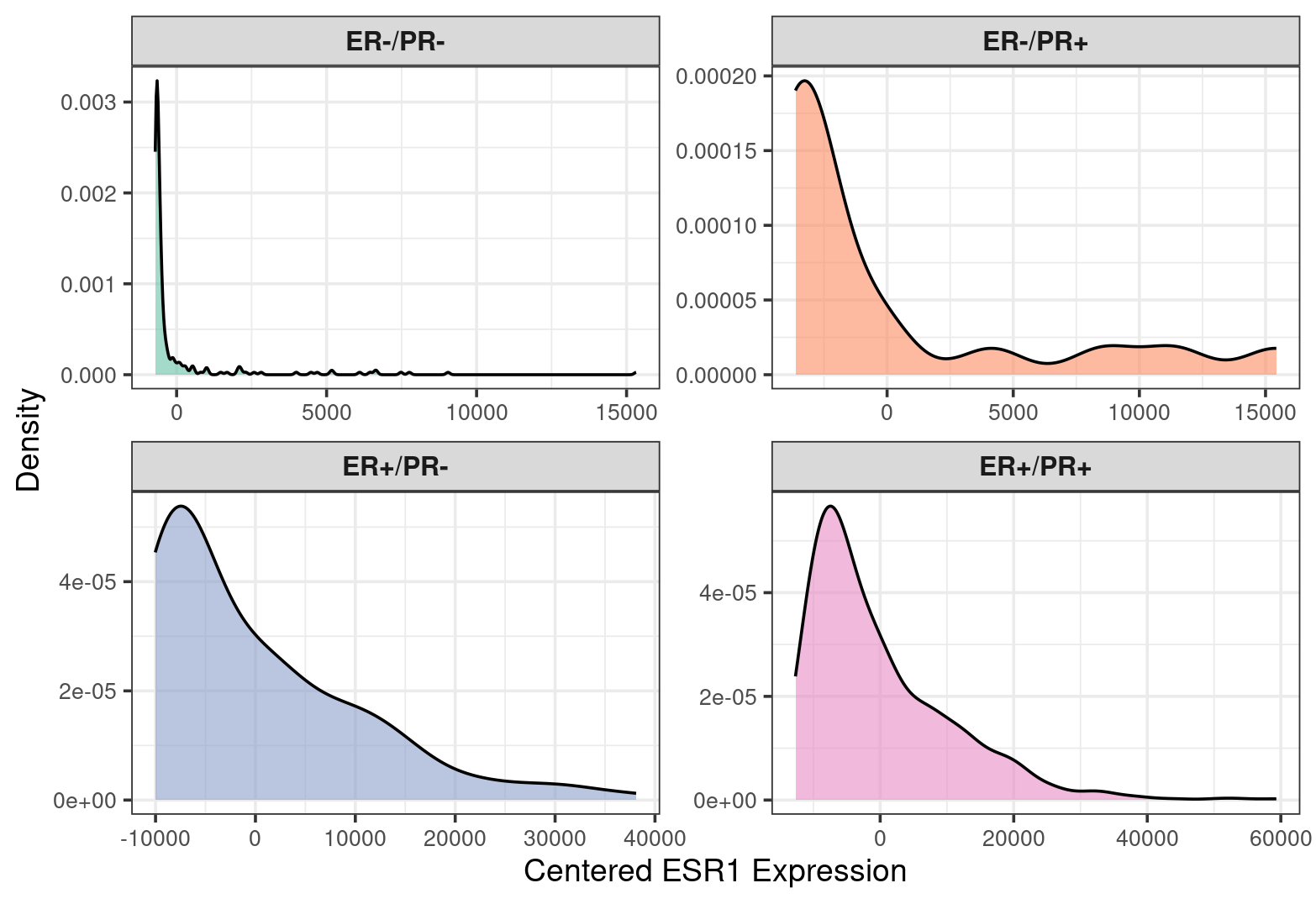}
        \caption{Density plots of mean-centered residuals.}
        \label{fig:sub_density}
    \end{subfigure}
    
    \vspace{0.5cm} 

    \begin{subfigure}[b]{\linewidth}
        \centering
        \includegraphics[width=0.9\linewidth]{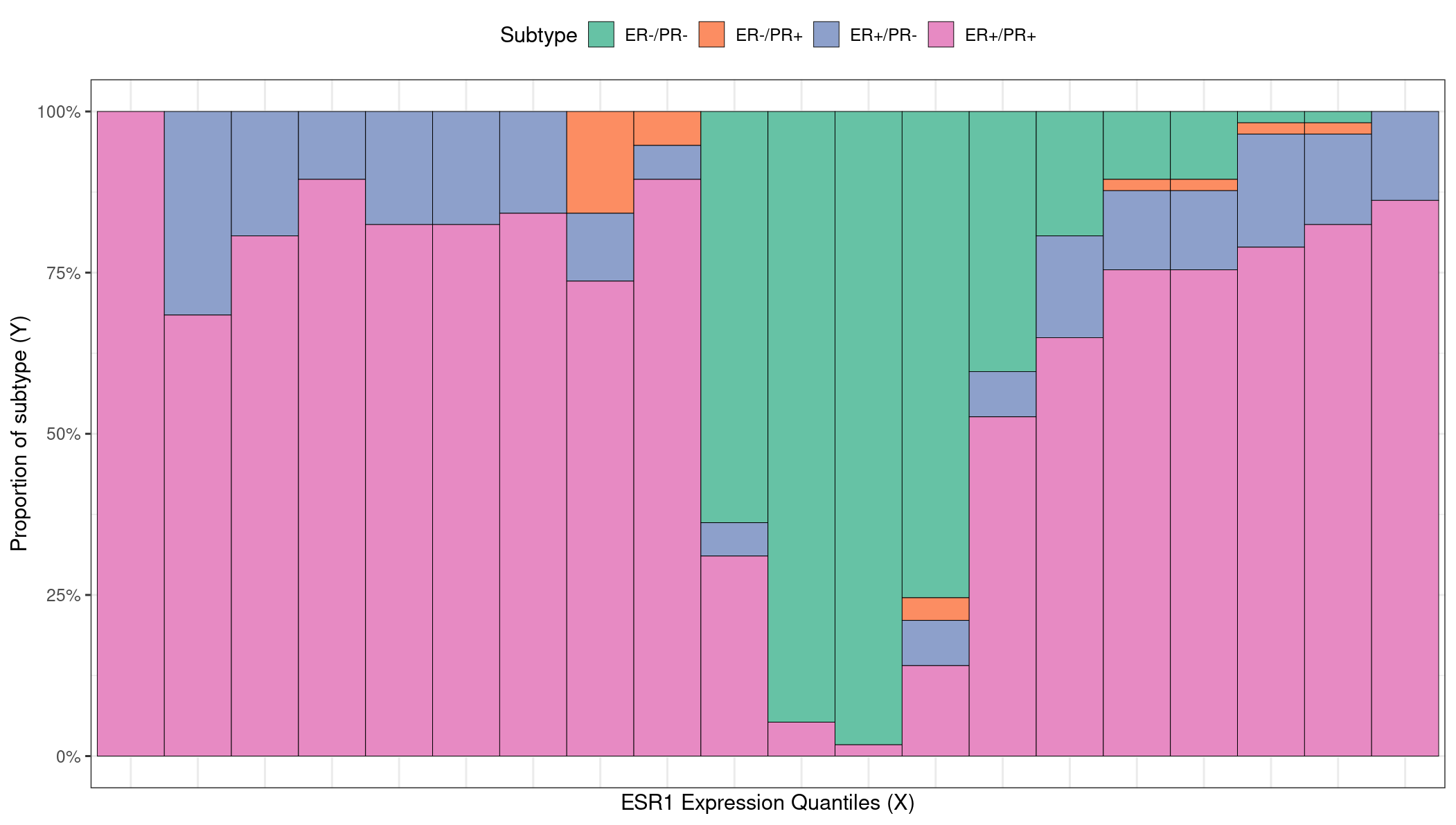}
        \caption{Conditional probability across ESR1 quantiles. The x-axis represents 20 equal-frequency quantiles.}
        \label{fig:sub_quantile}
    \end{subfigure}

    \caption{(a) Density plots show distinct variance differences between subtypes despite identical means. (b) Quantile histograms confirm this pattern as a structural shift in conditional probability.}
    \label{fig:combined_analysis}
\end{figure}

We also identified a critical limitation in applying rank-based correlations like Chatterjee's $\xi_n$ to categorical data: sensitivity to the arbitrary ordering of categories. To quantify this, we computed $\xi_n$ across all $4! = 24$ possible label permutations. As illustrated in Figure~\ref{fig:instability}, while the median of $\xi_n$ in the residuals scenario ($0.418$) is comparable to $\xi_n'$, the values fluctuate wildly between $0.288$ and $0.664$ depending on the encoding. This implies that a researcher could obtain a significantly underestimated or overestimated result purely by chance.
In contrast, our proposed $\xi_n'$ yields a single invariant estimate of 0.446, guaranteeing reproducibility and reliability for nominal data.


In the current analysis with a sufficient sample size ($n=1,143$), Chatterjee's $\xi_n$ rejects the null hypothesis of independence ($H_0: X \perp Y$) across all permutations. However, this rejection is not guaranteed; in settings where $n$ is small, the outcome of the hypothesis test could arbitrarily flip depending on the chosen integer coding. (See our simulation study in Section~\ref{subsec:perm-stability}.)
In contrast, our proposed $\xi_n'$ not only yields a single invariant estimate but also facilitates the construction of confidence intervals as defined in Eq.~\eqref{eq:ci_definition}. The 95\% confidence intervals for the Original, Residuals, and Standardized scenarios are estimated as $[0.354, 0.463]$, $[0.392, 0.501]$, and $[0.220, 0.335]$, respectively. These distinct intervals provide statistical evidence that the dependence strength significantly differs across scenarios.

\begin{figure}[tp]
    \centering
    \includegraphics[width=\linewidth]{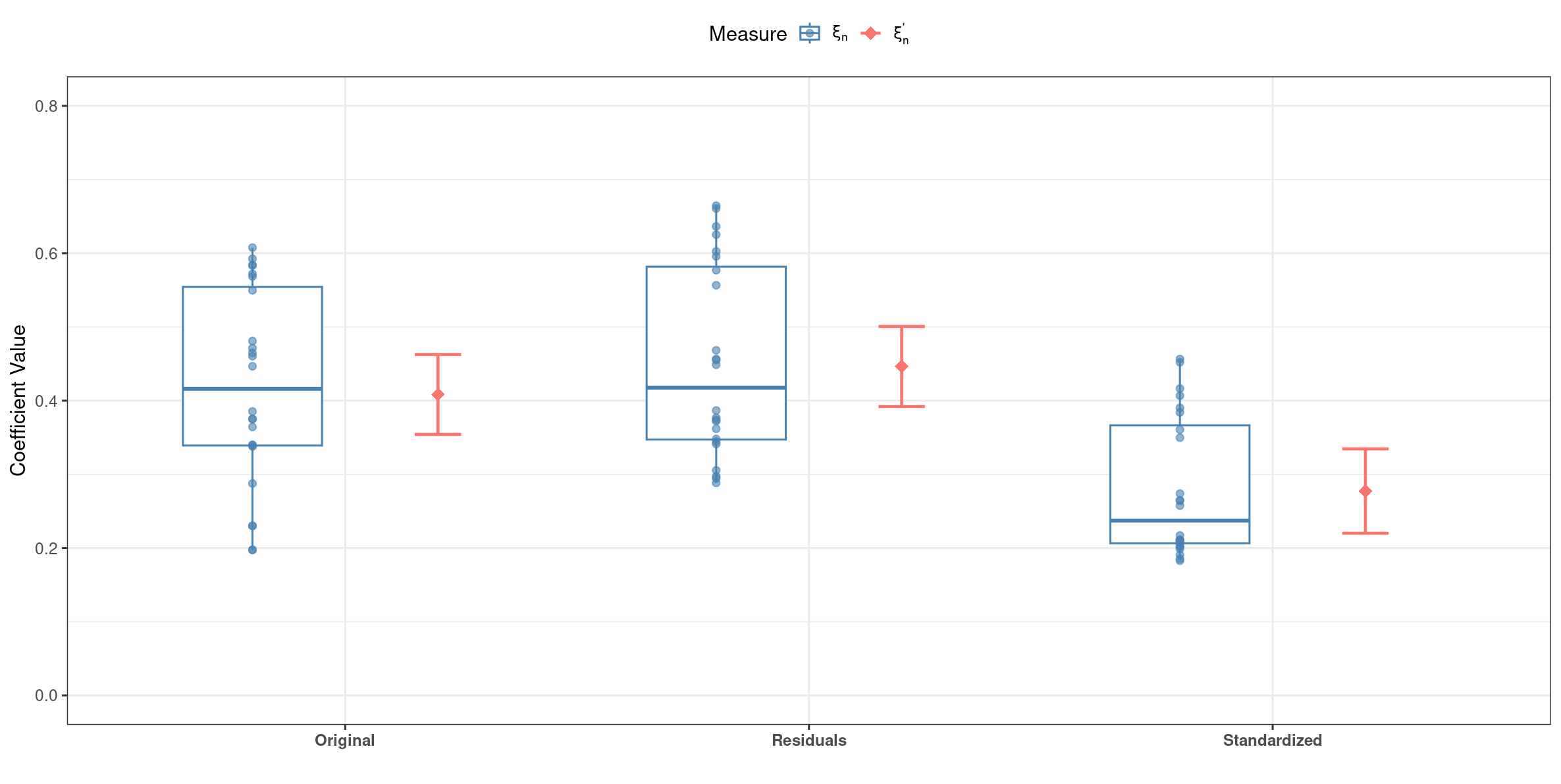}
    \caption{Instability of Chatterjee's $\xi_n$ versus the stability of $\xi_n'$. Blue dots represent $\xi_n$ values computed under all 24 possible label permutations. In the residuals setting, $\xi_n$ ranges from 0.288 to 0.664, demonstrating severe sensitivity to arbitrary coding. In contrast, the red diamond represents the proposed $\xi_n'$ estimate, with the 95\% asymptotic confidence interval.}
    \label{fig:instability}
\end{figure}

To understand the increase in $\xi_n'$ in the residual data compared to the original data, we examine the contribution of each category to the measure $\xi_n'$. In Eq.~\eqref{eq:xisample}, the numerator of $\xi_n'$ can be written as $\sum_{j=1}^k C_j$, where
\[
C_j = \frac{1}{n-1}\sum_{i=1}^{n-1} \mathbf{1}\{Y_{(i)}=j, Y_{(i+1)}=j\} - \left( \frac{1}{n} \sum_{i=1}^n \mathbf{1}\{Y_i=j\} \right)^2.
\]
Intuitively, $C_j$ measures the excess adjacent-match frequency for category $j$ in the $X$-ordered label sequence relative to its empirical marginal coincidence probability. Accordingly, we interpret $C_j$ as the net contribution of category $j$.

Figure~\ref{fig:net_contribution} presents the decomposition of these contributions for each category $j$. In the residuals scenario, we observe distinct behaviors across subgroups: the contribution of the minority group \texttt{ER-/PR-} (20.8\% of the data) increased, whereas that of the majority group \texttt{ER+/PR+} (66.0\% of the data) decreased. Because $\xi_n'$ employs an unweighted aggregation strategy, the local signal from the minority group outweighed the reduction in the majority group, thereby driving the $\xi_n'$ value up.

In the standardized scenario, the number of matches decreased globally across all categories as variance differences were normalized. This resulted in a uniform reduction in all $C_j$'s, leading to the observed decrease in the $\xi_n'$ value.

\begin{figure}[tp]
    \centering
    \includegraphics[width=\linewidth]{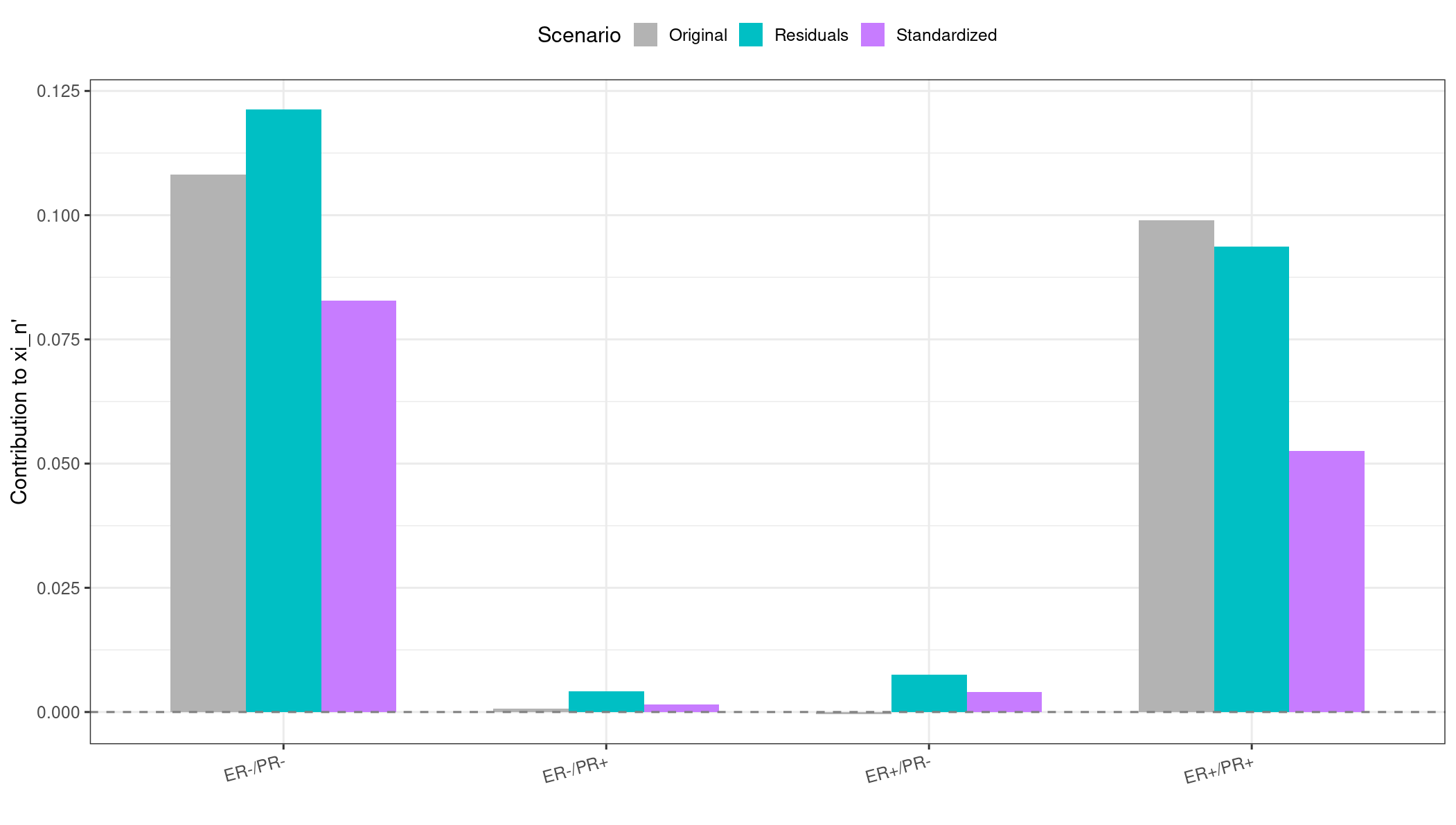}
    \caption{Decomposition of contributions. In the residuals scenario, the contribution of the minority class (\texttt{ER-/PR-}) increases significantly due to its low variance (high local density). The unweighted aggregation of $\xi_n'$ allows this local signal to outweigh the noise from the majority class.}
    \label{fig:net_contribution}
\end{figure}

\subsection{Biological validation}

Finally, we validated the coefficient by identifying the primary source of the dependence based on known biological mechanisms. The \textit{ESR1} gene directly encodes the ER protein, whereas PR expression is a downstream effect regulated by ER. Therefore, \textit{ESR1} expression should exhibit a stronger direct dependence on ER status than on PR status. When calculating $\xi_n'$ for each receptor separately, the coefficient for ER status ($0.670$) is substantially higher than for PR status ($0.419$). This aligns with the biological hierarchy.

Interestingly, the composite variable representing the four ER/PR subtypes yielded $\xi_n' \approx 0.408$, which is lower than the ER-only case. This reduction is statistically consistent with the data structure. We investigated the distributional overlap within the receptor subgroups to understand why stratifying by PR dilutes the signal. We focused our analysis on the ER-positive group ($n=887$), as the ER-negative/PR-positive subgroup was too small ($n=18$) to permit reliable distributional inference.

As illustrated in Figure~\ref{fig:overlap_evidence}, within the ER-positive group, the distributions of \textit{ESR1} expression for PR+ and PR- tumors exhibit substantial overlap. 
Since \textit{ESR1} expression levels do not significantly differ between PR+ and PR- tumors within the same ER group, stratifying by PR introduces label variation that $X$ cannot predict. Consequently, the probability of adjacent matches decreases in the sorted sequence, leading to a lower association. This confirms that $\xi_n'$ correctly reflects the dilution of dependence when sub-categories lacking distinct separability by $X$ are introduced.
\begin{figure}[tp]
    \centering
    \includegraphics[width=\linewidth]{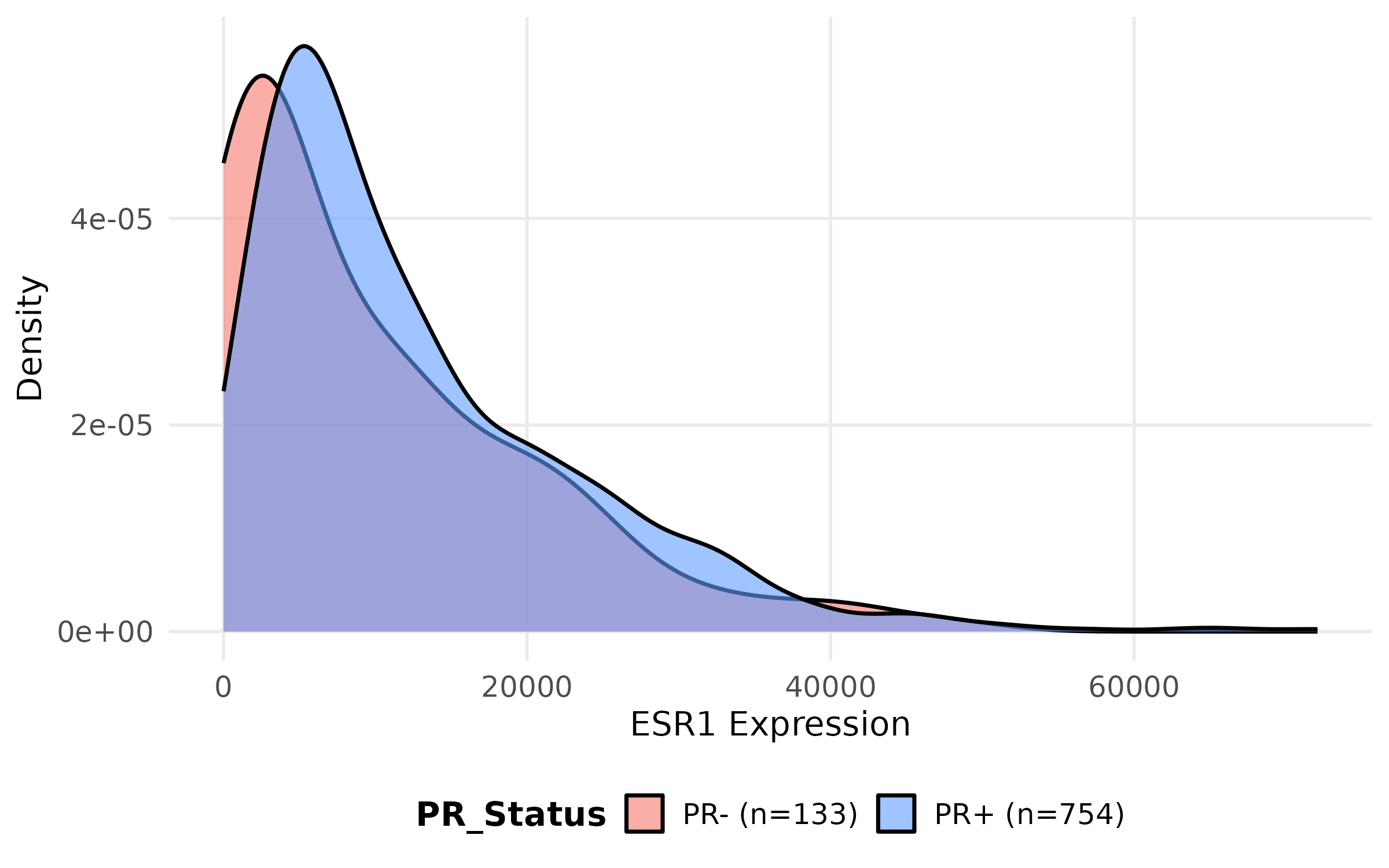}
    \caption{Density plots of \textit{ESR1} expression within the ER-positive group ($n=887$), stratified by PR status. The substantial overlap between PR+ (blue) and PR- (red) distributions explains why adding PR status to the classification dilutes the dependence measured by $\xi_n'$.}
    \label{fig:overlap_evidence}
\end{figure}

\section{Discussion}

The results presented above establish $\xi'$ as a simple coefficient of association between a real-valued $X$ and a categorical $Y$, together with the sample estimator $\xi_n'$. 
At the population level, $\xi'$ is interpretable on a $0$--$1$ scale: it equals $0$ if and only if $X$ and $Y$ are independent, and it equals $1$ if and only if $Y=f(X)$ almost surely. 
The coefficient and its estimator are defined without parametric modeling assumptions, while the studentized inferential procedures account for the distribution-dependent variance through the proposed variance estimators.
Furthermore, $\xi'$ provides a label-invariant nominal-category analogue of Chatterjee's coefficient by employing unweighted aggregation, while $\xi_n'$ retains the efficient computational complexity of $O(n\log n)$.

Our empirical studies highlight three crucial advantages. 
Most notably, $\xi_n'$ guarantees coding invariance. Unlike $\xi_n$, whose value can vary substantially under label permutations because of coding artifacts, $\xi_n'$ remains strictly invariant. This stability makes $\xi_n'$ reliable for applications involving unordered categories.

Beyond stability, $\xi_n'$ offers a balance between power and computational efficiency. It outperforms $\xi_n$ in detection power in the settings considered and provides a competitive alternative to powerful yet time-consuming tests such as dCov and HSIC.
Moreover, in the simulation settings where the centered Wald test was well calibrated, it avoided permutation calibration and was substantially faster than the permutation-based procedures under the implementations used here.

Finally, the asymptotic theory of $\xi_n'$ enables the construction of confidence intervals. 
While most independence measures rely solely on testing to provide a $p$-value, $\xi_n'$ admits a consistent variance estimator that allows researchers to quantify the uncertainty of the estimated dependence. This capability is particularly valuable in scientific contexts where effect size estimation is as critical as significance testing.

Several directions for future work remain. Just as Chatterjee's coefficient has been extended to multivariate settings using graph-based approaches, we anticipate that $\xi_n'$ admits a similar generalization for multivariate $X$. Additionally, developing conditional versions analogous to the Azadkia–Chatterjee coefficient \citep{AzadkiaChatterjee2021} would further broaden its scope to causal discovery applications.

\clearpage

\appendix
\section{Additional examples}
\label{app:additional_examples}

\subsection{Multinomial logit illustration}
\label{app:sim_details}

For the illustrative simulation in Figure~\ref{fig:concentration_gain2}, the conditional probabilities were generated as follows. Let $X \sim U(0,4)$. For each class $j \in \{1, 2, 3\}$, we defined the logit function $z_j(x)$ as a quadratic form centered at $c_j = j$:
\[
z_j(x) = -\gamma (x - c_j)^2 + C,
\]
where $C$ is an arbitrary constant common to all classes and therefore cancels from the softmax probabilities. The conditional probabilities were obtained via the softmax transformation:
\[
P(Y=j \mid X=x) = \frac{\exp(z_j(x))}{\sum_{l=1}^3 \exp(z_l(x))}.
\]
The parameter $\gamma$ controls the concentration of the distribution. We set $\gamma = 0.5$ for the weak dependency scenario and $\gamma = 20$ for the strong dependency scenario.

\subsection{Rare class population comparisons}
\label{app:rare_case_plots}

This section supplements Figure~\ref{fig:example_k4} by displaying the same
rare-class construction for $k=3$ and $k=10$ (Figures~\ref{fig:example_k3} and~\ref{fig:example_k10}). In both cases, the weighted
measure $\xi''$ is markedly smaller than $\xi'$ when the dependence is
concentrated in a low-prevalence class, illustrating the effect of assigning
class contributions in proportion to marginal prevalence.

\begin{figure}[htbp]
\centering
\includegraphics[width=0.85\linewidth]{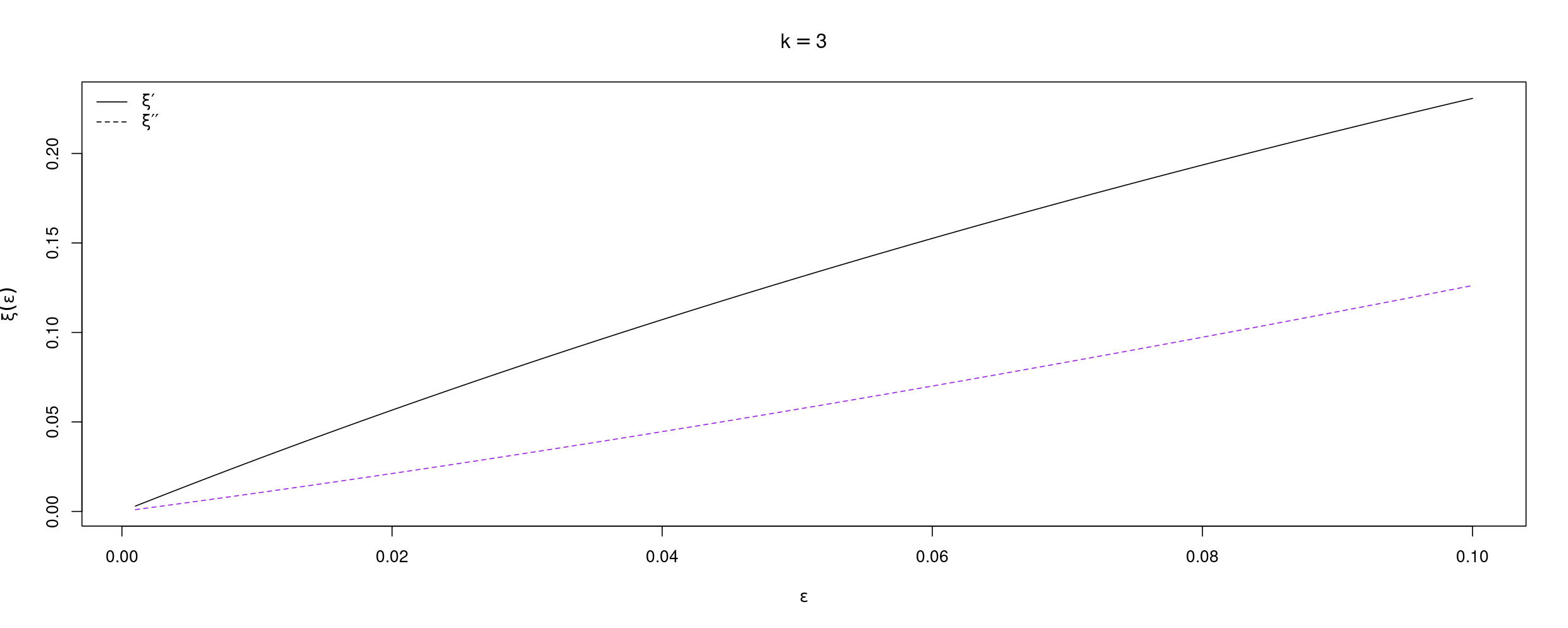}
\caption{
Comparison between $\xi'$ (solid) and $\xi''$ (dashed) in the rare-class
construction with $k=3$. The prevalence-weighted measure $\xi''$ is smaller
because the contribution of the rare class is scaled by its marginal
probability, even though that class is perfectly predicted.
}
\label{fig:example_k3}
\end{figure}

\begin{figure}[htbp]
\centering
\includegraphics[width=0.85\linewidth]{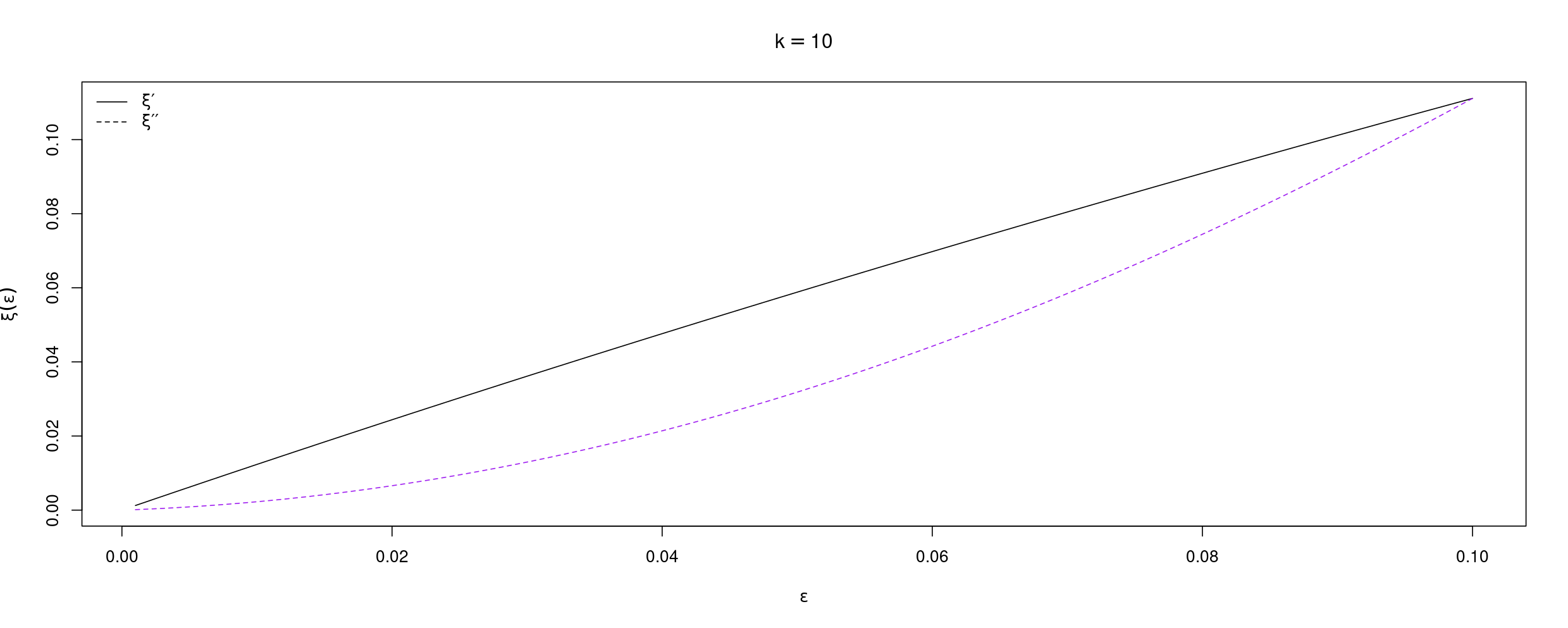}
\caption{
Comparison between $\xi'$ (solid) and $\xi''$ (dashed) in the rare-class
construction with $k=10$. The discrepancy becomes more pronounced as the
number of categories increases, illustrating the effect of prevalence
weighting on the contribution of rare classes as $k$ grows.
}
\label{fig:example_k10}
\end{figure}

\clearpage
\section{Additional finite-sample studies}
\label{app:finite_sample_studies}

Unless otherwise stated, all tests in this section are conducted at significance level $\alpha=0.05$, and empirical rejection probabilities are computed as the proportions of rejections over the stated numbers of independent Monte Carlo replicates. For every permutation-based test, if $T_{\mathrm{obs}}$ is the observed statistic and $T_1^*,\ldots,T_{B_{\mathrm{perm}}}^*$ are the statistics obtained from $B_{\mathrm{perm}}$ random label permutations, we use the plus-one permutation $p$-value
\[
p_{\mathrm{perm}}
=
\frac{1+\sum_{b=1}^{B_{\mathrm{perm}}}
\mathbf 1\{T_b^*\ge T_{\mathrm{obs}}\}}
{B_{\mathrm{perm}}+1}.
\]
The permutations are sampled independently with replacement from the permutation space, and the test rejects when $p_{\mathrm{perm}}\le\alpha$.

When reported, $\xi$ without a prime denotes Chatterjee's population coefficient, whereas $\xi_n$ denotes its sample counterpart. When HSIC is included, we use a Gaussian kernel for $X$, with bandwidth equal to the sample median of the positive pairwise distances and held fixed across label permutations, and a delta kernel for $Y$.

\subsection{Dominant noise, rare signal design}
\label{app:dominant_noise_sim}

This appendix supplements Example~\ref{ex:dominant_noise} with a finite sample power comparison in its rare signal regime. For the simulations, we generated $X\sim\operatorname{Unif}(0,1)$ and set
\[
I_j=((j-1)\delta,j\delta],
\qquad j=1,\ldots,k-1.
\]
On each interval $I_j$, the rare class component is concentrated on class $j$, while the dominant background class still occurs with probability $1-(k-1)\epsilon$. Hence $Y$ is not deterministic given $X$, but the rare class signal is locally structured.

We then compared permutation tests based on $\xi_n'$, the prevalence-weighted statistic $\xi_n''$, and Chatterjee's rank coefficient $\xi_n$ applied to the numerical coding $Y\in\{1,\ldots,k\}$. Chatterjee's statistic is included only as a representative ordinal-coding comparison, since for nominal categorical $Y$ its value depends on the numerical coding of the categories.

For completeness, let
\[
\widehat A_{j,n}
=
\frac{1}{n-1}\sum_{i=1}^{n-1}
\mathbf 1\{Y_{(i)}=j,\;Y_{(i+1)}=j\}.
\]
The prevalence-weighted sample statistic used in this comparison is
\[
\xi_n''
=
\frac{\displaystyle
\sum_{j=1}^k \widehat p_j
\bigl(\widehat A_{j,n}-\widehat p_j^2\bigr)}
{\displaystyle
\sum_{j=1}^k \widehat p_j^2(1-\widehat p_j)}.
\]

For the numerical experiment, we set $k=4$ and $\delta=0.30$. We considered
\[
\epsilon\in\{0.03,0.08,0.15,0.20\}
\]
and
\[
n\in\{300,600,900,1200\}.
\]
For each configuration, empirical power was computed from 300 Monte Carlo replicates using 399 random permutations.

\begin{table}[htbp]
\centering
\small
\setlength{\tabcolsep}{4.5pt}
\caption{Permutation-test power comparison under the dominant noise, rare signal design. We use $k=4$, $\delta=0.30$, 300 Monte Carlo replicates, 399 permutations, and significance level $\alpha=0.05$. Population coefficients and powers are rounded to two decimal places.}
\label{tab:dominant-noise-power-epsilon}
\begin{tabular}{cccccccc}
\toprule
$\epsilon$
& $n$
& $\xi'$
& $\xi''$
& $\xi$
& Power of $\xi_n'$
& Power of $\xi_n''$
& Power of $\xi_n$ \\
\midrule
0.03 & 300 & 0.03 & 0.00 & 0.00 & 0.13 & 0.05 & 0.04 \\
0.03 & 600 & 0.03 & 0.00 & 0.00 & 0.27 & 0.09 & 0.07 \\
0.03 & 900 & 0.03 & 0.00 & 0.00 & 0.30 & 0.08 & 0.08 \\
0.03 & 1200 & 0.03 & 0.00 & 0.00 & 0.31 & 0.06 & 0.05 \\
\midrule
0.08 & 300 & 0.09 & 0.02 & 0.01 & 0.56 & 0.08 & 0.06 \\
0.08 & 600 & 0.09 & 0.02 & 0.01 & 0.83 & 0.12 & 0.09 \\
0.08 & 900 & 0.09 & 0.02 & 0.01 & 0.94 & 0.17 & 0.12 \\
0.08 & 1200 & 0.09 & 0.02 & 0.01 & 0.95 & 0.20 & 0.13 \\
\midrule
0.15 & 300 & 0.19 & 0.09 & 0.07 & 0.99 & 0.64 & 0.34 \\
0.15 & 600 & 0.19 & 0.09 & 0.07 & 1.00 & 0.91 & 0.59 \\
0.15 & 900 & 0.19 & 0.09 & 0.07 & 1.00 & 0.93 & 0.68 \\
0.15 & 1200 & 0.19 & 0.09 & 0.07 & 1.00 & 0.98 & 0.79 \\
\midrule
0.20 & 300 & 0.30 & 0.23 & 0.16 & 1.00 & 1.00 & 0.94 \\
0.20 & 600 & 0.30 & 0.23 & 0.16 & 1.00 & 1.00 & 0.99 \\
0.20 & 900 & 0.30 & 0.23 & 0.16 & 1.00 & 1.00 & 1.00 \\
0.20 & 1200 & 0.30 & 0.23 & 0.16 & 1.00 & 1.00 & 1.00 \\
\bottomrule
\end{tabular}
\end{table}

The results confirm the population level comparison. When the rare class signal is weak, the population value of $\xi''$ is much smaller than that of $\xi'$, and the corresponding permutation test has substantially lower power. For example, when $\epsilon=0.08$, the population values are
\[
\xi'=0.086,
\qquad
\xi''=0.018,
\qquad
\xi=0.012,
\]
and the power of $\xi_n'$ is already $0.830$ at $n=600$, whereas the powers of $\xi_n''$ and $\xi_n$ are $0.117$ and $0.090$, respectively. As $\epsilon$ increases, the rare classes are less down-weighted, and the prevalence weighted test eventually catches up.

\subsection{Minority- versus majority-class signals}
\label{app:signal_location_sim}

To compare minority and majority class signals directly, we considered $k=4$ with marginal class probabilities
\[
(p_1,p_2,p_3,p_4)
=
\left(0.60,\frac{0.40}{3},\frac{0.40}{3},\frac{0.40}{3}\right).
\]
Let $c$ denote the signal class. For $X\sim\operatorname{Unif}(0,1)$, define the structured conditional probabilities
\[
h_c(x)=\mathbf 1\{x\le p_c\},
\qquad
h_j(x)=\frac{p_j}{1-p_c}\mathbf 1\{x>p_c\},
\quad j\ne c.
\]
For $x\le p_c$, this structured model assigns class $c$ deterministically. For $x>p_c$, it assigns only the remaining classes in proportions chosen to preserve their marginal probabilities. Indeed, $E\{h_j(X)\}=p_j$ for every $j$. We interpolate between independence and this structured model by setting
\[
P(Y=j\mid X=x)=(1-\lambda)p_j+\lambda h_j(x),
\qquad 0\le\lambda\le1,
\]
where $\lambda$ controls the signal strength. For each choice of the signal-bearing class, we selected $\lambda$ so that the population target $\xi'$ equaled $0.02$, $0.10$, or $0.30$. We took $c=1$ for the majority-class signal and $c=2$ for the minority-class signal. For each configuration, empirical power was estimated from $2{,}000$ Monte Carlo replicates at $n=300$, using $399$ random permutations.

\begin{table}[htbp]
\centering
\small
\setlength{\tabcolsep}{5pt}
\caption{Permutation test power by signal location at matched values of the unweighted population target $\xi'$. The majority class has prevalence $0.60$; each minority class has prevalence $0.40/3$.}
\label{tab:signal-location-power}
\begin{tabular}{lcccc}
\toprule
Signal-bearing class
& $\xi'$
& $\xi''$
& Power of $\xi_n'$
& Power of $\xi_n''$ \\
\midrule
Majority & 0.02 & 0.030 & 0.105 & 0.141 \\
Majority & 0.10 & 0.149 & 0.813 & 0.906 \\
Majority & 0.30 & 0.447 & 1.000 & 1.000 \\
\midrule
Minority & 0.02 & 0.017 & 0.113 & 0.100 \\
Minority & 0.10 & 0.086 & 0.769 & 0.559 \\
Minority & 0.30 & 0.259 & 1.000 & 0.999 \\
\bottomrule
\end{tabular}
\end{table}

At the moderate signal level $\xi'=0.10$, prevalence weighting improves power when the majority class carries the signal ($0.906$ versus $0.813$), but both tests remain highly sensitive. When a minority class carries a signal of the same unweighted population strength, the direction reverses and the power advantage of $\xi_n'$ is larger ($0.769$ versus $0.559$). Neither test is highly sensitive at the weakest level, and both are essentially fully powered at the strongest level. Together with the exact equality $\xi'=\xi''$ under uniform class probabilities, these results show how the practical effect of aggregation depends on the location and strength of the signal rather than establishing uniform dominance of either target.

\subsection{Additional block-design results}
\label{sec:appendix_sim}

This appendix provides simulation results under the block design model for signal strengths not reported in the main text.
Because $\xi_n'$ is an affine function of the integer-valued number of runs conditional on the class counts, its permutation distribution is discrete and contains many ties. Even distinct permutations therefore often produce the same run count. The resulting nonrandomized upper-tail test can be slightly conservative under the null. In a separate null experiment with $10{,}000$ replicates, the rejection rates of the $\xi_n'$ permutation test were $0.0383$, $0.0452$, $0.0450$, and $0.0519$ for $n=200,500,1000$, and $2000$, respectively; the corresponding mean proportions of permuted statistics tied with the observed value decreased from $0.0545$ to $0.0167$.
Table~\ref{tab:type1} presents the empirical Type I error rates ($\theta=0$) across varying sample sizes, with the proposed Wald test computed from $Z_n^{\mathrm{c}}$. Table~\ref{tab:null-negative-centering-response} separately documents the effect of the finite-sample correction by comparing the original and centered Wald tests over a wider range of sample sizes.
Table~\ref{tab:power_supp} summarizes the power for moderate to strong signals ($\theta \in \{0.50, 0.75\}$).
As expected, as the signal strength increases, the power of all methods converges to 1.
For $\theta=0.50$, the centered $\xi_n'$ Wald test reaches unit power at $n=150$. For $\theta=0.75$, all methods have power of at least $0.990$ at $n=50$ and attain unit power from $n=100$ onward.

\begin{table}[htbp]
\centering
\caption{Empirical Type I error rates ($\theta=0$) at $\alpha=0.05$ based on 1,000 replicates. The Wald column reports the centered $\xi_n'$ test based on $Z_n^{\mathrm{c}}$.}
\label{tab:type1}
\begin{tabular}{c c c c c c c}
\toprule
$n$ & $\xi_n'$ (centered Wald) & $\xi_n'$ (Perm) & $\xi_n$ & HSIC & dCov & $\eta^2$ \\
\midrule
50  & 0.054 & 0.036 & 0.048 & 0.055 & 0.055 & 0.054 \\
100 & 0.063 & 0.045 & 0.053 & 0.053 & 0.047 & 0.041 \\
150 & 0.045 & 0.034 & 0.051 & 0.058 & 0.051 & 0.061 \\
200 & 0.043 & 0.032 & 0.053 & 0.044 & 0.049 & 0.053 \\
\bottomrule
\end{tabular}
\end{table}

\begin{table}[htbp]
\centering
\small
\setlength{\tabcolsep}{4pt}
\caption{Finite-sample null centering and empirical rejection rates under independence for $k=6$, based on $50{,}000$ Monte Carlo replicates. The conditional target $-1/(n-1)$ is the exact conditional mean on the event $\{B_n<1\}$. The original Wald test is based on $\sqrt n\,\xi_n'/\widehat\kappa$, whereas the centered Wald test is based on $\sqrt n\{\xi_n'+1/(n-1)\}/\widehat\kappa$.}
\label{tab:null-negative-centering-response}
\begin{tabular}{ccccc}
\toprule
$n$ & \shortstack{Conditional\\target} & \shortstack{Empirical\\mean} & \shortstack{Original Wald\\test} & \shortstack{Centered Wald\\test} \\
\midrule
200  & $-0.005025$ & $-0.004973$ & $0.0373$ & $0.0509$ \\
500  & $-0.002004$ & $-0.001989$ & $0.0416$ & $0.0518$ \\
1000 & $-0.001001$ & $-0.001018$ & $0.0440$ & $0.0511$ \\
2000 & $-0.000500$ & $-0.000528$ & $0.0446$ & $0.0492$ \\
\bottomrule
\end{tabular}
\end{table}

\begin{table}[htbp]
\centering
\caption{Empirical power for moderate ($\theta=0.50$) and strong ($\theta=0.75$) signals based on 1,000 replicates.}
\label{tab:power_supp}
\begin{tabular}{c c c c c c c c}
\toprule
$\theta$ & $n$ & $\xi_n'$ (centered Wald) & $\xi_n'$ (Perm) & $\xi_n$ & HSIC & dCov & $\eta^2$ \\
\midrule
0.50 & 50  & 0.885 & 0.855 & 0.750 & 0.906 & 0.913 & 0.767 \\
     & 100 & 0.990 & 0.987 & 0.935 & 0.999 & 0.999 & 0.973 \\
     & 150 & 1.000 & 1.000 & 0.992 & 1.000 & 1.000 & 1.000 \\
     & 200 & 1.000 & 1.000 & 0.993 & 1.000 & 1.000 & 1.000 \\
\midrule
0.75 & 50  & 1.000 & 1.000 & 0.999 & 1.000 & 1.000 & 0.998 \\
     & 100 & 1.000 & 1.000 & 1.000 & 1.000 & 1.000 & 1.000 \\
     & 150 & 1.000 & 1.000 & 1.000 & 1.000 & 1.000 & 1.000 \\
     & 200 & 1.000 & 1.000 & 1.000 & 1.000 & 1.000 & 1.000 \\
\bottomrule
\end{tabular}
\end{table}

\subsubsection{Sensitivity to the number of response categories}
\label{subsec:category_k_sensitivity}

To assess sensitivity to the number of response categories, we generalized the block design to $k\in\{3,10,20,30\}$ equal-width blocks and categories. For each $k$, the preferred labels of the blocks were assigned by a fixed random permutation of $\{1,\ldots,k\}$ to avoid aligning the numerical codes with the block order. The marginal class probabilities are uniform, and the signal is calibrated so that the population coefficient satisfies $\xi'=\theta^2$ for every $k$. Table~\ref{tab:power_by_category_k} reports the results for $n=300$. Under independence, the centered Wald test maintains rejection rates close to the nominal level across all examined values of $k$. Under the weak signal $\theta=0.25$, dCov and HSIC are more powerful for $k=3$ and $k=10$, whereas the centered $\xi_n'$ Wald test is more powerful for $k=20$ and $k=30$. Thus, no method uniformly dominates the others. In this balanced block design, the increasing power of the centered Wald test is consistent with the fact that $\xi'=\theta^2$ remains fixed while its null variance is $\kappa^2=1/(k-1)$ under uniform marginal class probabilities. This pattern is specific to the design considered here and is not intended as a general monotonicity claim for arbitrary dependence structures.

\begin{table}[htbp]
\centering
\small
\setlength{\tabcolsep}{4pt}
\caption{Empirical rejection rates across different numbers of response categories under the $k$-block model. Results are based on $1{,}000$ Monte Carlo replicates with $n=300$ and $\alpha=0.05$. Each permutation-based test uses $199$ random permutations, and $\xi'=\theta^2$ for every $k$.}
\label{tab:power_by_category_k}
\begin{tabular}{cc ccccc}
\toprule
$\theta$ & $k$ & \shortstack{$\xi_n'$\\(centered Wald)} & \shortstack{$\xi_n'$\\(Perm.)} & \shortstack{Chatterjee's\\$\xi_n$} & dCov & HSIC \\
\midrule
0    & 3  & 0.048 & 0.047 & 0.054 & 0.048 & 0.047 \\
     & 10 & 0.048 & 0.039 & 0.045 & 0.045 & 0.047 \\
     & 20 & 0.047 & 0.035 & 0.046 & 0.048 & 0.045 \\
     & 30 & 0.058 & 0.045 & 0.049 & 0.047 & 0.052 \\
\midrule
0.25 & 3  & 0.456 & 0.415 & 0.370 & 0.990 & 0.995 \\
     & 10 & 0.873 & 0.858 & 0.472 & 0.982 & 0.978 \\
     & 20 & 0.968 & 0.956 & 0.483 & 0.942 & 0.910 \\
     & 30 & 0.991 & 0.982 & 0.468 & 0.893 & 0.846 \\
\midrule
0.50 & 3  & 1.000 & 1.000 & 1.000 & 1.000 & 1.000 \\
     & 10 & 1.000 & 1.000 & 1.000 & 1.000 & 1.000 \\
     & 20 & 1.000 & 1.000 & 1.000 & 1.000 & 1.000 \\
     & 30 & 1.000 & 1.000 & 1.000 & 1.000 & 1.000 \\
\bottomrule
\end{tabular}
\end{table}

\subsection{Wiggly block design}
\label{sec:appendix_wiggly}

To further investigate the local sensitivity of $\xi_n'$, we let $X\sim\operatorname{Unif}(0,1)$ and consider a wiggly block design with $k=6$ categories and $L=40$ equal-width blocks. For $x\in(0,1)$, let
\[
b(x)=\min\{\lfloor Lx\rfloor+1,L\},
\qquad
j_{b(x)}=1+\{b(x)-1\}\bmod k.
\]
Conditional on $X=x$, the preferred category $j_{b(x)}$ is assigned probability
\[
P\{Y=j_{b(x)}\mid X=x\}
=
\theta+\frac{1-\theta}{k},
\]
and every other category is assigned probability $(1-\theta)/k$. Thus, the preferred category rotates cyclically across successive blocks, creating a rapidly oscillating dependence structure.
Figure~\ref{fig:wiggly_viz} visualizes this data distribution at the deterministic limit ($\theta=1$).

\begin{figure}[htbp]
    \centering
    \includegraphics[width=\linewidth]{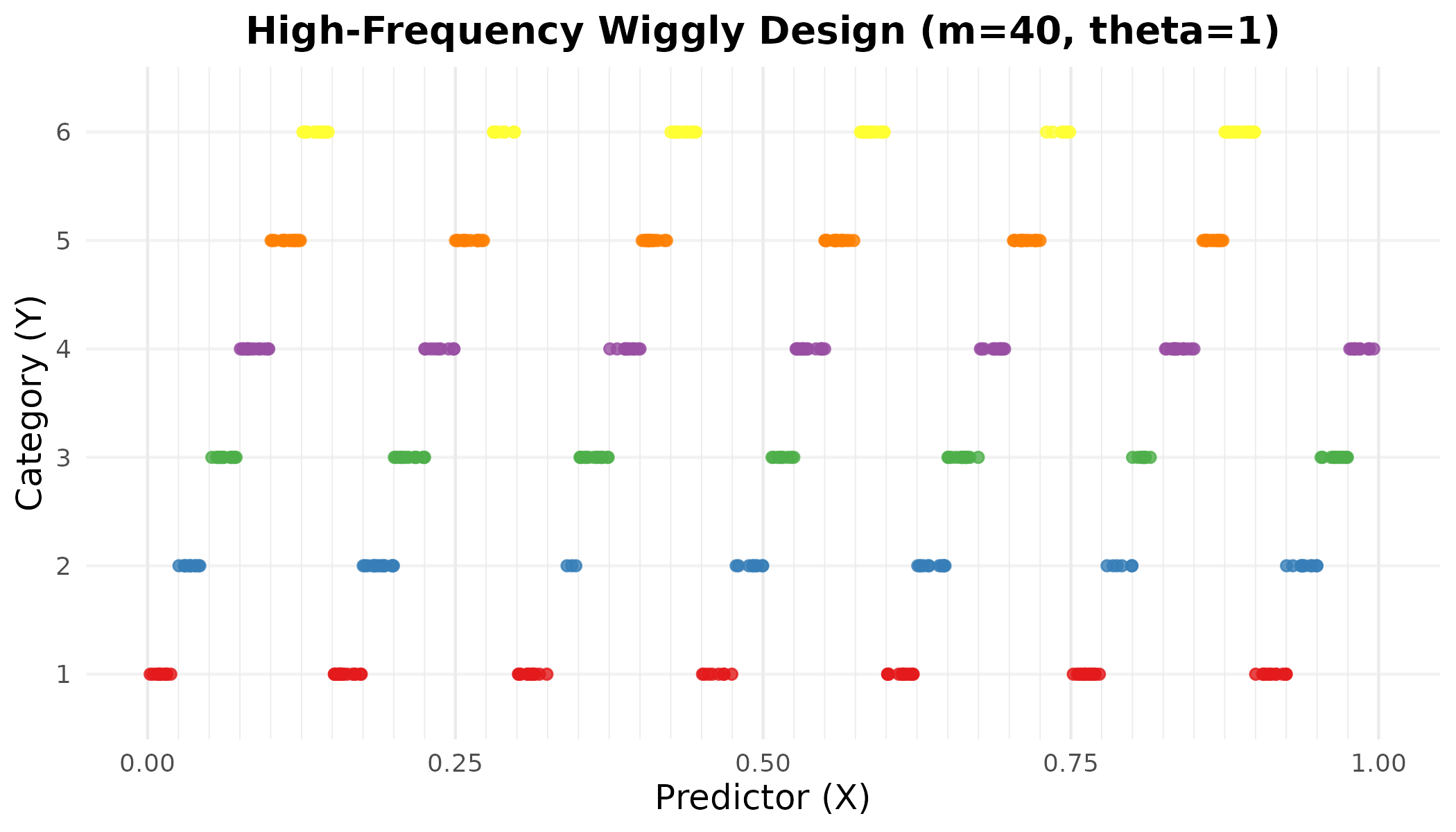}
    \caption{Visualization of the high-frequency wiggly block design ($L=40, \theta=1, n=500$). The vertical lines represent the block boundaries. Within each interval of length $1/40$, $Y$ is fixed to a specific category, but these labels rotate rapidly across the full range of $X$.}
    \label{fig:wiggly_viz}
\end{figure}

As shown in Table~\ref{tab:power_wiggly}, HSIC and dCov remain substantially less powerful than $\xi_n'$ in this design. For example, at $n=200$ and $\theta=1$, their powers are $0.130$ and $0.279$, respectively, whereas both versions of the proposed test have power one. This pattern is consistent with the global distance scales and fixed kernel bandwidth being too coarse for the rapid oscillations used here.

In contrast, both versions of the $\xi_n'$ test have power one at $\theta=0.75$ for $n\ge100$ and at $\theta=1$ for every sample size considered. By evaluating only immediate neighbors in the sorted $X$ sequence, $\xi_n'$ focuses on the local shifts in this design.

\begin{table}[p]
\centering
\caption{Empirical rejection probabilities under the wiggly block design. Results are based on 1,000 replicates, with $B_{\mathrm{perm}}=199$ permutations used for each permutation-based test.}
\label{tab:power_wiggly}
\begin{tabular}{c c c c c c c c}
\toprule
$n$ & $\theta$ & $\xi_n'$ (centered Wald) & $\xi_n'$ (Perm) & $\xi_n$ & HSIC & dCov & $\eta^2$ \\
\midrule
50  & 0.00 & 0.053 & 0.038 & 0.065 & 0.049 & 0.047 & 0.040 \\
    & 0.25 & 0.100 & 0.067 & 0.085 & 0.056 & 0.063 & 0.063 \\
    & 0.50 & 0.352 & 0.294 & 0.324 & 0.054 & 0.060 & 0.061 \\
    & 0.75 & 0.844 & 0.792 & 0.817 & 0.054 & 0.059 & 0.049 \\
    & 1.00 & 1.000 & 1.000 & 1.000 & 0.068 & 0.093 & 0.065 \\
\midrule
100 & 0.00 & 0.053 & 0.041 & 0.049 & 0.055 & 0.057 & 0.050 \\
    & 0.25 & 0.210 & 0.172 & 0.160 & 0.041 & 0.041 & 0.045 \\
    & 0.50 & 0.897 & 0.875 & 0.769 & 0.057 & 0.065 & 0.059 \\
    & 0.75 & 1.000 & 1.000 & 0.999 & 0.083 & 0.104 & 0.088 \\
    & 1.00 & 1.000 & 1.000 & 1.000 & 0.080 & 0.149 & 0.088 \\
\midrule
150 & 0.00 & 0.053 & 0.038 & 0.049 & 0.044 & 0.040 & 0.040 \\
    & 0.25 & 0.324 & 0.278 & 0.226 & 0.053 & 0.060 & 0.054 \\
    & 0.50 & 0.986 & 0.983 & 0.921 & 0.061 & 0.081 & 0.074 \\
    & 0.75 & 1.000 & 1.000 & 1.000 & 0.076 & 0.125 & 0.086 \\
    & 1.00 & 1.000 & 1.000 & 1.000 & 0.100 & 0.183 & 0.095 \\
\midrule
200 & 0.00 & 0.047 & 0.038 & 0.044 & 0.069 & 0.067 & 0.060 \\
    & 0.25 & 0.440 & 0.397 & 0.289 & 0.066 & 0.071 & 0.065 \\
    & 0.50 & 0.999 & 0.999 & 0.988 & 0.071 & 0.097 & 0.078 \\
    & 0.75 & 1.000 & 1.000 & 1.000 & 0.081 & 0.135 & 0.085 \\
    & 1.00 & 1.000 & 1.000 & 1.000 & 0.130 & 0.279 & 0.118 \\
\bottomrule
\end{tabular}
\end{table}

\clearpage
\section{Proofs}
\label{app:proofs}

Throughout this section, we work on a fixed probability space $(\Omega, \mathcal{F}, P)$ and use the notation introduced in Section~\ref{sec:xiprime}. Let $X$ and $Y$ be random variables defined on this space, where $Y$ takes values in $\{1, \dots, k\}$.

\subsection{Population properties and representations}

\subsubsection{Proof of Proposition~\ref{prop:basic-xi-prime}}
\label{subsec:proof-basic-properties}
\begin{proof}
Use the notation $p_j$, $g_j$, $h$, $A$, and $B$ introduced in Section~\ref{sec:xiprime}. Define the convex function $\psi:\Delta_k\to\mathbb{R}$ by $\psi(u_1,\dots,u_k)=\sum_{j=1}^k u_j^2$, where $\Delta_k = \{ (u_1, \dots, u_k) \in \mathbb{R}^k \mid u_j \ge 0, \sum_{j=1}^k u_j = 1 \}$.

\medskip\noindent
\textbf{(i) Range $[0,1]$.}
Write
\[
\xi'(X,Y)=\frac{A-B}{1-B}.
\]

\emph{Lower bound.} By Jensen’s inequality applied to $\psi$,
\[
A
=\mathbb{E}[\psi(g_1(X),\dots,g_k(X))]
\;\ge\; \psi\big(\mathbb{E}[g_1(X)],\dots,\mathbb{E}[g_k(X)]\big)
=B,
\]
so the numerator is nonnegative and therefore $\xi'\ge 0$.

\emph{Upper bound.} For each fixed $x$, $\sum_{j=1}^k g_j(x)=1$ and $g_j(x)\ge0$, hence,
\[
h(x)\le \Big(\sum_{j=1}^k g_j(x)\Big)^2=1.
\]

Substituting the random variable $X$ for $x$ and taking the expectation gives $A\le 1$, whence the numerator is at most $1-B$ and thus $\xi'\le 1$. Combining, $\xi'\in[0,1]$.

\medskip\noindent
\textbf{(ii) $\xi'=0$ if and only if independence.}
($\Rightarrow$) If $\xi'=0$, then
\[
A \;=\; B.
\]
Since $\psi$ is strictly convex, equality holds in Jensen’s inequality if and only if
\[
(g_1(X),\dots,g_k(X)) \;=\; (p_1,\dots,p_k)\quad \text{almost surely.}
\]
Equivalently, $P(Y=j\mid X)=P(Y=j)$ a.s.\ for all $j$, which is independence of $X$ and $Y$.

($\Leftarrow$) Conversely, if $X$ and $Y$ are independent then $g_j(X)\equiv p_j$, so the numerator is zero and $\xi'=0$.

\medskip\noindent
\textbf{(iii) $\xi'=1$ if and only if $Y=f(X)$ a.s.}
($\Rightarrow$) If $\xi'=1$, then
\[
A \;=\; 1.
\]
Because $h(x)\le1$ for each $x$, we must have $h(X)=1$ almost surely. This holds if and only if for almost every $x$ exactly one $g_j(x)$ equals $1$ and the rest are $0$. Define
\[
f(x)\;:=\;\arg\max_{1\le j\le k} g_j(x).
\]
Then $P(Y=f(X)\mid X=x)=1$ for almost every $x$, i.e.\ $P(Y=f(X))=1$.

($\Leftarrow$) Conversely, if $Y=f(X)$ a.s., then for a.e.\ $x$ there is a unique $j$ such that $j=f(x)$ with $g_j(x)=1$ and $g_\ell(x)=0$ for $\ell\neq j$. Hence $h(X)=1$, so the numerator equals the denominator and $\xi'=1$.

\medskip\noindent
\textbf{(iv) Invariance under strictly monotone transformations.}
Let $\varphi:\mathbb{R}\to\mathbb{R}$ be a strictly monotone function. We show that $\xi'(\varphi(X),Y)=\xi'(X,Y)$.
Note that the definition of $\xi'(X,Y)$ depends on $X$ solely through the term $\mathbb{E}[\sum_{j=1}^k P(Y=j\mid X)^2]$.
Since $\varphi$ is strictly monotone, it is Borel measurable and injective, implying that the $\sigma$-algebras generated by $X$ and $\varphi(X)$ coincide, i.e., $\sigma(X)=\sigma(\varphi(X))$.
By the property of conditional expectation,
\[
\begin{aligned}
P(Y=j\mid \varphi(X)) &\;=\; \mathbb{E}[\mathbf{1}\{Y=j\}\mid \sigma(\varphi(X))] \\
&\;=\; \mathbb{E}[\mathbf{1}\{Y=j\}\mid \sigma(X)] \;=\; P(Y=j\mid X) \quad \text{a.s.}
\end{aligned}
\]
Squaring and taking expectations yields
\[
\mathbb{E}\Big[\sum_{j=1}^k P(Y=j\mid \varphi(X))^2\Big] \;=\; \mathbb{E}\Big[\sum_{j=1}^k P(Y=j\mid X)^2\Big].
\]
Thus, the value of $\xi'$ remains unchanged.

\medskip\noindent
\textbf{(v) Invariance under permutations of the category labels.}
Let $\pi:\{1,\ldots,k\}\to\{1,\ldots,k\}$ be a bijection and set $\widetilde Y=\pi(Y)$. Then
\[
P(\widetilde Y=j\mid X)=g_{\pi^{-1}(j)}(X),
\qquad
P(\widetilde Y=j)=p_{\pi^{-1}(j)}.
\]
Therefore, both sums defining $A$ and $B$ are unchanged after relabeling, and hence $\xi'(X,\widetilde Y)=\xi'(X,Y)$. At the sample level, relabeling preserves every indicator $\mathbf 1\{Y_{(i+1)}=Y_{(i)}\}$ and merely permutes the sample proportions $\hat p_j$. Thus, $A_n$ and $B_n$ are unchanged, so $\xi_n'(X,\widetilde Y)=\xi_n'(X,Y)$.

\medskip\noindent
\textbf{(vi) Binary case: $\xi'=\xi$ (population).}
Assume $Y\in\{0,1\}$, and let $p:=P(Y=1)$ and $r(X):=P(Y=1\mid X)$.
Recall the definition of Chatterjee's coefficient:
\[
\xi(X,Y) = \frac{\int \mathrm{Var}(E[\mathbf{1}\{Y \ge t\} \mid X]) \, dF_Y(t)}{\int \mathrm{Var}(\mathbf{1}\{Y \ge t\}) \, dF_Y(t)}.
\]
Since $Y$ takes values in $\{0,1\}$, the integral with respect to $dF_Y(t)$ is a sum over the support points weighted by their probabilities.
However, for $t=0$, the indicator $\mathbf{1}\{Y \ge 0\}$ is almost surely $1$, so its variance is $0$.
Thus, the only non-zero contribution comes from $t=1$, where $\mathbf{1}\{Y \ge 1\} = \mathbf{1}\{Y=1\}$ and the weight is $P(Y=1)=p$.
Canceling the common factor $p$ from the numerator and denominator, we obtain:
\begin{align*}
\xi(X,Y)
&= \frac{p \cdot \mathrm{Var}(E[\mathbf{1}\{Y=1\}\mid X])}{p \cdot \mathrm{Var}(\mathbf{1}\{Y=1\})} \\
&= \frac{\mathrm{Var}(r(X))}{p(1-p)} 
= \frac{\mathbb{E}[r(X)^2]-p^2}{p(1-p)}.
\end{align*}

For our proposed coefficient $\xi'$, with $k=2$, the numerator is the unweighted sum of variances for $j=0$ and $j=1$. Note that $P(Y=0|X) = 1-r(X)$.
\begin{align*}
\xi'(X,Y)
&= \frac{\mathrm{Var}(P(Y=1|X)) + \mathrm{Var}(P(Y=0|X))}{\mathrm{Var}(\mathbf{1}\{Y=1\}) + \mathrm{Var}(\mathbf{1}\{Y=0\})} \\
&= \frac{\mathrm{Var}(r(X)) + \mathrm{Var}(1-r(X))}{p(1-p) + p(1-p)} \\
&= \frac{2\mathrm{Var}(r(X))}{2p(1-p)} 
= \frac{\mathbb{E}[r(X)^2]-p^2}{p(1-p)}.
\end{align*}
Hence $\xi'=\xi$ at the population level.
\end{proof}

\subsubsection{Derivations for the R-squared and Gini-impurity representations}
\label{subsec:r2-gini-proof}

\begin{proof}
We first treat the binary--binary case and then the Gini impurity representation.

\medskip\noindent
\textbf{(a) Binary $X$ and binary $Y$: $\xi'$ equals $R^2$.}
Assume $X\in\{0,1\}$ and $Y\in\{0,1\}$.  Let
\[
p := P(Y=1),
\qquad
p_x := P(Y=1\mid X=x),
\qquad
\pi_x := P(X=x),
\quad x\in\{0,1\}.
\]
Write $r(X) := P(Y=1\mid X)$, so that $r(X)\in\{p_0,p_1\}$ and
\[
\mathbb{E}[r(X)] = \sum_{x=0}^1 \pi_x p_x = P(Y=1) = p.
\]

For binary $Y$, we have
\[
\sum_{j=0}^1 P(Y=j\mid X)^2
= P(Y=1\mid X)^2 + P(Y=0\mid X)^2
= r(X)^2 + \{1-r(X)\}^2.
\]
Taking expectations over $X$,
\[
\mathbb{E}\!\left[\sum_{j=0}^1 P(Y=j\mid X)^2\right]
= \mathbb{E}\big[r(X)^2 + \{1-r(X)\}^2\big].
\]
Similarly, at the marginal level,
\[
\sum_{j=0}^1 P(Y=j)^2
= p^2 + (1-p)^2.
\]
Hence the numerator of $\xi'(X,Y)$ is
\begin{align*}
&\mathbb{E}\!\left[\sum_{j=0}^1 P(Y=j\mid X)^2\right]
- \sum_{j=0}^1 P(Y=j)^2 \\
&= \mathbb{E}\big[r(X)^2 + \{1-r(X)\}^2\big]
   - \big\{p^2 + (1-p)^2\big\} \\
&= 2\,\mathbb{E}[r(X)^2] - 2p^2.
\end{align*}
The denominator is
\[
1 - \sum_{j=0}^1 P(Y=j)^2
= 1 - \{p^2 + (1-p)^2\}
= 2p(1-p)
= 2\,\mathrm{Var}(Y),
\]
since $Y$ is Bernoulli$(p)$.  Consequently,
\begin{align*}
\xi'(X,Y)
&= \frac{2\,\mathbb{E}[r(X)^2] - 2p^2}{2p(1-p)}
 = \frac{\mathbb{E}[r(X)^2] - p^2}{p(1-p)}.
\end{align*}
On the other hand,
\[
\mathrm{Var}_X\big(\mathbb{E}[Y\mid X]\big)
= \mathrm{Var}_X(r(X))
= \mathbb{E}[r(X)^2] - \{\mathbb{E}[r(X)]\}^2
= \mathbb{E}[r(X)^2] - p^2,
\]
and \(\mathrm{Var}(Y) = p(1-p).\)
Therefore
\[
\xi'(X,Y)
= \frac{\mathrm{Var}_X(\mathbb{E}[Y\mid X])}{\mathrm{Var}(Y)},
\]
which is exactly the population $R^2$ from the linear regression of $Y$ on~$X$.

\medskip\noindent
\textbf{(b) Gini impurity representation.}
Now let $Y$ take values in $\{1,\dots,k\}$ (with $k\ge 2$). Recall the Gini impurities
\[
G(Y):=1-B,
\qquad
G(Y\mid X=x):=1-h(x).
\]
Then $A=1-E[G(Y\mid X)]$ and $B=1-G(Y)$. Hence the numerator of $\xi'(X,Y)$ is
\begin{align*}
A-B
&=\big\{1-E[G(Y\mid X)]\big\}-\{1-G(Y)\}\\
&=G(Y)-E[G(Y\mid X)].
\end{align*}
The denominator is
\[
1-B=G(Y).
\]
Therefore
\[
\xi'(X,Y)
= \frac{G(Y) - \mathbb{E}[G(Y\mid X)]}{G(Y)}
= \frac{G(Y) - E_X[G(Y\mid X)]}{G(Y)},
\]
which is the claimed Gini impurity reduction representation.  This identity holds 
for general $X$.
\end{proof}

\subsection{Consistency and inference under independence}

\subsubsection{Proof of Theorem~\ref{thm:consistency-xin}}
\label{app:proof_consistency}
\begin{lemma}
\label{lem:uniform-adjacent-increments}
Let \(U_1,U_2,\ldots\) be i.i.d. \(U(0,1)\) random variables, and let
\(r:(0,1)\to\mathbb R\) be a bounded Borel measurable function. Then
\[
    \frac1{n-1}
    \sum_{i=1}^{n-1}
    \left|r(U_{(i+1)})-r(U_{(i)})\right|
    \xrightarrow{a.s.}0,
\]
where \(U_{(1)}<\cdots<U_{(n)}\) are the order statistics of
\(U_1,\ldots,U_n\).
\end{lemma}

\begin{proof}
Fix \(\varepsilon>0\) and \(\eta>0\). Let $\lambda$ be the Lebesgue measure on $\mathbb{R}$ and \(\|r\|_\infty\le M\).

By Lusin's theorem \citep[Theorem~7.10]{Folland1999}, there exists a compact set \(C\subset(0,1)\) such that
\[
    \lambda((0,1) \setminus C)<\eta
\]
and \(r\) is continuous on $C$. Since \(C\) is compact, \(r\) is uniformly
continuous on $C$. Hence there exists \(\delta>0\), for $u,v \in C$, such that
\[
    |u-v|<\delta
    \quad\Longrightarrow\quad
    |r(u)-r(v)|<\varepsilon.
\]

By the classical spacing result for uniform order statistics \citep{Devroye1982}, we have
\[
    \max_{1\le i\le n-1}(U_{(i+1)}-U_{(i)})
    \xrightarrow{a.s.}0.
\]
Hence, almost surely, for all sufficiently large \(n\), whenever
\(U_{(i)},U_{(i+1)}\in C\), we have
\[
    |r(U_{(i+1)})-r(U_{(i)})|<\varepsilon.
\]
Thus, for all sufficiently large \(n\),
\[
\begin{aligned}
    \frac1{n-1}
    \sum_{i=1}^{n-1}
    |r(U_{(i+1)})-r(U_{(i)})|
    &\le
    \varepsilon
    +
    \frac{2M}{n-1}
    \#\{i:U_{(i)}\notin C\ \text{or}\ U_{(i+1)}\notin C\} \\
    &\le
    \varepsilon
    +
    \frac{4M}{n-1}
    \sum_{\ell=1}^n \mathbf 1\{U_\ell\notin C\}.
\end{aligned}
\]
By the strong law of large numbers,
\[
    \frac1n\sum_{\ell=1}^n \mathbf 1\{U_\ell\notin C\}
    \xrightarrow{a.s.}
    \lambda((0,1) \setminus C).
\]
Therefore,
\[
    \limsup_{n\to\infty}
    \frac1{n-1}
    \sum_{i=1}^{n-1}
    |r(U_{(i+1)})-r(U_{(i)})|
    \le
    \varepsilon+4M\eta
    \quad a.s.
\]
Since \(\varepsilon>0\) and \(\eta>0\) are arbitrary, the desired result follows.
\end{proof}

\begin{proof}[Proof of Theorem~\ref{thm:consistency-xin}]
Recall the notation $p_j$, $g_j$, $h$, $A$, $B$, $A_n$, $\hat p_j$, and $B_n$ from Section~\ref{sec:xiprime}. Then $\xi'=(A-B)/(1-B)$ and, on the event $\{B_n<1\}$, $\xi_n'=(A_n-B_n)/(1-B_n)$. Since \(Y\) is not almost surely constant, \(B<1.\)

The map
\[
    f(a,b):=\frac{a-b}{1-b}
\]
is continuous at \((A,B)\). Therefore, by the continuous mapping theorem, once we prove
\[
    A_n\xrightarrow{a.s.}A, \qquad B_n \xrightarrow{a.s.}B,
\]
we obtain
\[
    \xi_n'
    =
    f(A_n,B_n)
    \xrightarrow{a.s.}
    f(A,B)
    =
    \xi'.
\] 

By the strong law of large numbers,
\[
    \hat p_j\xrightarrow{a.s.}p_j,
    \qquad j=1,\ldots,k.
\]
Since \(k\) is fixed, it follows that
\[
    B_n
    =
    \sum_{j=1}^k \hat p_j^2
    \xrightarrow{a.s.}
    \sum_{j=1}^k p_j^2
    =
    B.
\]
Thus, it remains to prove that
\[
    A_n\xrightarrow{a.s.}A.
\]

Let \(F_X\) be the distribution function of \(X\). For \(x\in\mathbb R\), write
\[
    F_X(x-):=\lim_{t\uparrow x}F_X(t),
    \qquad
    \Delta F_X(x):=F_X(x)-F_X(x-).
\]
Let \(V_1,V_2,\ldots\) be i.i.d. \(U(0,1)\) random variables, independent of
\((X_i,Y_i)_{i\ge 1}\), and define
\[
    U_i
    :=
    F_X(X_i-)+V_i\Delta F_X(X_i),
    \qquad i\ge 1.
\]
Then \(U_1,U_2,\ldots\) are i.i.d and $U_i \sim U(0,1)$. Moreover, if
\[
    Q(u):=\inf\{x\in\mathbb R:F_X(x)\ge u\},
    \qquad 0<u<1,
\]
then \(X_i=Q(U_i)\) almost surely.
The ordering induced by \(U_i\) is the same as the ordering induced by \(X_i\),
except that ties in \(X_i\) are broken by the independent uniforms \(V_i\).
Therefore, sorting by \(U_i\) exactly represents sorting by \(X_i\) with
independent random tie-breaking.
Let \(\pi_n\) be the random permutation satisfying
\[
    U_{\pi_n(1)}<\cdots<U_{\pi_n(n)}.
\]
Since \(U_1,\ldots,U_n\) have a continuous joint distribution, ties among the
\(U_i\)'s occur with probability zero. Write
\[
    U_{(i)}:=U_{\pi_n(i)},
    \qquad
    Y_{(i)}:=Y_{\pi_n(i)}.
\]

Define
\[
    \tilde g_j(u):=g_j(Q(u)),
    \qquad 0<u<1,
\]
and
\[
    \tilde h(u):=\sum_{j=1}^k \tilde g_j(u)^2.
\]
Since \(Q\) is monotone and \(g_j\) is Borel measurable, \(\tilde g_j\) is Borel
measurable. Also, since \(X_i=Q(U_i)\) almost surely,
\[
    g_j(X_i)=\tilde g_j(U_i)
    \quad a.s.,
\]
and therefore
\[
    h(X_i)=\tilde h(U_i)
    \quad a.s.
\]

Let
\(
    \mathcal U_n:=\sigma(U_1,\ldots,U_n).
\)
Conditionally on \(\mathcal U_n\), the random variables
\(Y_{(1)},\ldots,Y_{(n)}\) are independent, and
\[
    P(Y_{(i)}=j\mid \mathcal U_n)
    =
    \tilde g_j(U_{(i)}).
\]
Therefore,
\[
\begin{aligned}
    P(Y_{(i)}=Y_{(i+1)}\mid \mathcal U_n)
    &=
    \sum_{j=1}^k
    P(Y_{(i)}=j\mid \mathcal U_n)
    P(Y_{(i+1)}=j\mid \mathcal U_n) \\
    &=
    \sum_{j=1}^k
    \tilde g_j(U_{(i)})
    \tilde g_j(U_{(i+1)}).
\end{aligned}
\]
Set
\(
    q_{i,n}
    :=
    \sum_{j=1}^k
    \tilde g_j(U_{(i)})
    \tilde g_j(U_{(i+1)}).
\)
Then
\(
    E[\mathbf 1\{Y_{(i)}=Y_{(i+1)}\}\mid \mathcal U_n]
    =
    q_{i,n}.
\)
Hence
\[
\begin{aligned}
    A_n
    &=
    \frac1{n-1}\sum_{i=1}^{n-1}
    \mathbf 1\{Y_{(i)}=Y_{(i+1)}\} \\
    &=
    \frac1{n-1}\sum_{i=1}^{n-1}q_{i,n}
    +
    \frac1{n-1}\sum_{i=1}^{n-1}D_{i,n},
\end{aligned}
\]
where
\(
    D_{i,n}
    :=
    \mathbf 1\{Y_{(i)}=Y_{(i+1)}\}-q_{i,n}.
\)

We first show that
\[
    \frac1{n-1}\sum_{i=1}^{n-1}D_{i,n}
    \xrightarrow{a.s.}0.
\]
For every \(i\),
\(
    E[D_{i,n}\mid \mathcal U_n]=0,
\) and \(|D_{i,n}|\le 1.\)
Let
\[
    \mathcal I_o:=\{i\in\{1,\ldots,n-1\}: i\ \text{is odd}\},
    \qquad
    \mathcal I_e:=\{i\in\{1,\ldots,n-1\}: i\ \text{is even}\}.
\]
Conditionally on \(\mathcal U_n\), the random variables
\(\{D_{i,n}:i\in\mathcal I_o\}\) are independent, because they depend on
disjoint pairs of labels. Similarly, \(\{D_{i,n}:i\in\mathcal I_e\}\) are
conditionally independent.
Define
\[
    S_n^o:=\frac1{n-1}\sum_{i\in\mathcal I_o}D_{i,n},
    \qquad
    S_n^e:=\frac1{n-1}\sum_{i\in\mathcal I_e}D_{i,n}.
\]
Then
\[
    \frac1{n-1}\sum_{i=1}^{n-1}D_{i,n}
    =
    S_n^o+S_n^e.
\]
By Hoeffding's inequality, for every \(\varepsilon>0\), there exists a constant \(c>0\) such that
\[
    P(|S_n^o|>\varepsilon\mid \mathcal U_n)
    \le
    2\exp(-c n\varepsilon^2),
\]
and
\[
    P(|S_n^e|>\varepsilon\mid \mathcal U_n)
    \le
    2\exp(-c n\varepsilon^2).
\]
Since the right-hand sides of the conditional bounds are deterministic and do not depend on \(\mathcal U_n\), taking expectations gives the corresponding unconditional bounds. Therefore, for some constant \(c'>0\),
\[
\begin{aligned}
    P\left(
        \left|
        \frac1{n-1}\sum_{i=1}^{n-1}D_{i,n}
        \right|
        >\varepsilon
    \right)
    &\le
    P(|S_n^o|>\varepsilon/2)
    +
    P(|S_n^e|>\varepsilon/2) \\
    &\le
    4\exp(-c'n\varepsilon^2).
\end{aligned}
\]
Since
\(
    \sum_{n=1}^\infty 4\exp(-c'n\varepsilon^2)<\infty,
\)
the Borel--Cantelli lemma gives
\[
    \frac1{n-1}\sum_{i=1}^{n-1}D_{i,n}
    \xrightarrow{a.s.}0.
\]

It remains to prove that
\[
    \frac1{n-1}\sum_{i=1}^{n-1}q_{i,n}
    \xrightarrow{a.s.}
    A.
\]
For each \(i\),
\[
\begin{aligned}
    q_{i,n}-\tilde h(U_{(i)})
    &=
    \sum_{j=1}^k
    \tilde g_j(U_{(i)})
    \tilde g_j(U_{(i+1)})
    -
    \sum_{j=1}^k
    \tilde g_j(U_{(i)})^2 \\
    &=
    \sum_{j=1}^k
    \tilde g_j(U_{(i)})
    \left\{
        \tilde g_j(U_{(i+1)})
        -
        \tilde g_j(U_{(i)})
    \right\}.
\end{aligned}
\]
Since \(0\le \tilde g_j\le 1\),
\[
\begin{aligned}
    |q_{i,n}-\tilde h(U_{(i)})|
    &\le
    \sum_{j=1}^k
    \left|
        \tilde g_j(U_{(i+1)})
        -
        \tilde g_j(U_{(i)})
    \right|.
\end{aligned}
\]
Therefore,
\[
\begin{aligned}
    \frac1{n-1}\sum_{i=1}^{n-1}
    |q_{i,n}-\tilde h(U_{(i)})|
    &\le
    \sum_{j=1}^k
    \frac1{n-1}\sum_{i=1}^{n-1}
    \left|
        \tilde g_j(U_{(i+1)})
        -
        \tilde g_j(U_{(i)})
    \right|.
\end{aligned}
\]
By Lemma~\ref{lem:uniform-adjacent-increments}, applied to
\(r=\tilde g_j\), we have
\[
    \frac1{n-1}\sum_{i=1}^{n-1}
    \left|
        \tilde g_j(U_{(i+1)})
        -
        \tilde g_j(U_{(i)})
    \right|
    \xrightarrow{a.s.}0
\]
for each \(j=1,\ldots,k\). Since \(k\) is fixed,
\[
    \frac1{n-1}\sum_{i=1}^{n-1}
    |q_{i,n}-\tilde h(U_{(i)})|
    \xrightarrow{a.s.}0.
\]
Hence
\[
    \frac1{n-1}\sum_{i=1}^{n-1}q_{i,n}
    -
    \frac1{n-1}\sum_{i=1}^{n-1}\tilde h(U_{(i)})
    \xrightarrow{a.s.}0.
\]

Now we identify the limit of \(\frac1{n-1}\sum_{i=1}^{n-1}\tilde h(U_{(i)}).\)
Since \(0\le \tilde h\le 1\),
\[
\begin{aligned}
    \left|
    \frac1{n-1}\sum_{i=1}^{n-1}\tilde h(U_{(i)})
    -
    \frac1n\sum_{i=1}^n \tilde h(U_i)
    \right|
    &\le
    \frac{2}{n-1}.
\end{aligned}
\]
By the strong law of large numbers,
\[
    \frac1n\sum_{i=1}^n \tilde h(U_i)
    \xrightarrow{a.s.}
    E[\tilde h(U)],
\]
where \(U\sim U(0,1)\). Since \(Q(U)\stackrel{d}=X\),
\[
\begin{aligned}
    E[\tilde h(U)]
    &=
    E\left[
        \sum_{j=1}^k g_j(Q(U))^2
    \right] \\
    &=
    E\left[
        \sum_{j=1}^k g_j(X)^2
    \right] \\
    &=
    E[h(X)] \\
    &=
    A.
\end{aligned}
\]
Therefore,
\[
    \frac1{n-1}\sum_{i=1}^{n-1}\tilde h(U_{(i)})
    \xrightarrow{a.s.}
    A.
\]

Combining the previous results yields
\[
    \frac1{n-1}\sum_{i=1}^{n-1}q_{i,n}
    \xrightarrow{a.s.}
    A.
\]
Together with \(\frac1{n-1}\sum_{i=1}^{n-1}D_{i,n} \xrightarrow{a.s.}0,\)
we conclude that
\[
    A_n\xrightarrow{a.s.}A.
\]

This completes the proof.
\end{proof}

\subsubsection{Proof of Theorem~\ref{thm:asymptotic-xin}}
\label{app:proof_asymptotic_normality}

\begin{proof}[Proof of Theorem~\ref{thm:asymptotic-xin}]
Recall that
\[
    A_n
    :=
    \frac1{n-1}
    \sum_{i=1}^{n-1}
    \mathbf 1\{Y_{(i+1)}=Y_{(i)}\},
    \qquad
    B_n
    :=
    \sum_{j=1}^k
    \left(
        \frac1n
        \sum_{i=1}^n
        \mathbf 1\{Y_i=j\}
    \right)^2,
\]
and
\[
    \xi_n'
    =
    \frac{A_n-B_n}{1-B_n}.
\]
Since $Y$ is not almost surely constant, $B<1$.

Let $\pi_n$ be the tie-broken random permutation such that
\[
    X_{\pi_n(1)}<\cdots<X_{\pi_n(n)}.
\]
Then
\[
    Y_{(i)}=Y_{\pi_n(i)},
    \qquad i=1,\ldots,n.
\]
Since $X$ and $Y$ are independent and the tie breaking randomization is independent of the data, $\pi_n$ is independent of $(Y_1,\ldots,Y_n)$. Therefore,
\[
    (Y_{(1)},\ldots,Y_{(n)})
    \overset{d}{=}
    (Y_1,\ldots,Y_n).
\]
Thus, in the rest of the proof, we may regard $Y_{(1)},\ldots,Y_{(n)}$ as an i.i.d.\ categorical sequence with probabilities $p_1,\ldots,p_k$.

For notational convenience, define
\[
    Z_i:=Y_{(i)},
    \qquad
    W_i:=\mathbf 1\{Z_{i+1}=Z_i\},
    \qquad
    i=1,\ldots,n-1.
\]
Also define
\[
    p_{Z_i}
    :=
    \sum_{j=1}^k p_j\mathbf 1\{Z_i=j\}.
\]
Then
\[
    E[W_i]=B,
    \qquad
    E[p_{Z_i}]=B.
\]

Since
\(
    A_n
    =
    \frac1{n-1}\sum_{i=1}^{n-1}W_i,
\)
we have
\[
    A_n-B
    =
    \frac1{n-1}
    \sum_{i=1}^{n-1}(W_i-B).
\]
The variables $W_i-B$ are bounded and 1-dependent. Hence, by Chebyshev's inequality, 
\[
    \sum_{i=1}^{n-1}(W_i-B)=O_p(\sqrt n),
\]
and consequently
\[
\begin{aligned}
    \sqrt n(A_n-B)
    &=
    \frac{\sqrt n}{n-1}
    \sum_{i=1}^{n-1}(W_i-B)
    \\
    &=
    \frac1{\sqrt n}
    \sum_{i=1}^{n-1}(W_i-B)
    +o_p(1).
\end{aligned}
\]

Since the ordering only permutes the labels,
\[
    \hat p_j
    =
    \frac1n\sum_{i=1}^n\mathbf 1\{Z_i=j\},
\]
where $\hat p_j$ is the empirical class proportion already used in the definition of $B_n$. We decompose
\[
\begin{aligned}
    B_n-B
    &=
    \sum_{j=1}^k(\hat p_j^2-p_j^2)
    \\
    &=
    2\sum_{j=1}^k p_j(\hat p_j-p_j)
    +
    \sum_{j=1}^k(\hat p_j-p_j)^2.
\end{aligned}
\]
Because $k$ is fixed and $\hat p_j-p_j=O_p(n^{-1/2})$ for each $j$ from the central limit theorem,
\[
    \sum_{j=1}^k(\hat p_j-p_j)^2
    =
    O_p(n^{-1}).
\]
Thus
\[
\begin{aligned}
    \sqrt n(B_n-B)
    &=
    \frac2{\sqrt n}
    \sum_{i=1}^n(p_{Z_i}-B)
    +o_p(1)
    \\
    &=
    \frac2{\sqrt n}
    \sum_{i=1}^{n-1}(p_{Z_i}-B)
    +o_p(1),
\end{aligned}
\]
where the last equality follows because $p_{Z_n}-B$ is bounded.

Therefore,
\[
    \sqrt n
    \begin{pmatrix}
        A_n-B\\
        B_n-B
    \end{pmatrix}
    =
    \frac1{\sqrt n}
    \sum_{i=1}^{n-1}
    S_i
    +o_p(1),
\]
where
\[
    S_i
    :=
    \begin{pmatrix}
        W_i-B\\
        2(p_{Z_i}-B)
    \end{pmatrix},
    \qquad
    i=1,\ldots,n-1.
\]
The sequence $(S_i)_{i\ge1}$ is bounded, and 1-dependent.

By the multivariate central limit theorem for 1-dependent sequences \citep{HoeffdingRobbins1948},
\[
    \frac1{\sqrt n}
    \sum_{i=1}^{n-1}S_i
    \xrightarrow{d}
    N_2(0,\Gamma),
\]
where
\[
    \Gamma
    =
    \operatorname{Var}(S_1)
    +
    \operatorname{Cov}(S_1,S_2)
    +
    \operatorname{Cov}(S_2,S_1).
\]
We now compute $\Gamma$ explicitly. First,
\[
    E[W_1]
    =
    P(Z_1=Z_2)
    =
    \sum_{j=1}^k p_j^2
    =
    B.
\]
Moreover,
\[
    E[W_1W_2]
    =
    P(Z_1=Z_2=Z_3)
    =
    \sum_{j=1}^k p_j^3
    =
    \rho.
\]
Therefore,
\[
    \operatorname{Var}(W_1)
    =
    B-B^2,
    \qquad
    \operatorname{Cov}(W_1,W_2)
    =
    \rho-B^2.
\]
It follows that
\[
    \Gamma_{11}
    =
    B-B^2+2(\rho-B^2)
    =
    B+2\rho-3B^2.
\]

Next,
\[
    E[p_{Z_1}^2]
    =
    \sum_{j=1}^k p_j^3
    =
    \rho.
\]
Since $Z_1$ and $Z_2$ are independent,
\[
    \operatorname{Cov}(p_{Z_1},p_{Z_2})=0.
\]
Thus
\[
    \Gamma_{22}
    =
    4\operatorname{Var}(p_{Z_1})
    =
    4(\rho-B^2).
\]
Finally,
\[
\begin{aligned}
    E[W_1p_{Z_1}]
    &=
    E[\mathbf 1\{Z_1=Z_2\}p_{Z_1}]
    =
    \sum_{j=1}^k p_j^3
    =
    \rho,
    \\
    E[W_1p_{Z_2}]
    &=
    E[\mathbf 1\{Z_1=Z_2\}p_{Z_2}]
    =
    \sum_{j=1}^k p_j^3
    =
    \rho.
\end{aligned}
\]
Also, $W_2$ is independent of $Z_1$, and hence
\[
    \operatorname{Cov}(W_2,p_{Z_1})=0.
\]
Therefore,
\[
\begin{aligned}
    \Gamma_{12}
    &=
    \operatorname{Cov}(W_1,2p_{Z_1})
    +
    \operatorname{Cov}(W_1,2p_{Z_2})
    +
    \operatorname{Cov}(W_2,2p_{Z_1})
    \\
    &=
    2(\rho-B^2)+2(\rho-B^2)+0
    \\
    &=
    4(\rho-B^2).
\end{aligned}
\]
Similarly, $\Gamma_{21}=4(\rho-B^2)$.

Define
\[
    \sigma_B^2
    :=
    4(\rho-B^2),
    \qquad
    \sigma^2
    :=
    B(1+B)-2\rho.
\]
Then
\[
    \Gamma
    =
    \begin{pmatrix}
        \sigma^2+\sigma_B^2 & \sigma_B^2\\
        \sigma_B^2 & \sigma_B^2
    \end{pmatrix}.
\]
Indeed,
\[
    \sigma^2+\sigma_B^2
    =
    B(1+B)-2\rho+4(\rho-B^2)
    =
    B+2\rho-3B^2
    =
    \Gamma_{11}.
\]
Consequently,
\[
    \sqrt n
    \begin{pmatrix}
        A_n-B\\
        B_n-B
    \end{pmatrix}
    \xrightarrow{d}
    N_2
    \left(
        0,
        \begin{pmatrix}
            \sigma^2+\sigma_B^2 & \sigma_B^2\\
            \sigma_B^2 & \sigma_B^2
        \end{pmatrix}
    \right).
\]

Define
\[
    f(a,b)
    :=
    \frac{a-b}{1-b}.
\]
Then $\xi_n'=f(A_n,B_n)$. The gradient is
\[
    \nabla f(a,b)
    =
    \begin{pmatrix}
        \dfrac1{1-b}\\[1em]
        \dfrac{a-1}{(1-b)^2}
    \end{pmatrix}.
\]
Under independence, the corresponding population value of both $A_n$ and $B_n$ is $B$, and hence the gradient is evaluated at $(B,B)$:
\[
    \nabla f(B,B)
    =
    \begin{pmatrix}
        \dfrac1{1-B}\\[1em]
        -\dfrac1{1-B}
    \end{pmatrix}.
\]
By the delta method,
\[
    \sqrt n\,\xi_n'
    \xrightarrow{d}
    N(0,\kappa^2),
\]
where
\[
\begin{aligned}
    \kappa^2
    &=
    \nabla f(B,B)^T
    \Gamma
    \nabla f(B,B)
    \\
    &=
    \frac1{(1-B)^2}
    \left\{
        (\sigma^2+\sigma_B^2)
        +
        \sigma_B^2
        -
        2\sigma_B^2
    \right\}
    \\
    &=
    \frac{\sigma^2}{(1-B)^2}.
\end{aligned}
\]
A direct expansion gives
\[
\sigma^2
=
\sum_{a=1}^k\sum_{b=1}^k
p_ap_b
\left(
\mathbf 1\{a=b\}-p_a-p_b+B
\right)^2.
\]
Thus, $\sigma^2\ge0$. If $\sigma^2=0$, then, for every $j$ with $p_j>0$, the term corresponding to $a=b=j$ implies
\[
1-2p_j+B=0,
\qquad\text{or equivalently}\qquad
p_j=\frac{1+B}{2}>\frac12.
\]
Since $Y$ is not almost surely constant, at least two class probabilities are positive, which is impossible. Therefore, $\sigma^2>0$, and since $B<1$, we have $0<\kappa^2<\infty$.

The estimator $\hat\kappa^2$ is obtained by replacing $p_j$ with $\hat p_j$ in the expression for $\kappa^2$. Since
\[
    \hat p_j\xrightarrow{a.s.}p_j,
    \qquad j=1,\ldots,k,
\]
by the strong law of large numbers, and since $Y$ is not almost surely constant, $1-B>0$. Therefore the map
\[
    (p_1,\ldots,p_k)
    \mapsto
    \frac{
        \sum_{j=1}^k p_j^2(1+\sum_{\ell=1}^k p_\ell^2)
        -
        2\sum_{j=1}^k p_j^3
    }{
        (1-\sum_{j=1}^k p_j^2)^2
    }
\]
is continuous at $(p_1,\ldots,p_k)$. Hence, by the continuous mapping theorem,
\[
    \hat\kappa^2\xrightarrow{p}\kappa^2.
\]
Finally, Slutsky's theorem gives
\[
    Z_n
    =
    \frac{\sqrt n\,\xi_n'}{\hat\kappa}
    =
    \frac{\sqrt n\,\xi_n'}{\kappa}
    \cdot
    \frac{\kappa}{\hat\kappa}
    \xrightarrow{d}
    N(0,1).
\]
\end{proof}
\subsubsection{Proof of Corollaries~\ref{cor:type1_error} and~\ref{cor:power_consistency}}
\label{subsec:proof-corollaries}
\begin{proof}[Proof of Corollary~\ref{cor:type1_error}]
Assume the null hypothesis $H_0: X \perp Y$ holds.
Let $\Phi(\cdot)$ denote the cumulative distribution function (CDF) of the standard normal distribution, and let $Z \sim N(0,1)$.

According to Theorem~\ref{thm:asymptotic-xin} and $Z_n^{\mathrm{c}}-Z_n=o_p(1)$, the centered statistic converges in distribution to the standard normal distribution $Z$ under $H_0$:
\[
Z_n^{\mathrm{c}} \xrightarrow{d} Z \sim N(0,1).
\]
By the definition of convergence in distribution, this implies that
\[
\lim_{n\to\infty} P_{H_0}(Z_n^{\mathrm{c}} \le x) = \Phi(x)
\]
for every point $x \in \mathbb{R}$. Since the standard normal CDF $\Phi$ is continuous everywhere, this convergence holds specifically at $x = z_{1-\alpha}$.
For a finite sample size $n$, the size of the test is:
\[
P_{H_0}(\phi_n = 1) = P_{H_0}(Z_n^{\mathrm{c}} > z_{1-\alpha}) = 1 - P_{H_0}(Z_n^{\mathrm{c}} \le z_{1-\alpha}).
\]
Taking the limit as $n \to \infty$, we obtain the asymptotic size:
\[
\begin{aligned}
\lim_{n\to\infty} P_{H_0}(\phi_n = 1)
&= \lim_{n\to\infty} \left[ 1 - P_{H_0}(Z_n^{\mathrm{c}} \le z_{1-\alpha}) \right] \\
&= 1 - \Phi(z_{1-\alpha}) \\
&= 1 - (1 - \alpha) \\
&= \alpha.
\end{aligned}
\]
Thus, the test controls the type I error asymptotically at level $\alpha$.
\end{proof}
\medskip\noindent
\begin{proof}[Proof of Corollary~\ref{cor:power_consistency}]
Assume $\xi' > 0$.
First, consider the asymptotic behavior of the ratio $\{\xi_n'+1/(n-1)\}/\hat\kappa$.
By Theorem~\ref{thm:consistency-xin}, the numerator converges in probability to the population value $\xi'$ and the denominator converges in probability to the population value $\kappa$.
Since $Y$ is not almost surely constant, the positivity argument in the proof of Theorem~\ref{thm:asymptotic-xin} gives $0<\kappa<\infty$.
By Slutsky's theorem, the ratio converges to a positive constant:
\[
\frac{\xi_n'+1/(n-1)}{\hat\kappa} \xrightarrow{p} \frac{\xi'}{\kappa} := c > 0.
\]

Next, consider the power for a fixed significance level $\alpha$. The rejection condition is $Z_n^{\mathrm{c}} > z_{1-\alpha}$. We can rewrite this inequality as:
\[
\sqrt{n} \frac{\xi_n'+1/(n-1)}{\hat\kappa} > z_{1-\alpha} \quad \iff \quad \frac{\xi_n'+1/(n-1)}{\hat\kappa} > \frac{z_{1-\alpha}}{\sqrt{n}}.
\]
Thus, the power can be written as:
\[
P_{H_1}(\phi_n = 1) = P_{H_1}\left( \frac{\xi_n'+1/(n-1)}{\hat\kappa} > \frac{z_{1-\alpha}}{\sqrt{n}} \right).
\]
Let $W_n := \frac{\xi_n'+1/(n-1)}{\hat\kappa} - \frac{z_{1-\alpha}}{\sqrt{n}}$.
Since $\frac{z_{1-\alpha}}{\sqrt{n}} \to 0$ as $n \to \infty$, it follows that $W_n \xrightarrow{p} c$.
Because $c > 0$, for any sufficiently small $0<\epsilon<c$, the event $\{ |W_n - c| < \epsilon \}$ implies $\{ W_n > 0 \}$.
By the definition of convergence in probability, $\lim_{n\to\infty} P(|W_n - c| < \epsilon) = 1$.
Therefore,
\[
\lim_{n\to\infty} P_{H_1}(W_n > 0) \ge \lim_{n\to\infty} P(|W_n - c| < \epsilon) = 1.
\]
This concludes that the power converges to 1.
\end{proof}

\subsection{Inference under general dependence}

\subsubsection{Proof of Theorem~\ref{thm:general_clt}}
\label{app:proof_general_clt}
\begin{lemma}
\label{lem:ordered-bv-bound}
Suppose Assumption~\ref{asm:general_clt} holds. Then, for each fixed integer $r\ge1$,
\[
    \max_{1\le j\le k}
    \sum_{i=1}^{n-r}
    |g_j(X_{(i+r)})-g_j(X_{(i)})|
    =o_p(\sqrt n).
\]
\end{lemma}

\begin{proof}
By Assumption~\ref{asm:general_clt},
\[
    P\left(\text{there exists }i\le n\text{ such that }X_i\notin I_n\right)
    \le
    \sum_{i=1}^nP(X_i\notin I_n)
    =
    nP(X\notin I_n)
    \to0.
\]
Hence, with probability tending to one, all observations lie in $I_n$. On this event, for each fixed $r\ge1$,
\[
\begin{aligned}
    \sum_{i=1}^{n-r}|g_j(X_{(i+r)})-g_j(X_{(i)})|
    &\le
    r\sum_{i=1}^{n-1}|g_j(X_{(i+1)})-g_j(X_{(i)})|\\
    &\le
    rV_{I_n}(g_j).
\end{aligned}
\]
Taking the maximum over $j$ gives
\[
    \max_{1\le j\le k}
    \sum_{i=1}^{n-r}|g_j(X_{(i+r)})-g_j(X_{(i)})|
    \le
    r\max_{1\le j\le k}V_{I_n}(g_j)
    =o(\sqrt n)
\]
on an event whose probability tends to one. Hence the bound becomes \(o_p(\sqrt n).\)
\end{proof}

\begin{proof}[Proof of Theorem~\ref{thm:general_clt}]
Let
\(
    \mathcal X_n:=\sigma(X_1,\ldots,X_n).
\)
Since ties in \(X\) are broken uniformly at random independently of the labels, and since labels with the same \(X\)-value have the same conditional distribution, the ordered labels remain conditionally independent given \(\mathcal X_n\). Moreover,
\[
    P(Y_{(i)}=j\mid\mathcal X_n)
    =
    g_j(X_{(i)}),
    \qquad
    i=1,\ldots,n,\quad j=1,\ldots,k.
\]
For notational convenience, write
\[
    g_{ij}:=g_j(X_{(i)}),
    \qquad
    h_i:=h(X_{(i)}),
    \qquad
    m_i:=m(X_{(i)}),
    \qquad
    \tau_i:=\tau(X_{(i)}).
\]

Define
\[
    W_i:=\mathbf 1\{Y_{(i+1)}=Y_{(i)}\},
    \qquad
    i=1,\ldots,n-1.
\]
Then
\[
\begin{aligned}
    \mu_i
    &:={E}[W_i\mid\mathcal X_n]\\
    &=P(Y_{(i+1)}=Y_{(i)}\mid\mathcal X_n)\\
    &=\sum_{j=1}^k
    P(Y_{(i)}=j,Y_{(i+1)}=j\mid\mathcal X_n)\\
    &=\sum_{j=1}^k g_{ij}g_{i+1,j}.
\end{aligned}
\]
Let
\[
    M_n:={E}[A_n\mid\mathcal X_n]
    =
    \frac1{n-1}\sum_{i=1}^{n-1}\mu_i.
\]
We decompose
\[
\begin{aligned}
    A_n-A
    &=
    (A_n-M_n)
    +
    \left\{
        \frac1n\sum_{i=1}^n h(X_i)-A
    \right\}+
    \left\{
        M_n-\frac1n\sum_{i=1}^n h(X_i)
    \right\}.
\end{aligned}
\]

We first show that the last term in the above decomposition is negligible on the $\sqrt n$ scale.
\[
\begin{aligned}
    \mu_i-h_i
    &=
    \sum_{j=1}^k g_{ij}g_{i+1,j}
    -
    \sum_{j=1}^k g_{ij}^2\\
    &=
    \sum_{j=1}^k g_{ij}(g_{i+1,j}-g_{ij}),
\end{aligned}
\]
and $0\le g_{ij}\le1$,
\[
    |\mu_i-h_i|
    \le
    \sum_{j=1}^k |g_{i+1,j}-g_{ij}|.
\]
Therefore, by Lemma~\ref{lem:ordered-bv-bound} with $r=1$,
\[
\begin{aligned}
    \left|
        \frac1{n-1}\sum_{i=1}^{n-1}\mu_i
        -
        \frac1{n-1}\sum_{i=1}^{n-1}h_i
    \right|
    &\le
    \frac1{n-1}
    \sum_{j=1}^k\sum_{i=1}^{n-1}|g_{i+1,j}-g_{ij}|\\
    &=
    o_p(n^{-1/2}).
\end{aligned}
\]
Also, since $0\le h_i\le1$,
\[
\begin{aligned}
    \left|
        \frac1{n-1}\sum_{i=1}^{n-1}h_i
        -
        \frac1n\sum_{i=1}^n h_i
    \right|
    &=
    \left|
        \frac1{n(n-1)}\sum_{i=1}^{n-1}h_i
        -
        \frac1n h_n
    \right|\\
    &\le
    \frac2n
    =O(n^{-1}).
\end{aligned}
\]
Hence
\[
    M_n-\frac1n\sum_{i=1}^n h(X_i)
    =o_p(n^{-1/2}),
\]
and so
\[
    \sqrt n
    \left\{
        M_n-\frac1n\sum_{i=1}^n h(X_i)
    \right\}
    =o_p(1).
\]

Next, we analyze the first term of the decomposition
\[
    A_n-M_n
    =
    \frac1{n-1}\sum_{i=1}^{n-1}(W_i-\mu_i).
\]
Conditional on $\mathcal X_n$, the variables $W_i-\mu_i$ and $W_\ell-\mu_\ell$ are independent whenever $|i-\ell|>1$. Also $|W_i-\mu_i|\le1$. Therefore,
\[
\begin{aligned}
    \operatorname{Var}\left(
        \sum_{i=1}^{n-1}(W_i-\mu_i)
        \middle|\mathcal X_n
    \right)
    &\le
    \sum_{i=1}^{n-1}\operatorname{Var}(W_i\mid\mathcal X_n)\\
    &\quad+
    2\sum_{i=1}^{n-2}
    \left|
        \operatorname{Cov}(W_i,W_{i+1}\mid\mathcal X_n)
    \right|\\
    &\le
    (n-1)+2(n-2)\\
    &\le 3n.
\end{aligned}
\]
By Chebyshev's inequality, since the right-hand side is deterministic,
\[
    \sum_{i=1}^{n-1}(W_i-\mu_i)=O_p(\sqrt n).
\]
Consequently,
\[
\begin{aligned}
    \sqrt n(A_n-M_n)
    &=
    \frac{\sqrt n}{n-1}
    \sum_{i=1}^{n-1}(W_i-\mu_i)\\
    &=
    \frac1{\sqrt n}
    \sum_{i=1}^{n-1}(W_i-\mu_i)
    +o_p(1).
\end{aligned}
\]
Thus
\[
    \sqrt n(A_n-A)
    =
    \frac1{\sqrt n}\sum_{i=1}^{n-1}(W_i-\mu_i)
    +
    \frac1{\sqrt n}\sum_{i=1}^n\{h(X_i)-A\}
    +o_p(1).
\]

We now decompose $B_n-B$. Since
\[
    B_n-B=\sum_{j=1}^k(\hat p_j^2-p_j^2),
\]
we have
\[
\begin{aligned}
    B_n-B
    &=
    \sum_{j=1}^k
    \left\{
        2p_j(\hat p_j-p_j)+(\hat p_j-p_j)^2
    \right\}.
\end{aligned}
\]
By the multivariate central limit theorem,
\[
    \hat p_j-p_j=O_p(n^{-1/2}),
    \qquad j=1,\ldots,k.
\]
Since $k$ is fixed,
\[
    \sqrt n\sum_{j=1}^k(\hat p_j-p_j)^2
    =O_p(n^{-1/2})=o_p(1).
\]
Furthermore, writing $p_{Y_i} := \sum_{j=1}^k p_j \mathbf{1}\{Y_i =j \}$, we have
\[
\begin{aligned}
    2\sum_{j=1}^k p_j(\hat p_j-p_j)
    &=
    \frac2n\sum_{i=1}^n
    \left\{
        \sum_{j=1}^k p_j\mathbf 1\{Y_i=j\}
        -
        \sum_{j=1}^k p_j^2
    \right\}\\
    &=
    \frac2n\sum_{i=1}^n(p_{Y_i}-B).
\end{aligned}
\]
Therefore
\[
    \sqrt n(B_n-B)
    =
    \frac2{\sqrt n}\sum_{i=1}^n(p_{Y_i}-B)+o_p(1).
\]
Since
\[
    E[p_Y\mid X]
    =
    \sum_{j=1}^k p_jP(Y=j\mid X)
    =
    m(X),
\]
we write
\[
    p_{Y_i}-B
    =
    \{m(X_i)-B\}+\{p_{Y_i}-m(X_i)\}.
\]
So, with the equation $\sum_{i=1}^n m(X_i) = \sum_{i=1}^n m(X_{(i)})$, we have
\[
\begin{aligned}
    \sqrt n(B_n-B)
    &=
    \frac2{\sqrt n}\sum_{i=1}^n\{m(X_i)-B\}\\
    &\quad+
    \frac2{\sqrt n}\sum_{i=1}^n\{p_{Y_{(i)}}-m(X_{(i)})\}
    +o_p(1).
\end{aligned}
\]

Define
\[
    D_n
    :=
    \frac1{\sqrt n}
    \sum_{i=1}^n
    \begin{pmatrix}
        h(X_i)-A\\
        2\{m(X_i)-B\}
    \end{pmatrix},
\]
which is the $\mathcal X_n$-measurable component. Also define
\[
    L_n
    :=
    \frac1{\sqrt n}
    \begin{pmatrix}
        \displaystyle\sum_{i=1}^{n-1}(W_i-\mu_i)\\[1.2em]
        \displaystyle 2\sum_{i=1}^n\{p_{Y_{(i)}}-m(X_{(i)})\}
    \end{pmatrix},
\]
which is the label--dependent component. Then
\[
    \sqrt n
    \begin{pmatrix}
        A_n-A\\
        B_n-B
    \end{pmatrix}
    =
    D_n+L_n+o_p(1).
\]

We first handle $D_n$. Let
\[
    Z_i
    :=
    \begin{pmatrix}
        h(X_i)-A\\
        2\{m(X_i)-B\}
    \end{pmatrix}.
\]
Then $Z_1,\ldots,Z_n$ are i.i.d. mean--zero random vectors. Since $0\le h(X)\le1$ and $0\le m(X)\le1$, each component of $Z_i$ is bounded. Hence, by the multivariate central limit theorem ~\citep{HoeffdingRobbins1948},
\[
    D_n
    =
    \frac1{\sqrt n}\sum_{i=1}^n Z_i
    \xrightarrow{d}
    N_2(0,\Sigma_D),
\]
where
\[
    \Sigma_D
    =
    \operatorname{Var}
    \begin{pmatrix}
        h(X)\\
        2m(X)
    \end{pmatrix}.
\]
Equivalently,
\[
    \Sigma_D
    =
    \begin{pmatrix}
        \operatorname{Var}(h(X))
        &
        2E[\{h(X)-A\}\{m(X)-B\}]\\[0.4em]
        2E[\{h(X)-A\}\{m(X)-B\}]
        &
        4\operatorname{Var}(m(X))
    \end{pmatrix}.
\]

We next compute the conditional covariance matrix of $L_n$. Write
\[
    L_n=
    \begin{pmatrix}
        L_{n,1}\\
        L_{n,2}
    \end{pmatrix},
\]
where
\[
    L_{n,1}:=\frac1{\sqrt n}\sum_{i=1}^{n-1}(W_i-\mu_i),
    \qquad
    L_{n,2}:=\frac2{\sqrt n}\sum_{i=1}^n(p_{Y_{(i)}}-m_i).
\]

By construction,
\[
    E[L_n\mid\mathcal X_n]=0.
\]
Conditional on $\mathcal{X}_n$, $W_i, W_\ell$ are independent whenever $|i-\ell| > 1$. Hence,
\[
\begin{aligned}
    \operatorname{Var}(L_{n,1}\mid\mathcal X_n)
    &=
    \frac1n\sum_{i=1}^{n-1}\operatorname{Var}(W_i\mid\mathcal X_n)\\
    &\quad+
    \frac2n\sum_{i=1}^{n-2}\operatorname{Cov}(W_i,W_{i+1}\mid\mathcal X_n).
\end{aligned}
\]
Since $W_i\mid\mathcal X_n\sim\operatorname{Bernoulli}(\mu_i)$,
\[
    \operatorname{Var}(W_i\mid\mathcal X_n)=\mu_i(1-\mu_i).
\]
The map $x\mapsto x(1-x)$ is 1--Lipschitz on $[0,1]$. Therefore,
\[
    |\mu_i(1-\mu_i)-h_i(1-h_i)|
    \le
    |\mu_i-h_i|.
\]
Hence, by Lemma~\ref{lem:ordered-bv-bound},
\[
\begin{aligned}
    \frac1n\sum_{i=1}^{n-1}
    |\mu_i(1-\mu_i)-h_i(1-h_i)|
    &\le
    \frac1n\sum_{j=1}^k\sum_{i=1}^{n-1}|g_{i+1,j}-g_{ij}|\\
    &=o_p(1).
\end{aligned}
\]
Since $h(X_i)(1-h(X_i))$ is bounded, the strong law of large numbers gives
\[
    \frac1n\sum_{i=1}^n h(X_i)(1-h(X_i))
    \xrightarrow{a.s.}
    E[h(X)(1-h(X))].
\]
Also,
\[
    \frac1n\sum_{i=1}^{n-1}h_i(1-h_i)
    =
    \frac1n\sum_{i=1}^n h(X_i)(1-h(X_i))+O(n^{-1}).
\]
Therefore,
\[
    \frac1n\sum_{i=1}^{n-1}\operatorname{Var}(W_i\mid\mathcal X_n)
    \xrightarrow{p}
    E[h(X)(1-h(X))].
\]

Next,
\[
\begin{aligned}
    E[W_iW_{i+1}\mid\mathcal X_n]
    &=
    P(Y_{(i)}=Y_{(i+1)}=Y_{(i+2)}\mid\mathcal X_n)\\
    &=
    \sum_{j=1}^k g_{ij}g_{i+1,j}g_{i+2,j}.
\end{aligned}
\]
Thus
\[
    \operatorname{Cov}(W_i,W_{i+1}\mid\mathcal X_n)
    =
    \sum_{j=1}^k g_{ij}g_{i+1,j}g_{i+2,j}
    -
    \mu_i\mu_{i+1}.
\]
For $0\le a,b,c\le1$,
\[
\begin{aligned}
    |abc-a^3|
    &\le
    |abc-a^2c|+|a^2c-a^3|\\
    &=
    |ac||b-a|+a^2|c-a|\\
    &\le
    |b-a|+|c-a|.
\end{aligned}
\]
Therefore, by Lemma~\ref{lem:ordered-bv-bound} with $r=1,2$,
\[
\begin{aligned}
    &\frac1n\sum_{i=1}^{n-2}
    \left|
        \sum_{j=1}^k g_{ij}g_{i+1,j}g_{i+2,j}
        -
        \sum_{j=1}^k g_{ij}^3
    \right|\\
    &\quad\le
    \frac1n\sum_{j=1}^k\sum_{i=1}^{n-2}
    \{|g_{i+1,j}-g_{ij}|+|g_{i+2,j}-g_{ij}|\}\\
    &=o_p(1).
\end{aligned}
\]
Since $\tau(X_i)$ is bounded,
\[
    \frac1n\sum_{i=1}^n\tau(X_i)
    \xrightarrow{a.s.}
    E[\tau(X)],
\]
and hence
\[
    \frac1n\sum_{i=1}^{n-2}\sum_{j=1}^k g_{ij}g_{i+1,j}g_{i+2,j}
    \xrightarrow{p}
    E[\tau(X)].
\]

We also compare $\mu_i\mu_{i+1}$ with $h_i^2$. Since $0\le\mu_i,h_i\le1$,
\[
\begin{aligned}
    |\mu_i\mu_{i+1}-h_i^2|
    &\le
    |\mu_i-h_i|+|\mu_{i+1}-h_i|\\
    &\le
    |\mu_i-h_i|+|\mu_{i+1}-h_{i+1}|+|h_{i+1}-h_i|.
\end{aligned}
\]
Moreover,
\[
\begin{aligned}
    |h_{i+1}-h_i|
    &=
    \left|
        \sum_{j=1}^k g_{i+1,j}^2
        -
        \sum_{j=1}^k g_{ij}^2
    \right|\\
    &\le
    \sum_{j=1}^k |g_{i+1,j}^2-g_{ij}^2|\\
    &\le
    2\sum_{j=1}^k |g_{i+1,j}-g_{ij}|.
\end{aligned}
\]
Applying Lemma~\ref{lem:ordered-bv-bound} gives
\[
    \frac1n\sum_{i=1}^{n-2}|\mu_i\mu_{i+1}-h_i^2|
    =o_p(1).
\]
Since $h(X_i)^2$ is bounded,
\[
    \frac1n\sum_{i=1}^n h(X_i)^2
    \xrightarrow{a.s.}
    E[h(X)^2],
\]
and so
\[
    \frac1n\sum_{i=1}^{n-2}\mu_i\mu_{i+1}
    \xrightarrow{p}
    E[h(X)^2].
\]
Combining the preceding displays,
\[
\begin{aligned}
    \operatorname{Var}(L_{n,1}\mid\mathcal X_n)
    &\xrightarrow{p}
    E[h(X)(1-h(X))]
    +2\{E[\tau(X)]-E[h(X)^2]\}\\
    &=
    E[h(X)-3h(X)^2+2\tau(X)].
\end{aligned}
\]

Then, consider the second term $L_{n,2}$. Since $Y_{(1)},\ldots,Y_{(n)}$ are conditionally independent given $\mathcal X_n$,
\[
\begin{aligned}
    \operatorname{Var}(L_{n,2}\mid\mathcal X_n)
    &=
    \frac4n\sum_{i=1}^n
    \operatorname{Var}(p_{Y_{(i)}}\mid\mathcal X_n)\\
    &=
    \frac4n\sum_{i=1}^n
    \left\{
        \sum_{j=1}^k p_j^2g_{ij}-m_i^2
    \right\}.
\end{aligned}
\]
The summand is a bounded function of $X_{(i)}$, and ordering only permutes the sample. Hence the strong law of large numbers gives
\[
    \operatorname{Var}(L_{n,2}\mid\mathcal X_n)
    \xrightarrow{a.s.}
    4E\left[
        \sum_{j=1}^k p_j^2g_j(X)-m(X)^2
    \right].
\]

We next compute the conditional covariance between $L_{n,1}$ and $L_{n,2}$. Conditional on $\mathcal{X}_n$, $W_i-\mu_i$ depends only on $Y_{(i)}$ and $Y_{(i+1)}$, while $p_{Y_{(\ell)}}-m_\ell$ depends only on $Y_{(\ell)}$. Hence the covariance is zero unless $\ell=i$ or $\ell=i+1$. Thus
\[
\begin{aligned}
    \operatorname{Cov}(L_{n,1},L_{n,2}\mid\mathcal X_n)
    &=
    \frac1n\sum_{i=1}^{n-1}
    \operatorname{Cov}
    \left(
        W_i-\mu_i,
        2(p_{Y_{(i)}}-m_i)
        \middle|\mathcal X_n
    \right)\\
    &\quad+
    \frac1n\sum_{i=1}^{n-1}
    \operatorname{Cov}
    \left(
        W_i-\mu_i,
        2(p_{Y_{(i+1)}}-m_{i+1})
        \middle|\mathcal X_n
    \right).
\end{aligned}
\]
For the first covariance,
\[
\begin{aligned}
    \operatorname{Cov}
    \left(
        W_i-\mu_i,
        2(p_{Y_{(i)}}-m_i)
        \middle|\mathcal X_n
    \right)&=
    2E[(W_i-\mu_i)(p_{Y_{(i)}}-m_i)\mid\mathcal X_n]\\
    &=2\left\{
        E[W_ip_{Y_{(i)}}\mid\mathcal X_n]
        -
        \mu_i m_i
    \right\}\\
    &=
    2\left\{
        \sum_{j=1}^k p_jg_{ij}g_{i+1,j}
        -
        \mu_i m_i
    \right\}.
\end{aligned}
\]
Similarly,
\[
\begin{aligned}
    \operatorname{Cov}
    \left(
        W_i-\mu_i,
        2(p_{Y_{(i+1)}}-m_{i+1})
        \middle|\mathcal X_n
    \right)
    =
    2\left\{
        \sum_{j=1}^k p_jg_{ij}g_{i+1,j}
        -
        \mu_i m_{i+1}
    \right\}.
\end{aligned}
\]
Therefore
\[
\begin{aligned}
    \operatorname{Cov}(L_{n,1},L_{n,2}\mid\mathcal X_n)
    &=
    \frac1n\sum_{i=1}^{n-1}
    \left[
        4\sum_{j=1}^k p_jg_{ij}g_{i+1,j}
        -
        2\mu_i(m_i+m_{i+1})
    \right].
\end{aligned}
\]
We compare the summand with
\[
    4\left\{
        \sum_{j=1}^k p_jg_{ij}^2-h_im_i
    \right\}.
\]
First,
\[
\begin{aligned}
    \left|
        \sum_{j=1}^k p_jg_{ij}g_{i+1,j}
        -
        \sum_{j=1}^k p_jg_{ij}^2
    \right|
    &\le
    \sum_{j=1}^k p_jg_{ij}|g_{i+1,j}-g_{ij}|\\
    &\le
    \sum_{j=1}^k |g_{i+1,j}-g_{ij}|.
\end{aligned}
\]
Second,
\[
\begin{aligned}
    |\mu_i(m_i+m_{i+1})-2h_im_i|
    &\le
    |(\mu_i-h_i)m_i|
    +
    |\mu_im_{i+1}-h_im_i|\\
    &\le
    |\mu_i-h_i|
    +
    |(\mu_i-h_i)m_{i+1}|
    +
    |h_i(m_{i+1}-m_i)|\\
    &\le
    2|\mu_i-h_i|+|m_{i+1}-m_i|.
\end{aligned}
\]
Moreover,
\(
    |\mu_i-h_i|
    \le
    \sum_{j=1}^k |g_{i+1,j}-g_{ij}|,
\)
and
\(
    |m_{i+1}-m_i|
    \le
    \sum_{j=1}^k p_j|g_{i+1,j}-g_{ij}|
    \le
    \sum_{j=1}^k |g_{i+1,j}-g_{ij}|.
\)
Hence
\[
\begin{aligned}
    &\left|
    \left[
        4\sum_{j=1}^k p_jg_{ij}g_{i+1,j}
        -
        2\mu_i(m_i+m_{i+1})
    \right]
    -
    4\left[
        \sum_{j=1}^k p_jg_{ij}^2-h_im_i
    \right]
    \right|\\
    &\qquad\le
    4\sum_{j=1}^k |g_{i+1,j}-g_{ij}|
    +
    2|\mu_i(m_i+m_{i+1})-2h_im_i|\\
    &\qquad\le
    4\sum_{j=1}^k |g_{i+1,j}-g_{ij}|
    +
    2\left\{
        2|\mu_i-h_i|+|m_{i+1}-m_i|
    \right\}\\
    &\qquad\le
    10\sum_{j=1}^k |g_{i+1,j}-g_{ij}|.
\end{aligned}
\]
By Lemma~\ref{lem:ordered-bv-bound} with \(r=1\),
\[
\begin{aligned}
    \frac1n
    \sum_{i=1}^{n-1}
    \sum_{j=1}^k |g_{i+1,j}-g_{ij}|
    &\le
    \frac{k}{n}
    \max_{1\le j\le k}
    \sum_{i=1}^{n-1}
    |g_j(X_{(i+1)})-g_j(X_{(i)})|\\
    &=
    \frac{k}{n}o_p(\sqrt n)\\
    &=
    o_p(1),
\end{aligned}
\]
since \(k\) is fixed. Therefore,
\[
\begin{aligned}
    &\frac1n
    \sum_{i=1}^{n-1}
    \left|
    \left[
        4\sum_{j=1}^k p_jg_{ij}g_{i+1,j}
        -
        2\mu_i(m_i+m_{i+1})
    \right]
    -
    4\left[
        \sum_{j=1}^k p_jg_{ij}^2-h_im_i
    \right]
    \right|
    =
    o_p(1).
\end{aligned}
\]
Since \(
    \sum_{j=1}^k p_jg_j(X)^2-h(X)m(X)
\)
is bounded, the strong law of large numbers gives
\[
\begin{aligned}
    \frac1n
    \sum_{i=1}^{n-1}
    4\left[
        \sum_{j=1}^k p_jg_{ij}^2-h_im_i
    \right]
    &=
    \frac4n
    \sum_{i=1}^{n}
    \left[
        \sum_{j=1}^k p_jg_j(X_i)^2-h(X_i)m(X_i)
    \right]
    +
    O(n^{-1})\\
    &\xrightarrow{a.s.}
    4E\left[
        \sum_{j=1}^k p_jg_j(X)^2-h(X)m(X)
    \right].
\end{aligned}
\]
Combining the preceding two displays, we obtain
\[
    \operatorname{Cov}(L_{n,1},L_{n,2}\mid\mathcal X_n)
    \xrightarrow{p}
    4E\left[
        \sum_{j=1}^k p_jg_j(X)^2-h(X)m(X)
    \right].
\]
Therefore,
\[
    \operatorname{Var}(L_n\mid\mathcal X_n)
    \xrightarrow{p}
    \Sigma_R,
\]
where
\[
    \Sigma_R
    =
    \begin{pmatrix}
        v_{11} & v_{12}\\
        v_{12} & v_{22}
    \end{pmatrix},
\]
with
\[
    v_{11}=E[h(X)-3h(X)^2+2\tau(X)],
\]
\[
    v_{22}=4E\left[
        \sum_{j=1}^k p_j^2g_j(X)-m(X)^2
    \right],
\]
and
\[
    v_{12}=4E\left[
        \sum_{j=1}^k p_jg_j(X)^2-h(X)m(X)
    \right].
\]

We now establish the conditional central limit theorem for $L_n$. Fix $t=(t_1,t_2)^\top\in\mathbb R^2$ and write
\[
    t^\top L_n
    =
    \sum_{i=1}^{n-1}E_i+
    \sum_{i=1}^n F_i,
\]
where
\[
    E_i:=\frac{t_1}{\sqrt n}(W_i-\mu_i),
    \qquad i=1,\ldots,n-1,
\]
and
\[
    F_i:=\frac{2t_2}{\sqrt n}(p_{Y_{(i)}}-m_i),
    \qquad i=1,\ldots,n.
\]
Conditional on $\mathcal X_n$, these variables are mean-zero. Construct a graph on the vertices
\[
    \mathcal{V}_n := \{E_1,\ldots,E_{n-1},F_1,\ldots,F_n\}
\]
by connecting two vertices whenever the corresponding random variables share at least one label $Y_{(i)}$. This is a dependency graph conditional on $\mathcal X_n$. Its maximum degree is bounded by $4$ since $E_i$ can interact only with $E_{i-1},E_{i+1},F_i,F_{i+1}$, while $F_i$ can interact only with $E_{i-1},E_i$. Moreover,
\[
    |E_i|\le\frac{|t_1|}{\sqrt n},
    \qquad
    |F_i|\le\frac{2|t_2|}{\sqrt n}.
\]
Let
\[
    s_n^2(t):=
    \operatorname{Var}(t^\top L_n\mid\mathcal X_n).
\]
Since
\[
    \operatorname{Var}(L_n\mid\mathcal X_n)
    \xrightarrow{p}
    \Sigma_R,
\]
the continuous mapping theorem gives
\[
    s_n^2(t)
    =
    t^\top\operatorname{Var}(L_n\mid\mathcal X_n)t
    \xrightarrow{p}
    t^\top\Sigma_Rt.
\]

Suppose that \(t^\top\Sigma_Rt>0\). Since
\(
    s_n^2(t)
    \xrightarrow{p}
    t^\top\Sigma_Rt,
\)
there exists \(\eta>0\) such that
\[
    P(s_n(t)>\eta)\to1.
\]
Fix such an \(\eta>0\). On the event \(\{s_n(t)>\eta\}\), the Berry--Esseen-type bound for sums with a dependency graph due to \citet{BaldiPierreYosef1989} gives a constant \(C_t<\infty\), independent of \(n\), such that
\[
\begin{aligned}
    \sup_{x\in\mathbb R}
    \left|
        P\left(
            \frac{t^\top L_n}{s_n(t)}\le x
            \,\middle|\,\mathcal X_n
        \right)
        -
        \Phi(x)
    \right|
    &\le
    C_t
    \left\{
        (2n-1)4^2
        \left(\frac{n^{-1/2}}{s_n(t)}\right)^3
    \right\}^{1/2}\\
    &\le
    C_t n^{-1/4}.
\end{aligned}
\]
Hence
\[
    \sup_{x\in\mathbb R}
    \left|
        P\left(
            \frac{t^\top L_n}{s_n(t)}\le x
            \,\middle|\,\mathcal X_n
        \right)
        -
        \Phi(x)
    \right|
    \xrightarrow{p}0.
\]
Let
\(
    v_t:=t^\top\Sigma_Rt>0.
\)
Define
\[
    G_n:=\{s_n(t)>\eta\}.
\]
On \(G_n\), set
\[
    \Delta_{n,0}(t)
    :=
    \sup_{y\in\mathbb R}
    \left|
        P\left(
            \frac{t^\top L_n}{s_n(t)}\le y
            \,\middle|\,\mathcal X_n
        \right)
        -
        \Phi(y)
    \right|.
\]
By the preceding Berry--Esseen bound,
\[
    \Delta_{n,0}(t)\le C_t n^{-1/4}
    \qquad\text{on }G_n.
\]
For \(v>0\), define
\[
    \Phi_v(x):=\Phi\left(\frac{x}{\sqrt v}\right),
    \qquad x\in\mathbb R.
\]
Let
\[
    \Delta_n(t)
    :=
    \sup_{x\in\mathbb R}
    \left|
        P(t^\top L_n\le x\mid\mathcal X_n)
        -
        \Phi_{v_t}(x)
    \right|.
\]
Then
\[
\begin{aligned}
    \Delta_n(t)
    &=
    \mathbf 1_{G_n}\Delta_n(t)
    +
    \mathbf 1_{G_n^c}\Delta_n(t)\\
    &\le
    \mathbf 1_{G_n}\Delta_n(t)
    +
    \mathbf 1_{G_n^c},
\end{aligned}
\]
because \(\Delta_n(t)\le 1\). On the event \(G_n\), we have \(s_n(t)>0\). Hence, for every \(x\in\mathbb R\),
\[
\begin{aligned}
    \left|
        P(t^\top L_n\le x\mid\mathcal X_n)
        -
        \Phi\left(\frac{x}{\sqrt{v_t}}\right)
    \right|
    &=
    \left|
        P\left(
            \frac{t^\top L_n}{s_n(t)}
            \le
            \frac{x}{s_n(t)}
            \,\middle|\,\mathcal X_n
        \right)
        -
        \Phi\left(\frac{x}{\sqrt{v_t}}\right)
    \right|\\
    &\le
    \left|
        P\left(
            \frac{t^\top L_n}{s_n(t)}
            \le
            \frac{x}{s_n(t)}
            \,\middle|\,\mathcal X_n
        \right)
        -
        \Phi\left(\frac{x}{s_n(t)}\right)
    \right|\\
    &\quad+
    \left|
        \Phi\left(\frac{x}{s_n(t)}\right)
        -
        \Phi\left(\frac{x}{\sqrt{v_t}}\right)
    \right|.
\end{aligned}
\]
Taking the supremum over \(x\in\mathbb R\), and using the change of variable
\(y=x/s_n(t)\), gives
\[
\begin{aligned}
    \mathbf 1_{G_n}\Delta_n(t)
    &\le
    \mathbf 1_{G_n}
    \sup_{y\in\mathbb R}
    \left|
        P\left(
            \frac{t^\top L_n}{s_n(t)}
            \le y
            \,\middle|\,\mathcal X_n
        \right)
        -
        \Phi(y)
    \right|\\
    &\quad+
    \mathbf 1_{G_n}
    \sup_{x\in\mathbb R}
    \left|
        \Phi\left(\frac{x}{s_n(t)}\right)
        -
        \Phi\left(\frac{x}{\sqrt{v_t}}\right)
    \right|\\
    &=
    \mathbf 1_{G_n}\Delta_{n,0}(t)
    +
    \mathbf 1_{G_n}
    \sup_{x\in\mathbb R}
    \left|
        \Phi\left(\frac{x}{s_n(t)}\right)
        -
        \Phi\left(\frac{x}{\sqrt{v_t}}\right)
    \right|.
\end{aligned}
\]
Combining the preceding bounds, we obtain
\[
\begin{aligned}
    \Delta_n(t)
    &\le
    \mathbf 1_{G_n}\Delta_{n,0}(t)
    +
    \mathbf 1_{G_n}
    \sup_{x\in\mathbb R}
    \left|
        \Phi\left(\frac{x}{s_n(t)}\right)
        -
        \Phi\left(\frac{x}{\sqrt{v_t}}\right)
    \right|
    +
    \mathbf 1_{G_n^c}.
\end{aligned}
\]
Since
\(
    s_n^2(t)\xrightarrow{p}v_t>0
\)
and \(s_n(t)\ge0\), the continuous mapping theorem gives
\[
    s_n(t)\xrightarrow{p}\sqrt{v_t}.
\]
Then, we have
\[
    \sup_{x\in\mathbb R}
    \left|
        \Phi\left(\frac{x}{s_n(t)}\right)
        -
        \Phi\left(\frac{x}{\sqrt{v_t}}\right)
    \right|
    \xrightarrow{p}0.
\]
Moreover, \(P(G_n^c)\to0\) and \(\Delta_{n,0}(t)\le C_tn^{-1/4}\) on \(G_n\).
Therefore,
\[
    \Delta_n(t)\xrightarrow{p}0.
\]
Thus the conditional distribution of \(t^\top L_n\) given \(\mathcal X_n\) converges weakly in probability to \(N(0,v_t)\). Consequently, for every bounded continuous function \(\varphi:\mathbb R\to\mathbb C\),
\[
    E[\varphi(t^\top L_n)\mid\mathcal X_n]
    \xrightarrow{p}
    E[\varphi(Z_t)],
    \qquad
    Z_t\sim N(0,v_t).
\]
Taking \(\varphi(u)=e^{iu}\), we obtain
\[
    E[e^{it^\top L_n}\mid\mathcal X_n]
    \xrightarrow{p}
    E[e^{iZ_t}]
    =
    \exp\left(-\frac12v_t\right)
    =
    \exp\left(-\frac12t^\top\Sigma_Rt\right).
\]
Moreover,
\[
    \left|
        E[e^{it^\top L_n}\mid\mathcal X_n]
        -
        \exp\left(-\frac12t^\top\Sigma_Rt\right)
    \right|
    \le 2.
\]
Therefore the convergence also holds in \(L^1\).

If \(t^\top\Sigma_Rt=0\), then \(s_n^2(t)\xrightarrow{p}0\). Since \(s_n^2(t)\) is uniformly bounded,
\[
    E[s_n^2(t)]\to0.
\]
Also, \(E[t^\top L_n\mid\mathcal X_n]=0\), and hence
\[
    E[(t^\top L_n)^2]
    =
    E[s_n^2(t)]
    \to0.
\]
Thus \(t^\top L_n\xrightarrow{L^2}0\). Therefore,
\[
\begin{aligned}
    E\left|
        E[e^{it^\top L_n}\mid\mathcal X_n]-1
    \right|
    &\le
    E|e^{it^\top L_n}-1|\\
    &\le
    E|t^\top L_n|\\
    &\le
    \{E[(t^\top L_n)^2]\}^{1/2}
    \to0.
\end{aligned}
\]
Since \(t^\top\Sigma_Rt=0\), this gives
\[
    E[e^{it^\top L_n}\mid\mathcal X_n]
    \xrightarrow{L^1}
    \exp\left(-\frac12t^\top\Sigma_Rt\right).
\]

Combining the two cases, for every \(t\in\mathbb R^2\),
\[
    \psi_n(t)
    :=
    E[e^{it^\top L_n}\mid\mathcal X_n]
    \xrightarrow{L^1}
    \exp\left(-\frac12t^\top\Sigma_Rt\right).
\]

It remains to combine $D_n$ and $L_n$. Since $D_n$ is $\mathcal X_n$-measurable,
\[
\begin{aligned}
    E[e^{it^\top(D_n+L_n)}]
    &=
    E\left[
        e^{it^\top D_n}
        E[e^{it^\top L_n}\mid\mathcal X_n]
    \right]\\
    &=
    E[e^{it^\top D_n}\psi_n(t)]\\
    &=
    \exp\left(-\frac12t^\top\Sigma_Rt\right)
    E[e^{it^\top D_n}]
    +o(1).
\end{aligned}
\]
Because
\[
    D_n\xrightarrow{d}N_2(0,\Sigma_D),
\]
we have
\[
    E[e^{it^\top D_n}]
    \to
    \exp\left(-\frac12t^\top\Sigma_Dt\right).
\]
Therefore
\[
    E[e^{it^\top(D_n+L_n)}]
    \to
    \exp\left(-\frac12t^\top(\Sigma_D+\Sigma_R)t\right).
\]
Hence
\[
    D_n+L_n
    \xrightarrow{d}
    N_2(0,\Sigma_D+\Sigma_R).
\]
Since
\[
    \sqrt n
    \begin{pmatrix}
        A_n-A\\
        B_n-B
    \end{pmatrix}
    =D_n+L_n+o_p(1),
\]
Slutsky's theorem gives
\[
    \sqrt n
    \begin{pmatrix}
        A_n-A\\
        B_n-B
    \end{pmatrix}
    \xrightarrow{d}
    N_2(0,\Sigma),
\]
where
\[
    \Sigma=\Sigma_D+\Sigma_R.
\]

We now identify the entries of $\Sigma$. First,
\[
\begin{aligned}
    \Sigma_{11}
    &=
    E[h(X)-3h(X)^2+2\tau(X)]
    +
    \operatorname{Var}(h(X)).
\end{aligned}
\]
Second,
\[
\begin{aligned}
    \Sigma_{22}
    &=
    4E\left[
        \sum_{j=1}^k p_j^2g_j(X)-m(X)^2
    \right]
    +
    4\operatorname{Var}(m(X))\\
    &=
    4\left\{
        E\left[\sum_{j=1}^k p_j^2g_j(X)\right]
        -
        E[m(X)^2]
        +
        E[m(X)^2]
        -
        B^2
    \right\}\\
    &=
    4\left\{
        \sum_{j=1}^k p_j^2E[g_j(X)]-B^2
    \right\}\\
    &=
    4\left(\sum_{j=1}^k p_j^3-B^2\right).
\end{aligned}
\]
Third,
\[
\begin{aligned}
    \Sigma_{12}
    &=
    4E\left[
        \sum_{j=1}^k p_jg_j(X)^2-h(X)m(X)
    \right]\\
    &\quad+
    2E[\{h(X)-A\}\{m(X)-B\}].
\end{aligned}
\]
Thus the joint central limit theorem is proved.

Finally, since \(Y\) is not almost surely constant, \(B<1\). Also,
\[
    B_n
    =
    \sum_{j=1}^k\hat p_j^2
    \xrightarrow{a.s.}
    \sum_{j=1}^k p_j^2
    =
    B.
\]
Hence \(P(B_n<1)\to1\). On the event \(\{B_n<1\}\), we have
\[
    f(A_n,B_n)=\xi_n'.
\]
Therefore,
\[
    \sqrt n\{\xi_n'-f(A_n,B_n)\}
    \xrightarrow{p}0.
\]
Since \(f(a,b)=(a-b)/(1-b)\) is continuously differentiable in a neighborhood of \((A,B)\), the multivariate delta method~\cite[Theorem~3.1]{vanderVaart1998} gives
\[
    \sqrt n\{f(A_n,B_n)-f(A,B)\}
    \xrightarrow{d}
    N(0,\kappa_D^2),
\]
where
\[
    \kappa_D^2=\nabla f(A,B)^\top\Sigma\nabla f(A,B).
\]
Since \(f(A,B)=\xi'\), Slutsky's theorem yields
\[
    \sqrt n(\xi_n'-\xi')
    \xrightarrow{d}
    N(0,\kappa_D^2).
\]
\end{proof}

\subsubsection{Proof of Proposition~\ref{prop:variance_consistency}}
\label{app:proof_variance_consistency}

\begin{proof}
The asymptotic variance $\kappa_{D}^2$ is a continuous function of population moments of the form $E[h(X)]$, $E[\tau(X)]$, et cetera. Therefore, by the Continuous Mapping Theorem, it suffices to show that the sample averages of the estimated functions converge in probability to their population counterparts. We demonstrate this for the term $A = E[h(X)]$; the proof for other terms is analogous.

We aim to show $\frac{1}{n}\sum_{i=1}^n \hat{h}(X_i) \xrightarrow{p} E[h(X)]$.
Decompose the error as:
\[
\left| \frac{1}{n}\sum_{i=1}^n \hat{h}(X_i) - E[h(X)] \right| 
\le \underbrace{\left| \frac{1}{n}\sum_{i=1}^n (\hat{h}(X_i) - h(X_i)) \right|}_{(I)} 
+ \underbrace{\left| \frac{1}{n}\sum_{i=1}^n h(X_i) - E[h(X)] \right|}_{(II)}.
\]

\textbf{Term (II):} Since $h(X) = \sum_{j=1}^k g_j(X)^2$ is a bounded random variable (as $0 \le g_j \le 1$), the Weak Law of Large Numbers (WLLN) implies that $(II) \xrightarrow{p} 0$.

\textbf{Term (I):} Recall $\hat{h}(x) = \sum_{j=1}^k \hat{g}_j(x)^2$ and $h(x) = \sum_{j=1}^k g_j(x)^2$. Since $g_j, \hat{g}_j \in [0,1]$, we have $|\hat{h}(x) - h(x)| \le \sum_{j=1}^k |\hat{g}_j(x)^2 - g_j(x)^2| \le 2 \sum_{j=1}^k |\hat{g}_j(x) - g_j(x)|$.
Thus,
\[
(I) \le \frac{2}{n}\sum_{i=1}^n \sum_{j=1}^k |\hat{g}_j(X_i) - g_j(X_i)|.
\]
By the additional assumption in Proposition~\ref{prop:variance_consistency}, the nonparametric estimators satisfy the required empirical $L_1$ consistency: $n^{-1}\sum_{i=1}^n |\hat g_j(X_i)-g_j(X_i)|\xrightarrow{p}0$ for each $j=1,\ldots,k$.
This implies that the average absolute estimation error converges to zero in probability. Hence $(I) \xrightarrow{p} 0$.

Combining these results, $\frac{1}{n}\sum_{i=1}^n \hat{h}(X_i) \xrightarrow{p} A$. Similarly, sample estimates for $B$, $\sigma_A^2$, $\sigma_B^2$, and $\sigma_{AB}$ converge to their population values. Consequently, $\hat{\kappa}_{D}^2 \xrightarrow{p} \kappa_{D}^2$.
\end{proof}

\subsubsection{Proof of Corollary~\ref{cor:confidence_interval}}
\label{app:proof_confidence_interval}
\begin{proof}
By Theorem~\ref{thm:general_clt}, we have the asymptotic normality result:
\[
\sqrt{n}(\xi_n' - \xi') \xrightarrow{d} N(0, \kappa_D^2).
\]
Proposition~\ref{prop:variance_consistency} establishes that $\hat{\kappa}_D^2 \xrightarrow{p} \kappa_D^2$, which implies $\hat{\kappa}_D \xrightarrow{p} \kappa_D$ by the continuous mapping theorem. Since $\kappa_D^2>0$ by assumption, studentization is well defined asymptotically.
We consider the quantity:
\[
T_n := \frac{\sqrt{n}(\xi_n' - \xi')}{\hat{\kappa}_D} = \frac{\sqrt{n}(\xi_n' - \xi')}{\kappa_D} \cdot \frac{\kappa_D}{\hat{\kappa}_D}.
\]
The first component converges in distribution to $Z \sim N(0,1)$, and the second component converges in probability to 1. By Slutsky's theorem, $T_n \xrightarrow{d} N(0,1)$.
The coverage probability is given by
\[
\begin{aligned}
P\bigl(\xi' \in \mathcal{C}_n(1-\alpha)\bigr)
&= P\left( \xi_n' - z_{\alpha/2} \frac{\hat{\kappa}_D}{\sqrt{n}} \le \xi' \le \xi_n' + z_{\alpha/2} \frac{\hat{\kappa}_D}{\sqrt{n}} \right) \\
&= P\left( -z_{\alpha/2} \le \frac{\sqrt{n}(\xi_n' - \xi')}{\hat{\kappa}_D} \le z_{\alpha/2} \right).
\end{aligned}
\]
Taking the limit as $n \to \infty$,
\[
\lim_{n\to\infty} P\bigl(-z_{\alpha/2} \le T_n \le z_{\alpha/2}\bigr) = P(-z_{\alpha/2} \le Z \le z_{\alpha/2}) = 1-\alpha.
\]
This completes the proof.
\end{proof}

\section*{Acknowledgements}
The authors thank the two anonymous referees for their careful reading and constructive comments, which helped improve the presentation and theoretical development of the manuscript.

\section*{Funding}
This work was supported by the National Research Foundation of Korea (NRF) grant funded by the Korea government (MSIT) (No. RS-2026-25499129, RS-2023-00301976, RS-2026-25476070 and No. RS-2024-00333399).
\bibliographystyle{imsart-nameyear}  
\bibliography{library}       

\end{document}